\documentclass[12pt,a4paper,reqno]{amsart}
\usepackage{etex}
\usepackage{turnstile}
\usepackage{graphicx}
\usepackage{rotating}
\usepackage{amsmath}
\usepackage{graphicx}
\usepackage{enumitem}
\usepackage{amsmath,amsthm,amssymb,xcolor,bussproofs,mathtools}
\usepackage[T1]{fontenc}
\usepackage[applemac]{inputenc}
\usepackage{stmaryrd,eucal,soul}
\usepackage{lscape}
\usepackage{fullpage}
\usepackage{tikz}
\usetikzlibrary{positioning,arrows.meta}

\renewenvironment{abstract}{\section*{Abstract}\small}{}
\newtheorem{definition}{Definition}[section]
\newtheorem{proposition}{Proposition}[section]
\newtheorem{theorem}{Theorem}[section]

\newtheorem{example}{Example}[section]

\newtheorem{lemma}{Lemma}[section]

\newcommand{\fullG}{\overline{\overline{\mathsf{G4}}}}
\newcommand{\fullGpn}{\overline{\overline{\mathsf{G4^{pn}}}}}

\newcommand{\Gastc}{\mathsf{G4^{\delta}}}

\newcommand{\coax}{\overline{ax}}

\newcommand{\G}{\mathsf{G4^{\delta}}}

\newcommand{\topp}{\mathsf{top}}
\newcommand{\topstar}{\mathsf{top}^{\#}}
\newcommand{\topf}{\mathsf{top_{cp}}}
\newcommand{\topfo}{\mathsf{top_{r}}}
\newcommand{\topfstar}{\mathsf{top_{cp}^{\#}}}
\newcommand{\topfstaro}{\mathsf{top_{r}^{\#}}}

\newcommand{\quest}{\scriptsize\sststile{}{\ast}}

\usepackage{stackengine}
\def\ruleoffset{1pt}

\newcommand\specialvdash[2]{\mathrel{\ensurestackMath{
  \mkern2mu\rule[-\dp\strutbox]{.4pt}{\baselineskip}\stackon[\ruleoffset]{
    \stackunder[\dimexpr\ruleoffset-.5\ht\strutbox+.5\dp\strutbox]{
      \rule[\dimexpr.5\ht\strutbox-.5\dp\strutbox]{2.5ex}{.4pt}}{
        \scriptstyle #1}}{\scriptstyle#2}\mkern2mu}}
}

\newcommand\specialvdashdef[2]{\mathrel{_{_{{\sf F}}}\ensurestackMath{
  \mkern2mu\rule[-\dp\strutbox]{.4pt}{\baselineskip}\stackon[\ruleoffset]{
    \stackunder[\dimexpr\ruleoffset-.5\ht\strutbox+.5\dp\strutbox]{
      \rule[\dimexpr.5\ht\strutbox-.5\dp\strutbox]{2.5ex}{.4pt}}{
        \scriptstyle #1}}{\scriptstyle#2}\mkern2mu}}
}

\newcommand\specialvdashnorm[2]{\mathrel{_{_{{\sf O}}}\ensurestackMath{
  \mkern2mu\rule[-\dp\strutbox]{.4pt}{\baselineskip}\stackon[\ruleoffset]{
    \stackunder[\dimexpr\ruleoffset-.5\ht\strutbox+.5\dp\strutbox]{
      \rule[\dimexpr.5\ht\strutbox-.5\dp\strutbox]{2.5ex}{.4pt}}{
        \scriptstyle #1}}{\scriptstyle#2}\mkern2mu}}
}

\newcommand\specialvdashprohib[2]{\mathrel{_{_{{\sf PR}}}\ensurestackMath{
  \mkern2mu\rule[-\dp\strutbox]{.4pt}{\baselineskip}\stackon[\ruleoffset]{
    \stackunder[\dimexpr\ruleoffset-.5\ht\strutbox+.5\dp\strutbox]{
      \rule[\dimexpr.5\ht\strutbox-.5\dp\strutbox]{2.5ex}{.4pt}}{
        \scriptstyle #1}}{\scriptstyle#2}\mkern2mu}}
}

\newcommand\specialvdashperm[2]{\mathrel{_{_{{\sf P}}}\ensurestackMath{
  \mkern2mu\rule[-\dp\strutbox]{.4pt}{\baselineskip}\stackon[\ruleoffset]{
    \stackunder[\dimexpr\ruleoffset-.5\ht\strutbox+.5\dp\strutbox]{
      \rule[\dimexpr.5\ht\strutbox-.5\dp\strutbox]{2.5ex}{.4pt}}{
        \scriptstyle #1}}{\scriptstyle#2}\mkern2mu}}
}

\newcommand\specialvdashany[2]{\mathrel{_{_{{\sf \times}}}\ensurestackMath{
  \mkern2mu\rule[-\dp\strutbox]{.4pt}{\baselineskip}\stackon[\ruleoffset]{
    \stackunder[\dimexpr\ruleoffset-.5\ht\strutbox+.5\dp\strutbox]{
      \rule[\dimexpr.5\ht\strutbox-.5\dp\strutbox]{2.5ex}{.4pt}}{
        \scriptstyle #1}}{\scriptstyle#2}\mkern2mu}}
}

\newcommand{\remove}[1]{}

\begin{document}




\title{A logic for default deontic reasoning}

\author[M. Piazza]{Mario Piazza}
\address{Scuola Normale Superiore di Pisa
\\Classe di Lettere e Filosofia}
\email[M. Piazza]{mario.piazza@sns.it}

\author[A. Sabatini]{Andrea Sabatini}
\address{Scuola Normale Superiore di Pisa
\\Classe di Lettere e Filosofia}
\email[A. Sabatini]{andrea.sabatini@sns.it}

\maketitle

\begin{abstract} 
\begin{changemargin}{3em}{3em}
\linespread{1.1}

In many real-life settings, agents must navigate dynamic environments while reasoning under incomplete information and acting on a {\em corpus} of unstable, context-dependent, and often conflicting norms. We introduce a general, non-modal, proof-theoretic framework for deontic reasoning grounded in default logic. Its central feature  is the notion of {\em controlled sequent} -- a sequent annotated with sets of formulas (control sets) that prescribe what should or should not be entailed by the formulas in the antecedent. When combined with distinct extra-logical rules representing defaults and norms, these control sets record the conditions and constraints governing their applicability, thereby enabling local soundness checks for derived sequents. We prove that controlled sequent calculi satisfies admissibility of contraction and non-analytic cuts, and we establish their strong completeness with respect to credulous consequence in default theories and normative systems. Finally, we illustrate in depth how controlled sequent calculi provide a flexible and expressive basis for resolving deontic conflicts and capturing dynamic deontic notions via appropriate extra-logical rules.
\end{changemargin}
\end{abstract}

\section{Introduction}

Norms -- whether expressed as parental recommendations, military orders, traffic regulations, religious commandments, workplace policies, or self-imposed codes of conduct -- present a logical puzzle. On one hand, their observance or violation is a binary matter, naturally modeled by the semantics of classical propositional logic. On the other hand, norms are often unstable and context-dependent: the introduction of a new norm can override or invalidate previous ones, which places their behavior squarely within the realm of non-monotonic logic. To take an example from the first book of Plato's {\em Republic}, the norm ``One ought to return what one has borrowed'' may be overridden in exceptional circumstances -- for instance, if the person who lent a weapon has since gone mad. The norm is not monotonic: it applies by default, but can be defeated by relevant contextual changes, just as the generalization ``dogs bark'' is overridden by the more specific fact that ``Basenjis do not bark''. Finally, norms interact in complex ways, sometimes generating conflicts, even though -- being commands rather than declarative statements -- they do not possess truth values in the traditional logical sense. 


As is well known, from a logical standpoint, the most traditional approach to normative reasoning is provided by the family of deontic logics, which extend classical modal logic by introducing specific operators that enable the formal representation of norms, obligations and permissions within a Hilbert-style system, sequent calculus, or one of their variants \cite{wrightbook, hilpinen2013, Orl14, CT24}. However, standard deontic logics suffer from several well-documented limitations, including difficulties in addressing deontic paradoxes (such as the Good Samaritan paradox \cite{Prior58} or contrary-to-duty obligations \cite{Chis63, For84, vdTT97}), and limited flexibility in accommodating context-sensitive, specificity-based or dynamic normative structures. By contrast, about 25 years ago a different proposal was developed in terms of the so-called Input/Output Logics (I/O logics) \cite{MakvdT00, MakvdT01,MvT03}. Such systems offer a non-modal alternative by treating conditional obligations as inference rules, explicitly separating input conditions (I) from outputs (O) expressing what is obliged, prohibited or permitted under those conditions. This framework emphasizes what ought to be done in a given situation, without treating obligations as modal truths. However -- while more adaptable and better suited for modeling conditional and defeasible norms -- I/O logics remain at a relatively early stage in the development of robust and comprehensive proof systems\footnote{For a recent proposal based on a hypersequent framework, see \cite{IOhypersequent}.}.

The novel perspective proposed in this work differs significantly from both of these paradigms: it is modular, computationally more efficient, and provides fine-grained control over inference steps. These features allow for more transparent, scalable, and adaptive normative reasoning, particularly in complex or multi-layered settings.
Our aim is to develop a general and non-modal proof-theoretic framework for norms from the standpoint of a rational or artificial agent navigating a dynamic environment. This environment is stable enough to support expected obligations and prohibitions while remaining open and adaptable to new, potentially conflicting norms. As evidence of such flexibility, we will consider a wide range of concrete deontic scenario. 

Moreover, the framework introduced articulates the interplay between normative reasoning and default reasoning. The latter is a type of non-monotonic reasoning that enables the derivation of plausible conclusions in contexts of incomplete information and absent explicit contradictory evidence. Such conclusions are defeasible: they can be retracted when challenged by new information \cite{Reiter80,marekbook,Antoniou97}. Default reasoning can be formalized by extending classical logic with a collection of {\em extra-logical axioms} -- which represent the propositional contents of an ideal reasoner's beliefs -- and a set of {\em default rules}, which encode the inferential pathways that lead to defeasible yet consistent conclusions. 
In adopting a {\em credulous} approach, the ideal reasoner accepts as many individually consistent beliefs as possible. When reasoning with norms, such credulity is not merely permissible but necessary: a skeptical stance would unduly inhibit normative inference by prematurely excluding viable obligations. The idea of connecting normative and default reasoning is not new, as it has been especially suggested by John Horty \cite{Horty1993, Horty2012}. However, our approach is implemented beyond a purely programmatic level, by means of a proof-theoretic platform.

Technically, such implementation is based on the notion of {\em controlled} classical proof-system: a system built upon classical sequents augmented with a layer of extra-logical information, whose propagation is  governed by the very structure of derivations.
In earlier work, controlled classical sequents $\Gamma \vdash_\mathbf{S} \Delta$ were introduced to capture non-monotonic and paraconsistent behaviour within a classical framework \cite{JLC17}. The control set $\mathbf{S}$ encodes extra-logical information that can selectively block the derivability of $\Delta$ from $\Gamma$ \cite{JAL2013,JAL2014,EJPS2015,JLC17,piatesi24}.
In the present paper, by contrast, controlled sequents are syntactic devices tailored specifically to represent defaults and norms over a classical base.
 We adopt specific extra-logical rules to formalize defaults, as well as (un)conditional obligations and permissions:  
each sequent in the calculus is annotated with a \emph{control pair} $\langle \mathbf{T}, \mathbf{S} \rangle$ (written:  $\Gamma\specialvdash{\mathbf{S}}{\mathbf{T}}\Delta$), where $\mathbf{T}$ and $\mathbf{S}$ are sets of formulas that specify prescriptions on what should or should not be entailed by the formulas in the antecedent for an extra-logical rules to apply. To provide an initial sense of the approach advanced in this paper, we begin with an informal example.

\remove{Informally, a controlled sequent $\Gamma\specialvdash{\mathbf{S}}{\mathbf{T}}\Delta$ signifies 
a normative statement: if $\bigwedge\Gamma$ holds, then $\bigvee\Delta$ {\em ought to} hold provided that
all formulas in $\mathbf{T}$ are entailed by $\bigwedge\Gamma$ {\color{blue}(plus other formulas)}, and no formula in $\mathbf{S}$ is entailed by $\bigwedge\Gamma$ {\color{blue}(plus other formulas)}.}

\begin{example}
    Let the extra-logical rule
    \smallskip
    \begin{center}
        {\AxiomC{$\Gamma\specialvdash{{\bf S_{1}}}{{\bf T_{1}}}p$}
        \AxiomC{$\Gamma\specialvdash{{\bf S_{2}}}{{\bf T_{2}}}q$}
        \BinaryInfC{$\Gamma\specialvdash{{\bf S}}{{\bf T}}r$}
        \DisplayProof}
    \end{center}
    \smallskip
    stand for the normative statement: `If one prays and he is a male, one ought to wear Tefillin' -- where ${\bf T}={\bf T_{1}}\cup{\bf T_{2}}\cup\{\{p,q\}\}$. Let $s$ stand for the (factual) statement `it is nighttime'. Now, consider the following extra-logical rule:
    \smallskip
    \begin{center}
        {\AxiomC{$\Gamma\specialvdash{{\bf S_{1}}}{{\bf T_{1}}}p$}
        \AxiomC{$\Gamma\specialvdash{{\bf S_{2}}}{{\bf T_{2}}}q$}
        \BinaryInfC{$\Gamma\specialvdash{{\bf S'}}{{\bf T}}r$}
        \DisplayProof}
    \end{center}
    \smallskip
    where ${\bf S'}={\bf S}\cup\{\{s\}\}$. Such extra-logical rule expresses the normative claim: `if one prays, then one ought to wear the Tefillin provided that one is male, unless it is nighttime'. This formulation captures the rule governing the wearing of Tefillin during prayer, taking both gender and temporal constraints into account.
\end{example}

\remove{\begin{example}
    Consider the sequent $p \specialvdash{\varnothing}{\{p\}} q$, representing the normative statement:  `If one prays, one ought to wear the Tefillin'. 

    \noindent The Torah mandates Tefillin for males during prayer, while the Talmud forbids wearing them at night. Let $q_{1},q_{2}$ stand for the (factual) statements `one is a male' and `it is nighttime', respectively.
    
    \noindent  The controlled sequent $p \specialvdashnorm{{\bf S_1}}{{\bf T_1}} q$, where ${\bf T_1} = {p, q_1}$ and ${\bf S_1} = {q_2}$, expresses the normative statement: `if one prays, then one ought to wear the Tefillin provided that one is male, unless it is nighttime'. This formulation captures the rule for wearing Tefillin during prayer, accounting for gender and time constraints.
\end{example}}

The paper is structured as follows. Section \ref{preliminary} contains the notions and results concerning normative systems and default logics that we will subsequently employ. In particular, we introduce a notion of {\em deontic extension} for systems of obligations and permissions, analogous to the notion of \L ukasiewicz extension for defaults. In Section \ref{controlled} we introduce a uniform proof-theoretic platform for defaults, obligations and permissions based on controlled sequents. Controlled sequent calculi enjoy admissibility of contraction and non-analytic cut, and proofs without non-analytic cuts exhibit a weakened form of analyticity. We establish that controlled sequent calculi are strongly complete with respect to credulous consequence based on \L ukasiewicz extensions (for defaults) and deontic extensions (for obligations and permissions). Hence, we show that controlled sequent calculi enable the formalization of weak versions of cumulative transitivity and cautious monotony for the underlying credulous consequence relations. Section \ref{examples} gathers a number of examples of deontic scenarios, observed through the {\em lens} of controlled calculi. Lastly, Section \ref{conclusion} outlines directions for future research.

\section{Preliminary notions and results}\label{preliminary}

\begin{figure}
\begin{tabular}{llll}
{\sc axioms} & \\
& \\
{\AxiomC{}
\RightLabel{$ax$}
\UnaryInfC{$\Gamma,p\vdash p,\Delta$}
\DisplayProof} \quad
{\AxiomC{}
\RightLabel{$ax$}
\UnaryInfC{$\Gamma,p,\neg p\vdash\Delta$}
\DisplayProof} \quad
{\AxiomC{}
\RightLabel{$ax$}
\UnaryInfC{$\Gamma\vdash p,\neg p,\Delta$}
\DisplayProof} \quad
{\AxiomC{}
\RightLabel{$ax$}
\UnaryInfC{$\Gamma,\neg p\vdash\neg p,\Delta$}
\DisplayProof}
\\
& \\
{\AxiomC{}
\RightLabel{$\coax$}
\UnaryInfC{$\Theta\dashv\Lambda$}
\DisplayProof} &  
\\
& \\
{\sc logical rules} & \\

& \\
{\AxiomC{$\Gamma,A,B\vdash \Delta$}
\RightLabel{$L\wedge$}
\UnaryInfC{$\Gamma,A\wedge B\vdash\Delta$}
\DisplayProof}
\quad\quad\quad\quad\quad\quad\quad
{\AxiomC{$\Gamma,A,B\dashv\Delta$}
\RightLabel{$L'\wedge$}
\UnaryInfC{$\Gamma,A\wedge B\dashv\Delta$}
\DisplayProof}
\\
& \\
{\AxiomC{$\Gamma\vdash\Delta,A$}
\AxiomC{$\Gamma\vdash\Delta,B$}
\RightLabel{$R\wedge$}
\BinaryInfC{$\Gamma\vdash\Delta,A\wedge B$}
\DisplayProof}
\quad\quad\quad
{\AxiomC{$\Gamma\dashv\Delta,A_{i}$}
\RightLabel{$R'_{i}\wedge$}
\UnaryInfC{$\Gamma\dashv\Delta,A_{1}\wedge A_{2}$}
\DisplayProof}
\\
& \\
{\AxiomC{$A,\Gamma\vdash\Delta$}
\AxiomC{$B,\Gamma\vdash\Delta$}
\RightLabel{$L\vee$}
\BinaryInfC{$A\vee B,\Gamma\vdash\Delta$}
\DisplayProof}
\quad\quad\quad
{\AxiomC{$A_{i},\Gamma\dashv\Delta$}
\RightLabel{$L'_{i}\vee$}
\UnaryInfC{$A_{1}\vee A_{2},\Gamma\dashv\Delta$}
\DisplayProof}
\\
& \\
{\AxiomC{$\Gamma\vdash \Delta,A,B$}
\RightLabel{$R\vee$}
\UnaryInfC{$\Gamma\vdash\Delta,A\vee B$}
\DisplayProof}
\quad\quad\quad\quad\quad\quad\quad
{\AxiomC{$\Gamma\dashv\Delta,A,B$}
\RightLabel{$R_{i}'\vee$}
\UnaryInfC{$\Gamma\dashv\Delta,A\vee B$}
\DisplayProof}
\\
& \\
{\AxiomC{$\Gamma,\neg A\vdash \Delta$}
\AxiomC{$\Gamma,\neg B\vdash\Delta$}
\RightLabel{$L\neg\wedge$}
\BinaryInfC{$\Gamma,\neg (A\wedge B)\vdash\Delta$}
\DisplayProof}
\quad
{\AxiomC{$\Gamma,\neg A_{i}\dashv\Delta$}
\RightLabel{$L_{i}'\neg\wedge$}
\UnaryInfC{$\Gamma,\neg(A_{1}\wedge A_{2})\dashv\Delta$}
\DisplayProof}
\\
& \\
{\AxiomC{$\Gamma\vdash\Delta,\neg A,\neg B$}
\RightLabel{$R\neg\wedge$}
\UnaryInfC{$\Gamma\vdash\Delta,\neg(A\wedge B)$}
\DisplayProof}
\quad\quad\quad\quad\quad
{\AxiomC{$\Gamma\dashv\Delta,\neg A,\neg B$}
\RightLabel{$R'\neg\wedge$}
\UnaryInfC{$\Gamma\dashv\Delta,\neg(A\wedge B)$}
\DisplayProof}
\\
& \\
{\AxiomC{$\neg A,\neg B,\Gamma\vdash\Delta$}
\RightLabel{$L\neg\vee$}
\UnaryInfC{$\neg(A\vee B),\Gamma\vdash\Delta$}
\DisplayProof}
\quad\quad\quad\quad\quad
{\AxiomC{$\neg A,\neg B,\Gamma\dashv\Delta$}
\RightLabel{$L'\neg\vee$}
\UnaryInfC{$\neg(A\vee B),\Gamma\dashv\Delta$}
\DisplayProof}
\\
& \\
{\AxiomC{$\Gamma\vdash \Delta,\neg A$}
\AxiomC{$\Gamma\vdash \Delta,\neg B$}
\RightLabel{$R\neg\vee$}
\BinaryInfC{$\Gamma\vdash\Delta,\neg(A\vee B)$}
\DisplayProof}
\quad
{\AxiomC{$\Gamma\dashv\Delta,\neg A_{i}$}
\RightLabel{$R_{i}'\neg\vee$}
\UnaryInfC{$\Gamma\dashv\Delta,\neg(A_{1}\vee A_{2})$}
\DisplayProof}
\\
& \\
{\AxiomC{$\Gamma,A\vdash \Delta$}
\RightLabel{$L\neg\neg$}
\UnaryInfC{$\Gamma,\neg\neg A\vdash\Delta$}
\DisplayProof}
\quad\quad\quad\quad\quad\quad\quad
{\AxiomC{$\Gamma,A\dashv\Delta$}
\RightLabel{$L'\neg\neg$}
\UnaryInfC{$\Gamma,\neg\neg A\dashv\Delta$}
\DisplayProof}
\\
& \\
{\AxiomC{$\Gamma\vdash \Delta,A$}
\RightLabel{$R\neg\neg$}
\UnaryInfC{$\Gamma\vdash\Delta,\neg\neg A$}
\DisplayProof}
\quad\quad\quad\quad\quad\quad\quad
{\AxiomC{$\Gamma\dashv\Delta,A$}
\RightLabel{$R'\neg\neg$}
\UnaryInfC{$\Gamma\dashv\Delta,\neg\neg A$}
\DisplayProof}
\end{tabular}
\caption{${\sf G4pn}$ and $\fullGpn$ sequent calculi}
\label{fig:fullGpn}
\end{figure}

\subsection{The $\fullGpn$ calculus} We work with a propositional language containing a denumerable set of atoms $p,q,r,\ldots$, the unary connective $\neg$, and the binary connectives $\wedge,\vee$ for conjunction and disjunction. Capital Latin letters $A,B,C,\ldots$ range over formulas, while capital Greek letters $\Gamma,\Delta,\Pi,\Sigma,\ldots$ denote finite multisets of formulas. We use $\Theta,\Lambda,\ldots$ for finite multisets of {\em literals}, that is, atoms $p,q,r,\ldots$ and their negations $\neg p,\neg q,\neg r,\ldots$.
For contexts $\Theta=\{p_{1},\ldots,p_{k},\neg q_{1},\ldots,\neg q_{m}\}$ and $\Gamma=\{A_{1},\ldots,A_{n}\}$, we set:
 $$\Theta^{\bot}=\{\neg p_{1},\ldots,\neg p_{k},q_{1},\ldots,q_{m}\}\hspace{13pt}
\bigwedge\Gamma=A_{1}\wedge A_{2}\wedge\cdots\wedge A_{n}\hspace{13pt}
\bigvee\Gamma=A_{1}\vee A_{2}\vee\cdots\vee A_{n}$$

For $\Theta={A}$, we write $A^{\bot}$ for $\Theta^{\bot}$. For $\Gamma=\emptyset$, we stipulate $\Gamma^{\perp}=\Gamma$, $\bigwedge\Gamma=\top$, and $\bigvee\Gamma=\bot$, where $\top$ and $\bot$ denote an arbitrary tautology and contradiction, respectively. The {\em logical complexity} $\mathrm{C}(A)$ of a formula $A$ is defined as follows: $\mathrm{C}(A)=1$ if $A$ is a literal; $\mathrm{C}(A)=\mathrm{C}(B)+\mathrm{C}(C)+1$ if $A$ has the form $B\otimes C$; and $\mathrm{C}(A)=\mathrm{C}(\neg B)+\mathrm{C}(\neg C)+1$ if $A$ has the form $\neg(B\otimes C)$, where $\otimes\in\{\wedge,\vee\}$. This extends to multisets $\Gamma=A_{1},\ldots,A_{n}$ by setting $\mathrm{C}(\Gamma)=\mathrm{C}(A_{1})+\ldots+\mathrm{C}(A_{n})$.

We consider Gentzen-style sequents $\Gamma \vdash \Delta$ as well as  {\em antisequents} $\Gamma \dashv \Delta$, where $\Gamma \dashv \Delta$ is valid if and only if $\Gamma \vdash \Delta$ is invalid \cite{Goranko94,GPS}. In classical logic, an antisequent is valid precisely when some Boolean valuation makes every formula in $\Gamma$ true and every formula in $\Delta$ false.

The system $\fullGpn$ for classical propositional logic is a modification of the system $\fullG$. The latter, imported from \cite{RSL20}, treats logical contexts as multisets of formulas and is obtained by extending Kleene's original ${\sf G4}$ system \cite[pp. 289-290, p. 306]{Kleene1967} with the complementary axiom
{\AxiomC{}
\RightLabel{$\coax$}
\UnaryInfC{$\Theta\dashv\Lambda$}
\DisplayProof}, where $\Theta\cap\Lambda=\varnothing$, together with distinct rules for antisequents. 
The system $\fullGpn$ retains the rules for $\neg$ from $\fullG$ via a broader class of $ax$ instances, and further incorporates rules for $\neg\wedge$, $\neg\vee$, and $\neg\neg$--- hence the superscript ${\sf pn}$ for `primitive negation'. To speak uniformly about sequents ($\Gamma \vdash \Delta$) and antisequents ($\Gamma \dashv \Delta$) whenever the distinction is immaterial and they are taken independently of any specific deduction-refutation system deriving them, we use the term `$\ast$-sequent' and write the generic form $\Gamma \mathbin{\quest} \Delta$.
The measure $\mathrm{C}$ extends to any $\ast$-sequent $\Gamma \quest \Delta$ by setting $\mathrm{C}(\Gamma \quest \Delta)=\mathrm{C}(\Gamma)+\mathrm{C}(\Delta)$.

A $\fullGpn$ derivation $\pi$ may terminate either in a sequent $\Gamma \vdash \Delta$ or in an antisequent $\Gamma \dashv \Delta$. In the former case, we say that $\pi$ is a {\em proof}  of $\Gamma \vdash \Delta$; in the latter, $\pi$ is a {\em refutation}  of $\Gamma \vdash \Delta$. 
Any $\ast$-sequent $\Gamma \quest \Delta$ can be decomposed into a set of atomic $\ast$-sequents by applying, bottom-up, the rules $L\wedge, R\wedge, L\vee, R\vee, L\neg\wedge, R\neg\wedge, L\neg\vee, R\neg\vee, L\neg\neg, R\neg\neg$ from Figure \ref{fig:fullGpn}, with $\quest$ in place of $\vdash$, until every leaf of the resulting tree is an atomic $\ast$-sequent -- namely, one containing only literals.

We now highlight two crucial features of the $\fullGpn$ proof system:

\begin{proposition}
$\fullGpn$ proves (refutes) $\Gamma\vdash\Delta$ if and 
only if the formula $\neg\bigwedge\Gamma\vee\bigvee\Delta$ is classically valid (\,{\em in}valid).
\end{proposition}

\begin{proposition}
\label{stab}
Maximal $\fullGpn$-decomposition yields a unique set of atomic $\ast$-sequents.
\end{proposition}
\begin{proof}
We argue as in the proof of the same property for $\fullG$ in \cite{Avron93,RSL20}.
\end{proof}

Proposition \ref{stab} allows us to refer directly to the unique set of atomic $\ast$-sequents associated with any given $\ast$-sequent $\Gamma \quest \Delta$, since this decomposition is independent of the particular derivation producing it.  For any set of atomic $\ast$-sequents $\mathcal{C}$, we say that $\mathcal{C}$ is {\em closed under Negation} whenever $\Theta_{1},\Lambda_{2}^{\bot}\quest\Theta_{2}^{\bot},\Lambda_{1}$ belongs to $\mathcal{C}$ if $\Theta_{1},\Theta_{2}\quest\Lambda_{2},\Lambda_{1}$ belongs to $\mathcal{C}$. We write $\topp(\Gamma \quest \Delta)$ for the closure under Negation of the atomic $\ast$-sequents associated with $\Gamma \quest \Delta$. We write $\topf(\Gamma \quest \Delta)$ for the subset of those atomic $\ast$-sequents $\Theta \quest \Lambda$ in $\topp(\Gamma \quest \Delta)$ such that $\Theta \cap \Lambda = \varnothing$ and neither $\Theta$ nor $\Lambda$ does not contains both an atom and its negation; the subscript ${\sf cp}$ stands for `complementary'.  Finally, $\topfo(\Gamma \quest \Delta)$ denotes the complementary $\ast$-sequents of the form $\quest\Lambda$ in $\topf(\Gamma\quest\Delta)$ -- where the subscript ${\sf r}$ stands for `right-sided'.

\begin{example}
This is a decomposition-tree of $\neg(p\wedge q)\vee r\quest s\vee\neg t$:
\smallskip
\begin{center}
    {\rootAtTop
    \AxiomC{$\neg p\quest s,\neg t$}
    \AxiomC{$\neg q\quest s,\neg t$}
    \LeftLabel{\scriptsize{$L\neg\wedge$}}
    \BinaryInfC{$\neg(p\wedge q)\quest s,\neg t$}
    \AxiomC{$r\quest s,\neg t$}
    \RightLabel{\scriptsize{$L\vee$}}
    \BinaryInfC{$\neg(p\wedge q)\vee r\quest s,\neg t$}
    \RightLabel{\scriptsize{$R\vee$}}
    \UnaryInfC{$\neg(p\wedge q)\vee r\quest s\vee\neg t$}
    \DisplayProof}
\end{center}
\smallskip
Hence, we have that $\topp(\neg(p\wedge q)\vee r\quest s\vee\neg t)=\{\neg p\quest s,\neg t\,;\,\neg q\quest s,\neg t\,;\,r\quest s,\neg t\,;\,\neg p,\neg s\quest\neg t\,;\,\neg p,t\quest s\,;\,\neg p,t,\neg s\quest\ \,;\,\ \quest s,\neg t,p\,;\,\neg s\quest\neg t,p\,;\,t\quest s,p\,;\,t,\neg s\quest p\,;\,\neg q,\neg s\quest\neg t\,;\,\neg q,\neg s,t\quest\ \,;\,\neg q,t\quest s\,;\,\ \quest s,\neg t,q\,;\,\neg s\quest \neg t,q\,;\,\neg s,t\quest q\,;\,t\quest s,q\,;\,r,\neg s\quest \neg t\,;\,r,t\quest s\,;\,r,\neg s,t\quest\ \,;\,\neg s\quest \neg t,\neg r\,;\,t\quest s,\neg r\,;\,\neg s,t\quest\neg r\}$, $\topf(\neg(p\wedge q)\vee r\quest s\vee\neg t)=\topp(\neg(p\wedge q)\vee r\quest s\vee\neg t)$ and $\topfo(\neg(p\wedge q)\vee r\quest s\vee\neg t)=\{\quest p,s,\neg t\,;\,\ \quest s,\neg t,q\,;\,\ \quest\neg r,s,\neg t\}$.
\end{example}

\noindent For any set $\mathcal{C}$ of atomic $\ast$-sequents, we say that $\mathcal{C}$ is closed under Contraction whenever
\smallskip
    \begin{itemize}
        \item[$-$] if $A,A,\Theta\quest\Lambda$ or $A,\Theta\quest\Lambda,A^{\bot}$ is in $\mathcal{C}$, then $A,\Theta\quest\Lambda$ is in $\mathcal{C}$;
        \smallskip
        \item[$-$] if $\Theta\quest\Lambda,A,A$ or $A^{\bot},\Theta\quest\Lambda,A$ is in $\mathcal{C}$, then $\Theta\quest\Lambda,A$ is in $\mathcal{C}$.
    \end{itemize}
\smallskip 
Furthermore, we say that $\mathcal{C}$ is closed under Cut when $\Phi,\Theta\quest\Lambda,\Psi$ is in $\mathcal{C}$ if one of the following condition holds:
\smallskip
    \begin{itemize}
            \item[$-$] $\Theta\quest\Lambda,A$ and $A,\Phi\quest\Psi$;
            \smallskip
            \item[$-$] $\Theta\quest\Lambda,A$ and $\Phi\quest\Psi,A^{\bot}$;
            \smallskip
            \item[$-$] $A^{\bot},\Theta\quest\Lambda$ and $A,\Phi\quest\Psi$\footnote{When employing $\quest$ in the definition of closure under Cut, we do not mean to imply that closure under Cut preserves refutability within any particular deduction-refutation system; indeed, it clearly does not. Rather, our aim is to provide a formulation of closure under Cut that abstracts away from the commitments of any specific deduction-refutation framework.}.
    \end{itemize}
\smallskip
We write $\topstar(\Gamma\quest\Delta)$ to refer to the set of atomic $\ast$-sequents which is obtained from $\topp(\Gamma\quest\Delta)$ by maximal application of the following steps (cf. \cite[p. 9]{PiazzaTesi24}):
\smallskip
\begin{itemize}
    \item[$(i)$] start with $\mathcal{C}_{0}=\topp(\Gamma\quest\Delta)$;
    \smallskip
    \item[$(ii)$] take the closure under Contraction of $\mathcal{C}_{n}$;
    \smallskip
    \item[$(iii)$] if either $\Theta\quest\Lambda,A$ and $A,\Phi\quest\Psi$, or $A^{\bot},\Theta\quest\Lambda$ and $A,\Phi\quest\Psi$, or $A,\Theta\quest\Lambda$ and $A^{\bot},\Phi\quest\Psi$ belong to $\mathcal{C}_{n}$, and $\Phi,\Theta\quest\Lambda,\Psi$ does not belong to $\mathcal{C}_{n}$, then take $\mathcal{C}_{n+1}=\mathcal{C}_{n}\cup\{\Phi,\Theta\quest\Lambda,\Psi\}$.
\end{itemize}
\smallskip
We use $\topfstar(\Gamma\quest\Delta)$ to denote the set of complementary $\ast$-sequents in $\topstar(\Gamma\quest\Delta)$. Finally, we employ $\topfstaro(\Gamma\quest\Delta)$ to denote the smallest set containing all the complementary $\ast$-sequents in $\topstar(\Gamma\quest\Delta)$ whose left-hand side is empty.

\begin{example}
Consider the $\ast$-sequent $\quest(\neg p\vee q)\wedge(\neg q\vee r)$: we have that $\topp(\quest(\neg p\vee q)\wedge(\neg q\vee r))=\{\ \quest\neg p,q\,;\,\neg q\quest\neg p\,;\,p\quest q\,;\,p,\neg q\quest\ \,;\,\ \quest\neg q,r,r\,;\,q\quest r,r\,;\,\neg r\quest\neg q,r\,;\,q,\neg r\quest r\,;\,\neg r\quest\neg q,r\,;\,\neg r,q\quest r\,;\,\neg r,\neg r\quest\neg q\,;\,\neg r,q,\neg r\quest\ \}$ and $\topf(\quest(\neg p\vee q)\wedge(\neg q\vee r))=\topp(\quest(\neg p\vee q)\wedge(\neg q\vee r))$.

\noindent Hence, $\topfstar(\quest(\neg p\vee q)\wedge(\neg q\vee r))=\topf(\quest(\neg p\vee q)\wedge(\neg q\vee r))\cup\{q,\neg r\quest\ \,;\,q\quest r\,;\,\neg r\quest\neg q\,;\,\ \quest\neg q,r\,;\,q\quest r\,;\,\ \quest\neg q,r\,;\,\ \quest\neg p,r,r\,;\,\neg r\quest\neg p,r\,;\,\neg r,\neg r\quest\neg p\,;\,\ \quest\neg p,r\,;\,\neg r\quest\neg p\,;\,\ \quest \neg p,r\,;\,\neg r\quest\neg p\,;\,\ p\quest r,r\,;\,\neg r,p\quest r\,;\,\neg r,\neg r,p\quest\ \,;\,\ p\quest r\,;\,\neg r,p\quest\ \,;\,p\quest r\,;\,p,\neg r\quest\ \}$.   
\end{example}

\subsection{Default theories}\label{default theories}

A {\em default} is a domain-specific inference rule of the form 
\begin{equation}\label{eq:def}
\dfrac{B:C_{1},\ldots,C_{n}}{D}
\end{equation}
where $B$ is the {\em prerequisite}, the non-contradictory formulas $C_{1},\ldots,C_{n}$ are the {\em justifications} and the non-tautological formula $D$ is the {\em conclusion} of the default. Its interpretation is that if $B$ is proved, then $D$ is provable, so long as the formulas $\neg C_{1},\ldots,\neg C_{n}$ are not provable. We say that a default rule is {\em triggered} whenever its prerequisite is proved.

If $n=1$, then a default rule of the form (\ref{eq:def}) is {\em normal} if and only if $\fullGpn$ proves both the sequent $C_{1}\vdash D$ and $D\vdash C_{1}$ and {\em semi-normal} if and only if $\fullGpn$ proves $C_{1}\vdash D$. For convenience, when discussing normal default rules, we only focus on those for which $C_{1}=D$. 

A {\em default theory} is a pair $\langle\mathcal{W},\mathcal{D}\rangle$, where $\mathcal{W}$ is a finite, consistent set of extra-logical axioms and $\mathcal{D}$ is a finite, non-empty set of default rules. A default theory $\langle\mathcal{W},\mathcal{D}\rangle$ is said to be normal if and only if all the defaults in $\mathcal{D}$ are normal. We will use $req(\mathcal{D}'),just(\mathcal{D}')$ and $concl(\mathcal{D}')$ to refer to the set of prerequisites, justifications and conclusions, respectively, of the defaults in any $\mathcal{D}'\subseteq\mathcal{D}$. 

A {\em modified extension} (called also {\em \L ukaszewicz extension}) is a set of formulas derivable from $\mathcal{W}$ by classical logic and the maximal application of default rules in $\mathcal{D}$ whose consistency condition holds (both before and after their being triggered) relatively to a simultaneously defined support set (\cite[pp. 75-76]{Antoniou97}, \cite{Lukas88}). Specifically, we say that $\langle\mathcal{E},\mathcal{F}\rangle$ is a modified extension of $\langle\mathcal{W},\mathcal{D}\rangle$ if and only if $\mathcal{E}$ and $\mathcal{F}$ are quasi-inductively defined as follows:
\smallskip
\begin{center}
$\mathcal{E}^{0}=\mathcal{W}\ \ \ \ \ \mathcal{F}^{0}=\varnothing$
\end{center}
\begin{center}
\small{
$\mathcal{E}^{k+1}=Cn(\mathcal{E}^{k})\cup\{concl(\mathcal{D'})\mid\mathcal{D}'\subseteq\mathcal{D},req(\mathcal{D}')\subseteq\mathcal{E}^{k},A\in(\mathcal{F}\cup just(\mathcal{D}'))\Rightarrow \neg A\not\in Cn(\mathcal{E}\cup concl(\mathcal{D}'))\}$}
\end{center}
\begin{center}
\small{
$\mathcal{F}^{k+1}=\mathcal{F}^{k}\cup\{just(\mathcal{D'})\mid\mathcal{D}'\subseteq\mathcal{D},req(\mathcal{D}')\subseteq\mathcal{E}^{k},A\in(\mathcal{F}\cup just(\mathcal{D}'))\Rightarrow \neg A\not\in Cn(\mathcal{E}\cup concl(\mathcal{D}'))\}$}
\end{center}
\begin{center}
$\mathcal{E}=\bigcup\limits_{i=0}^{\omega}\mathcal{E}^{i}\ \ \ \ \ \mathcal{F}=\bigcup\limits_{i=0}^{\omega}\mathcal{F}^{i}$
\end{center}

\begin{example}
Let $\langle\mathcal{W},\mathcal{D}\rangle$ be defined as follows:
$$\mathcal{W}=\varnothing$$
$$\mathcal{D}=\Big\{\dfrac{\top:\neg q\wedge p}{p}\ ,\ \dfrac{\top:\neg r\wedge q}{q}\ ,\ \dfrac{\top:\neg p\wedge r}{r}\Big\}$$
$\langle\mathcal{W},\mathcal{D}\rangle$ has multiple modified extensions -- namely, $\langle\mathcal{E}_{1},\mathcal{F}_{1}\rangle=\langle Cn(\{p\}),\{\neg q\wedge p\}\rangle$, $\langle\mathcal{E}_{2},\mathcal{F}_{2}\rangle=\langle Cn(\{q\})$, $\{\neg r\wedge q\}\rangle$ and $\langle\mathcal{E}_{3},\mathcal{F}_{3}\rangle=\langle Cn(\{r\}),\{\neg p\wedge r\}\rangle$.
\end{example}

If $\langle\mathcal{W},\mathcal{D}\rangle$ is a default theory, then $A$ is a {\em modified credulous consequence} (in short, an $m$-credulous consequence) of $\mathcal{W}$ if and only if $A$ belongs to at least one modified extension of $\langle\mathcal{W},\mathcal{D}\rangle$. Intuitively, the credulous agent commits to as many individually consistent beliefs as possible.

It is straightforward to see that the modified credulous consequence relation is supraclassical yet {\em non-monotonic}: conclusions it supports may fail to persist when arbitrary premises are added \cite{Makinson94}. Nevertheless, it satisfies the key property of {\em semi-monotonicity}: no modified credulous consequence of $\mathcal{W}$ can be invalidated by adding new default rules:

\begin{lemma}\label{semimon}
Let $\langle\mathcal{W},\mathcal{D}\rangle$ be a default theory, and $\mathcal{D}\subseteq\mathcal{D}'$. For any modified extension $\langle\mathcal{E},\mathcal{F}\rangle$ of $\langle\mathcal{W},\mathcal{D}\rangle$ there exists (at least) one modified extension $\langle\mathcal{E}',\mathcal{F}'\rangle$ of $\langle\mathcal{W},\mathcal{D}'\rangle$ such that $\mathcal{E}\subseteq\mathcal{E}'$ and $\mathcal{F}\subseteq\mathcal{F}'$.
\end{lemma}
\begin{proof}
For a proof see \cite[pp. 12-13]{Lukas88}.
\end{proof}

\subsection{Obligations and permissions as extra-logical rules}

We consider a {\em constrained conditional obligation} (permission) as a domain-specific inference rule of the form
\begin{equation}\label{eq:obperm}
\dfrac{B:C_{1},\ldots,C_{n}}{D}
\end{equation}
where $B$ is the {\em conditions}, the non-contradictory formulas $C_{1},\ldots,C_{n}$ are the {\em constraints} and the non-tautological formula $D$ is the {\em conclusion} of the constrained conditional obligation (permission, respectively). Its interpretation is that if $B$ is/ought to be (is/is permitted to be) the case, then $D$ ought to be (is permitted to be) the case, so long as $C_{1},\ldots,C_{n}$ ought not (are not permitted to, respectively) be the case. We say that a conditional obligation (permission) has a {\em factual condition} if the intended interpretation of the prerequisite is factual: if $B$ {\em is the case}, then $D$ ought (is permitted) to be the case. On the other hand, we say that a conditional obligation (permission) has a {\em deontic condition} if the intended interpretation of the prerequisite is deontic, non-factual: if $B$ {\em ought (is permitted) to be the case}, then $D$ ought (is permitted) to be the case. If $n=1$, we say that a constrained conditional obligation (permission) of the form (\ref{eq:obperm}) is {\em normal} whenever $C_{1}=D$\footnote{Let us emphasize that, whereas \cite{MakvdT01} considers a fixed set of constraints for all conditional obligations, we expand this perspective by allowing each conditional obligation (and each permission) to be associated with its own specific set of constraints. In addition, we address defaults and conditional norms within a unified framework, thereby enabling default reasoning over the factual conditions that govern obligations and permissions.}. If $B$ is logically equivalent to $\top$, a rule of the form \ref{eq:obperm} is a constrained {\em unconditional} obligation. We say that an obligation (permission) is {\em triggered} whenever its condition is proved. 

An {\em obligation system} is a pair $\langle \mathcal{W}^{\sf o}, \mathcal{O} \rangle$, where $\mathcal{W}^{\sf o}$ is a finite and consistent set of extra-logical axioms -- expressing unconstrained, unconditional obligations -- and $\mathcal{O}=\mathcal{O}^{\sf f}\cup\mathcal{O}^{\sf d}$, where $\mathcal{O}^{\sf f}$ is a finite, non-empty set of constrained (un)conditional obligations with factual conditions and $\mathcal{O}^{\sf d}$ is a finite set of constrained (un)conditional obligations with deontic conditions.
Similarly, a {\em permission system} is a pair $\langle \mathcal{W}^{\sf p}, \mathcal{P} \rangle$, where $\mathcal{W}^{\sf p}$ is a finite, consistent set of extra-logical axioms -- expressing unconstrained, unconditional permissions -- and $\mathcal{P}=\mathcal{P}^{\sf f}\cup\mathcal{P}^{\sf d}$, where $\mathcal{P}^{\sf f}$ is a finite, non-empty set of constrained (un)conditional permissions with factual conditions and $\mathcal{P}^{\sf d}$ is a finite set of constrained (un)conditional permissions with deontic conditions.

A {\em normative system} $\mathcal{N}$ is a pair $\langle\langle\mathcal{W}^{\sf o},\mathcal{O}\rangle,\langle\mathcal{W}^{\sf p},\mathcal{P}\rangle\rangle$, where $(i)$ $\langle\mathcal{W}^{\sf o},\mathcal{O}\rangle$ is an obligation system, $(ii)$ $\langle\mathcal{W}^{\sf p},\mathcal{P}\rangle$ is a permission system and $(iii)$ $\mathcal{W}^{\sf o}\subseteq\mathcal{W}^{\sf p}$ and $\mathcal{O}^{\star}\subseteq\mathcal{P}^{\star}$ with $\star\in\{{\sf f},{\sf d}\}$. A normative system $\langle\langle\mathcal{W}^{\sf o},\mathcal{O}\rangle,\langle\mathcal{W}^{\sf p},\mathcal{P}\rangle\rangle$ is said to be normal if and only if all the permissions in $\mathcal{P}$ are normal. We shall use $cond(\mathcal{Q}')$, $constr(\mathcal{Q}')$ and $concl(\mathcal{Q}')$ to refer to the set of conditions, constraints and conclusions, respectively, of the obligations (permissions) in any $\mathcal{Q}'\subseteq\mathcal{O}$ ($\mathcal{Q}'\subseteq\mathcal{P}$, respectively). 

Let $\langle\mathcal{W},\mathcal{D}\rangle$ be a default theory and $\mathcal{N}$ a normative system. Given a modified extension $\langle\mathcal{E},\mathcal{F}\rangle$ of $\langle\mathcal{W},\mathcal{D}\rangle$, an {\em obligation} ({\em permission}) {\em extension} is a set of formulas derivable from $\mathcal{E},\mathcal{W}^{\sf o}$ and $\mathcal{W}^{\sf p}$ by classical logic and maximal application of the obligations (permissions, respectively) in $\mathcal{N}$ whose consistency condition holds (both before and after their being triggered) relatively to a simultaneously defined support set. Specifically, we say that $\langle\mathcal{S}_{\mathcal{E}},\mathcal{T}_{\mathcal{E}}\rangle$ is an obligation extension of $\langle\langle\mathcal{W},\mathcal{D}\rangle,\mathcal{N}\rangle$ if and only if $\mathcal{S}_{\mathcal{E}}$ and $\mathcal{T}_{\mathcal{E}}$ are quasi-inductively defined as follows:
\smallskip
\begin{center}
\small{
$\mathcal{S}_{\mathcal{E}}^{0}=\mathcal{W}^{\sf o}\ \ \ \ \ \mathcal{T}_{\mathcal{E}}^{0}=\varnothing$}
\end{center}
\begin{center}
\footnotesize{
\[
\mathcal{S}_{\mathcal{E}}^{k+1}
=
Cn(\mathcal{S}_{\mathcal{E}}^{k})
\,\cup\,
\left\{
    concl(\mathcal{O}^{\sf f}_{x}),concl(\mathcal{O}^{\sf d}_{x})
    \;\middle|\;
    \begin{array}{l}
    \mathcal{O}^{\sf f}_{x}\subseteq\mathcal{O}^{\sf f},\mathcal{O}^{\sf d}_{x}\subseteq\mathcal{O}^{\sf d},\\
    req(\mathcal{O}^{\sf f}_{x})\subseteq\mathcal{E},req(\mathcal{O}^{\sf d}_{x})\subseteq\mathcal{S}_{\mathcal{E}}^{k},\\[2pt]
    A\in(\mathcal{T}_{\mathcal{E}}\cup just(\mathcal{O}^{\sf f}_{x}))\Rightarrow \neg A\not\in Cn(\mathcal{S}_{\mathcal{E}}\cup concl(\mathcal{O}_{x}^{\sf f})\cup concl(\mathcal{O}_{x}^{\sf d}))
    \end{array}
\right\}
\]}
\end{center}
\begin{center}
\footnotesize{
\[
\mathcal{T}_{\mathcal{E}}^{k+1}
=
\mathcal{T}_{\mathcal{E}}^{k}
\,\cup\,
\left\{
    just(\mathcal{O}^{\sf f}_{x}),just(\mathcal{O}^{\sf d}_{x})
    \;\middle|\;
    \begin{array}{l}
    \mathcal{O}^{\sf f}_{x}\subseteq\mathcal{O}^{\sf f},\mathcal{O}^{\sf d}_{x}\subseteq\mathcal{O}^{\sf d},\\
    req(\mathcal{O}^{\sf f}_{x})\subseteq\mathcal{E},req(\mathcal{O}^{\sf d}_{x})\subseteq\mathcal{S}_{\mathcal{E}}^{k},\\[2pt]
    A\in(\mathcal{T}_{\mathcal{E}}\cup just(\mathcal{O}^{\sf f}_{x}))\Rightarrow \neg A\not\in Cn(\mathcal{S}_{\mathcal{E}}\cup concl(\mathcal{O}_{x}^{\sf f})\cup concl(\mathcal{O}_{x}^{\sf d}))
    \end{array}
\right\}
\]}
\end{center}
\begin{center}
$\mathcal{S}_{\mathcal{E}}=\bigcup\limits_{i=0}^{\omega}\mathcal{S}_{\mathcal{E}}^{i}\ \ \ \ \ \mathcal{T}_{\mathcal{E}}=\bigcup\limits_{i=0}^{\omega}\mathcal{T}_{\mathcal{E}}^{i}$
\end{center}
\smallskip
A permission extension $\langle\mathcal{S}_{\mathcal{E}},\mathcal{T}_{\mathcal{E}}\rangle$ of $\langle\langle\mathcal{W},\mathcal{D}\rangle,\mathcal{N}\rangle$ is defined as an obligation extension of $\langle\langle\mathcal{W},\mathcal{D}\rangle,\mathcal{N}\rangle$, except for the fact that we replace $\mathcal{W}^{\sf o}$ with $\mathcal{W}^{\sf p}$ and $\mathcal{O}_{x}^{\star}$ with $\mathcal{P}_{x}^{\star}$, where $\star\in\{{\sf f},{\sf d}\}$. We say that $\langle\mathcal{S}_{\mathcal{E}},\mathcal{T}_{\mathcal{E}}\rangle$ is a deontic extension of $\langle\langle\mathcal{W},\mathcal{D}\rangle,\mathcal{N}\rangle$ if $\langle\mathcal{S}_{\mathcal{E}},\mathcal{T}_{\mathcal{E}}\rangle$ is an obligation extension or a permission extension of $\langle\langle\mathcal{W},\mathcal{D}\rangle,\mathcal{N}\rangle$.

If $\langle\langle\mathcal{W},\mathcal{D}\rangle,\mathcal{N}\rangle$ is a default theory associated with a normative system, then $A$ is a {\em deontic credulous consequence} (in short, a $d$-credulous consequence) of $\mathcal{W}\cup\mathcal{W}^{\sf o}\cup\mathcal{W}^{\sf p}$ if and only if $A$ belongs to at least one deontic extension of $\langle\langle\mathcal{W},\mathcal{D}\rangle,\mathcal{N}\rangle$. Intuitively, the credulous agent commits herself to as many individually consistent obligations and permissions as possible.

\begin{lemma}\label{deonticsemimon}
    Let $\langle\mathcal{E},\mathcal{F}\rangle$ be a modified extension of a default theory $\langle\mathcal{W},\mathcal{D}\rangle$, $\mathcal{N}$ be a normative system $\langle\langle \mathcal{W}^{\sf o},\mathcal{O}\rangle,\langle \mathcal{W}^{\sf p},\mathcal{P}\rangle\rangle$ and $\mathcal{N}'$ obtained from $\mathcal{N}$ by extending $\mathcal{O}$ or $\mathcal{P}$. For any deontic extension $\langle\mathcal{S}^{\mathcal{E}}_{1},\mathcal{T}^{\mathcal{E}}_{1}\rangle$ of $\langle\langle\mathcal{W},\mathcal{D}\rangle,\mathcal{N}\rangle$, there exists (at least) one deontic extension $\langle\mathcal{S}^{\mathcal{E}}_{2},\mathcal{T}^{\mathcal{E}}_{2}\rangle$ of $\langle\langle\mathcal{W},\mathcal{D}\rangle,\mathcal{N}'\rangle$ such that $\mathcal{S}^{\mathcal{E}}_{1}\subseteq\mathcal{S}^{\mathcal{E}}_{2}$ and $\mathcal{T}^{\mathcal{E}}_{1}\subseteq\mathcal{T}^{\mathcal{E}}_{2}$.
\end{lemma}

\begin{proof}
    We argue as in the proof of Lemma \ref{semimon}.
\end{proof}

\section{Controlled sequent calculi for conflicting information}\label{controlled}

To design controlled sequent calculi for deontic default reasoning, we introduce  the notion of {\em control pair}.  

\subsection{Control pairs}
In this subsection, we  use capital Greek letters $\Gamma,\Delta,\Theta,\ldots$ to denote multisets as well as {\em sets} of formulas (the distinction will be made explicit whenever necessary to avoid ambiguity). Boldface capital letters ${\bf S},{\bf S_{1}},\ldots...$ stand for finite sets of finite sets of formulas. 

A {\em control pair} is a pair of sets $\langle{\bf T},{\bf S}\rangle$, where ${\bf T}$ collects sets of {\em conditions}, and ${\bf S}$ collects sets of {\em constraints}. For their part, conditions and constraints are formulas, possibly labelled with specific superscripts ${\sf f}, {\sf o}, {\sf p},\ldots$ to denote their role as facts, obligations, permissions -- respectively. We decorate the turnstile $\vdash$ with a control pair. Intuitively, this decoration expresses that the derivation of a sequent $\Gamma \vdash \Delta$ must satisfy all conditions contained in each element of ${\bf T}$ and respect all constraints contained in each element of ${\bf S}$. In particular, we require that $\Gamma$ --- possibly strengthened by additional formulas (see Definition \ref{proof}) --- be compatible with the sets of conditions in ${\bf T}$ and with the sets of constraints in ${\bf S}$, in the following sense.

\begin{definition}\label{controlpair}
    Given a multiset $\Gamma$ and a control pair $\langle{\bf T},{\bf S}\rangle$, $\Gamma$ is {\em compatible} with $\langle{\bf T},{\bf S}\rangle$ (in symbols $\Gamma \, ||\, \langle{\bf T},{\bf S}\rangle$) if and only if for every formula $A$ in $\bigcup{\bf T}$ and every formula $B$ in $\bigcup{\bf S}$, $\fullGpn$ proves $\Gamma\vdash A$ and refutes $\Gamma\vdash B$.
\end{definition}

\begin{lemma}\label{compatibility}
    Let ${\bf S}\subseteq{\bf S'}$, ${\bf T}\subseteq{\bf T'}$ and the sequent $\Gamma\vdash\bigwedge\Delta$ be provable in $\fullGpn$. The following statements about compatibility hold:
    \smallskip
    \begin{enumerate}
    \item $\Gamma\,||\,\langle{\bf T'},{\bf S'}\rangle$ implies $\Gamma\,||\,\langle{\bf T},{\bf S}\rangle$; 
    \smallskip
    \item $\Delta\,||\,\langle{\bf T},\varnothing\rangle$ implies $\Gamma\,||\,\langle{\bf T},\varnothing\rangle$;
    \smallskip
    \item $\Gamma\,||\,\langle\varnothing,{\bf S}\rangle$ implies $\Delta\,||\,\langle\varnothing,{\bf S}\rangle$. \remove{\item $A\wedge B,\Gamma\, ||\, \langle{\bf T},{\bf S}\rangle$ if and only if $A,B,\Gamma\, ||\, \langle{\bf T},{\bf S}\rangle$.
        \smallskip
        \item $\neg(A\vee B),\Gamma\, ||\, \langle{\bf T},{\bf S}\rangle$ if and only if $\neg A,\neg B,\Gamma\, ||\, \langle{\bf T},{\bf S}\rangle$.
        \smallskip
        \item $A\vee B,\Gamma\, ||\, \langle{\bf T},{\bf S}\rangle$ if and only if $A,\Gamma\, ||\, \langle{\bf T},{\bf S}\rangle$ or $B,\Gamma\, ||\, \langle{\bf T},{\bf S}\rangle$.
        \smallskip
        \item $\neg(A\wedge B),\Gamma\, ||\, \langle{\bf T},{\bf S}\rangle$ if and only if $\neg A,\Gamma\, ||\, \langle{\bf T},{\bf S}\rangle$ or $\neg B,\Gamma\, ||\, \langle{\bf T},{\bf S}\rangle$.
        \smallskip
        \item $\neg\neg A,\Gamma\, ||\, \langle{\bf T},{\bf S}\rangle$ if and only if $A,\Gamma\, ||\, \langle{\bf T},{\bf S}\rangle$.
        \smallskip
        \item If $\Gamma\subseteq\Delta$, then $\Gamma\,||\,\langle{\bf T},\varnothing\rangle$ implies $\Delta\,||\,\langle{\bf T},\varnothing\rangle$ and $\Delta\,||\,\langle\varnothing,{\bf S}\rangle$ implies $\Gamma\,||\,\langle\varnothing,{\bf S}\rangle$.}
    \end{enumerate}
\end{lemma}

\begin{proof}
    Immediate from Definition \ref{controlpair}.
\end{proof}

Let $\langle\mathcal{W},\mathcal{D}\rangle$ be any default theory and $\mathcal{N}$ be any normative system associated with $\langle\mathcal{W},\mathcal{D}\rangle$. We exploit control pairs to define controlled sequent calculi for $m$-credulous consequences of $\mathcal{W}$ and $d$-credulous consequences of $\mathcal{W}\cup\mathcal{W}^{\sf o}\cup\mathcal{W}^{\sf p}$. 

To distinguish the roles of derivable formulas, we use labelled turnstiles in sequents: 
\smallskip
\begin{itemize}
    \item[$-$] $\specialvdashdef{}{}$ for sequents tracking $m$-credulous consequences,  
    \smallskip
    \item[$-$] $\specialvdashnorm{}{}$ for sequents tracking $d$-credulous consequences which correspond to obligations, and 
    \smallskip
    \item[$-$] $\specialvdashperm{}{}$ for sequents tracking $d$-credulous consequences which correspond to permissions.
\end{itemize}
\smallskip
When generalizing over a subset of $\{\specialvdashdef{}{}\,,\,\specialvdashnorm{}{}\,,\,\specialvdashperm{}{}\}$, we use the labelled turnstile $\specialvdashany{}{}$. 

\subsection{Control pairs for defaults} 

We are ready to present controlled sequent calculi for $m$-credulous consequence.

\begin{figure}
\begin{flushleft}\sc axioms and structural rules \\
\medskip
\bigskip
{\AxiomC{}
\RightLabel{$ax$}
\UnaryInfC{$\Theta\specialvdashdef{\varnothing}{\varnothing}\Lambda$}
\DisplayProof}
\\[\baselineskip]
{\AxiomC{$\Gamma\specialvdashdef{{\bf S}}{{\bf T}}\Delta$}
\RightLabel{$LW$}
\UnaryInfC{$A,\Gamma\specialvdashdef{{\bf S}}{{\bf T}}\Delta$}
\DisplayProof}\qquad
{\AxiomC{$\Gamma\specialvdashdef{{\bf S}}{{\bf T}}\Delta$}
\RightLabel{$RW$}
\UnaryInfC{$\Gamma\specialvdashdef{{\bf S}}{{\bf T}}\Delta,A$}
\DisplayProof}
\\[\baselineskip]
{\AxiomC{$\Gamma\specialvdashdef{{\bf S_{1}}}{{\bf T_{1}}}\Delta,A$}
\AxiomC{$\Gamma\specialvdashdef{{\bf S_{2}}}{{\bf T_{2}}}\Delta,\neg A$}
\RightLabel{$cut_{asa}$}
\BinaryInfC{$\Gamma\specialvdashdef{{\bf S}}{{\bf T}}\Delta$}
\DisplayProof}\qquad
{\AxiomC{$\Gamma\specialvdashdef{{\bf S}}{{\bf T}}\Delta$}
\RightLabel{$\sigma$}
\UnaryInfC{$\Gamma\specialvdashdef{{\bf S'}}{{\bf T}}\Delta$}
\DisplayProof}
\\[\baselineskip]
\medskip
\sc logical rules \\
\medskip
\bigskip
{\AxiomC{$\Gamma,A,B\specialvdashdef{{\bf S}}{{\bf T}} \Delta$}
\RightLabel{$L\wedge$}
\UnaryInfC{$\Gamma,A\wedge B\specialvdashdef{{\bf S}}{{\bf T}} \Delta$}
\DisplayProof}\qquad
{\AxiomC{$\Gamma\specialvdashdef{{\bf S_{1}}}{{\bf T_{1}}} \Delta,A$}
\AxiomC{$\Gamma\specialvdashdef{{\bf S_{2}}}{{\bf T_{2}}} \Delta,B$}
\RightLabel{$R\wedge$}
\BinaryInfC{$\Gamma\specialvdashdef{{\bf S}}{{\bf T}} \Delta,A\wedge B$}
\DisplayProof}
\\[\baselineskip]
{\AxiomC{$\Gamma,A\specialvdashdef{{\bf S_{1}}}{{\bf T_{1}}} \Delta$}
\AxiomC{$\Gamma,B\specialvdashdef{{\bf S_{2}}}{{\bf T_{2}}} \Delta$}
\RightLabel{$L\vee$}
\BinaryInfC{$\Gamma,A\vee B\specialvdashdef{{\bf S}}{{\bf T}}\Delta$}
\DisplayProof}\qquad
{\AxiomC{$\Gamma\specialvdashdef{{\bf S}}{{\bf T}} \Delta,A,B$}
\RightLabel{$R\vee$}
\UnaryInfC{$\Gamma\specialvdashdef{{\bf S}}{{\bf T}} \Delta, A\vee B$}
\DisplayProof}
\\[\baselineskip]
{\AxiomC{$\Gamma,\neg A\specialvdashdef{{\bf S_{1}}}{{\bf T_{1}}} \Delta$}
\AxiomC{$\Gamma,\neg B\specialvdashdef{{\bf S_{2}}}{{\bf T_{2}}} \Delta$}
\RightLabel{$L\neg\wedge$}
\BinaryInfC{$\Gamma,\neg(A\wedge B)\specialvdashdef{{\bf S}}{{\bf T}}\Delta$}
\DisplayProof}\qquad
{\AxiomC{$\Gamma\specialvdashdef{{\bf S}}{{\bf T}} \Delta,\neg A,\neg B$}
\RightLabel{$R\neg\wedge$}
\UnaryInfC{$\Gamma\specialvdashdef{{\bf S}}{{\bf T}} \Delta, \neg(A\wedge B)$}
\DisplayProof}
\\[\baselineskip]
{\AxiomC{$\Gamma,\neg A,\neg B\specialvdashdef{{\bf S}}{{\bf T}} \Delta$}
\RightLabel{$L\neg\vee$}
\UnaryInfC{$\Gamma,\neg(A\vee B)\specialvdashdef{{\bf S}}{{\bf T}} \Delta$}
\DisplayProof}\qquad
{\AxiomC{$\Gamma\specialvdashdef{{\bf S_{1}}}{{\bf T_{1}}} \Delta,\neg A$}
\AxiomC{$\Gamma\specialvdashdef{{\bf S_{2}}}{{\bf T_{2}}} \Delta,\neg B$}
\RightLabel{$R\neg\vee$}
\BinaryInfC{$\Gamma\specialvdashdef{{\bf S}}{{\bf T}} \Delta,\neg(A\vee B)$}
\DisplayProof}
\\[\baselineskip]
{\AxiomC{$A,\Gamma\specialvdashdef{{\bf S}}{{\bf T}}\Delta$}
\RightLabel{$L\neg\neg$}
\UnaryInfC{$\neg\neg A,\Gamma\specialvdashdef{{\bf S}}{{\bf T}}\Delta$}
\DisplayProof}\qquad
{\AxiomC{$\Gamma\specialvdashdef{{\bf S}}{{\bf T}}\Delta,A$}
\RightLabel{$R\neg\neg$}
\UnaryInfC{$\Gamma\specialvdashdef{{\bf S}}{{\bf T}}\Delta,\neg \neg A$}
\DisplayProof}
\\[\baselineskip]
\medskip
\sc extra-logical rules \\
\medskip
\bigskip
{\AxiomC{$\{\Gamma\specialvdashdef{{\bf S_{i}}}{{\bf T_{i}}}\Theta_{i}\}_{0\leq i\leq m}$}
\RightLabel{$\delta$}
\UnaryInfC{$\Gamma\specialvdashdef{{\bf S}}{{\bf T}}\Phi$}
\DisplayProof}
\end{flushleft}
\caption{Controlled sequent calculus for a default theory $\langle\mathcal{W},\mathcal{D}\rangle$}
\label{fig:Gdelta}
\end{figure}

\begin{definition}\label{definitioncalculi}
 Let $\langle\mathcal{W},\mathcal{D}\rangle$ be any default logic. The controlled sequent calculus $\Gastc$ for $\langle\mathcal{W},\mathcal{D}\rangle$ can be defined by adopting the rules in Figure \ref{fig:Gdelta}, provided that the conditions $(i)-(vii)$ below are fulfilled.
\smallskip
\begin{itemize}
    \item[$(i)$] For each instance of $ax$, one of the following conditions holds: $\Theta=\Lambda=\{A\}$ for some literal $A$,  $\Theta=\{A,\neg A\}$ or $\Lambda=\{A,\neg A\}$ for some atom $A$, $\Theta\quest\Lambda$ belongs to $\topfstar(\quest W)$, with $W$ being the conjunction of the formulas in $\mathcal{W}$.
    \smallskip
    \item[$(ii)$] For each instance of binary structural and logical rules, ${\bf T}={\bf T_{1}}\cup{\bf T_{2}}$ and ${\bf S}={\bf S_{1}}\cup{\bf S_{2}}$.
    \smallskip
    \item[$(iii)$] For each instance of $cut_{asa}$, $A$ is an atom occurring in some formula in $\Gamma\cup\mathcal{W}$.
    \smallskip
    \item[$(iv)$] For each instance of $\sigma$, ${\bf S}\subseteq{\bf S'}$.
    \smallskip
    \item[$(v)$] For any default rule of the form $\dfrac{B:C_{1},\ldots,C_{k}}{D}$ in $\mathcal{D}$, if $\topfo(\quest B)=\{\quest\Theta_{1},\ldots,\quest\Theta_{m}\}$ with $m>0$ and $\quest\Phi$ belongs to $\topfstaro(\quest D)$, then there exists an extra-logical rule of the  form:
    \smallskip
    \begin{center}
        {\AxiomC{$\{\Gamma\specialvdashdef{{\bf S_{i}}}{{\bf T_{i}}}\Theta_{i}\}_{1\leq i\leq m}$}
        \RightLabel{$\delta$}
        \UnaryInfC{$\Gamma\specialvdashdef{{\bf S}}{{\bf T}}\Phi$}
        \DisplayProof}
    \end{center}
    \smallskip
    with ${\bf T}=\bigcup_{i=1}^{m}{\bf T_{i}}\cup\{\{(\bigwedge_{i=1}^{m}(\bigvee\Theta_{i}))^{{\sf f}}\}\}$ and ${\bf S}=\bigcup\limits_{i=1}^{m}{\bf S_{i}}\cup\{\{\neg C_{1}^{{\sf f}},\ldots,\neg C_{k}^{{\sf f}}\}\}$. 
    
    \noindent If $\topfo(\quest B)=\varnothing$ and $\quest\Phi$ belongs to $\topfstaro(\quest D)$, then there exists an extra-logical rule of the  form:
    \smallskip
    \begin{center}
        {\AxiomC{$ $}
        \RightLabel{$\delta$}
        \UnaryInfC{$\Gamma\specialvdashdef{{\bf S}}{{\bf T}}\Phi$}
        \DisplayProof}
    \end{center}
    \smallskip
    with ${\bf T}$ being $\{\{\top^{\sf f}\}\}$ and ${\bf S}$ being $\{\{\neg C_{1}^{\sf f},\ldots,\neg C_{k}^{\sf f}\}\}$.
    \smallskip
    \item[$(vi)$] For any extra-logical rule
    \smallskip
    \begin{center}
        {\AxiomC{$\{\Gamma\specialvdashdef{{\bf S_{i}}}{{\bf T_{i}}}\Theta_{i}\}_{1\leq i\leq m}$}
        \RightLabel{$\delta$}
        \UnaryInfC{$\Gamma\specialvdashdef{{\bf S}}{{\bf T}}\Phi$}
        \DisplayProof}
    \end{center}
    \smallskip
    with $m\geq0$, if $\quest\Phi'$ occurs in $\topfstaro(\quest(\bigvee\Phi)\wedge W)$ without belonging to $\topfstaro(\quest W)$, then there is an extra-logical rule of the form:
    \smallskip
    \begin{center}
        {\AxiomC{$\{\Gamma\specialvdashdef{{\bf S_{i}}}{{\bf T_{i}}}\Theta_{i}\}_{1\leq i\leq m}$}
        \RightLabel{$\delta$}
        \UnaryInfC{$\Gamma\specialvdashdef{{\bf S}}{{\bf T}}\Phi'$}
        \DisplayProof}
    \end{center}
    \smallskip
    \item[$(vii)$] For any pair of extra-logical rules
    \smallskip
    \begin{center}
        {\AxiomC{$\{\Gamma\specialvdashdef{{\bf S_{i}'}}{{\bf T_{i}'}}\Theta_{i}\}_{1\leq i\leq m}$}
        \RightLabel{$\delta$}
        \UnaryInfC{$\Gamma\specialvdashdef{{\bf S_{1}}}{{\bf T_{1}}}\Phi$}
        \DisplayProof}\quad
        {\AxiomC{$\{\Gamma\specialvdashdef{{\bf S_{i}'}}{{\bf T_{i}'}}\Theta_{i}\}_{m+1\leq i\leq n}$}
        \RightLabel{$\delta$}
        \UnaryInfC{$\Gamma\specialvdashdef{{\bf S_{2}}}{{\bf T_{2}}}\Phi'$}
        \DisplayProof}
    \end{center}
    \smallskip
    with $m+n\geq0$, if $\quest\Psi$ occurs in $\topfstaro(\quest(\bigvee\Phi)\wedge(\bigvee\Phi'))$ without belonging to $\topfstaro(\quest(\bigvee\Phi))$, $\topfstaro(\quest(\bigvee\Phi'))$ or $\topfstaro(\quest W)$, then there exists an extra-logical rule of the  form:
    \smallskip
    \begin{center}
        {\AxiomC{$\{\Gamma\specialvdashdef{{\bf S_{i}'}}{{\bf T_{i}'}}\Theta_{i}\}_{1\leq i\leq n}$}
        \RightLabel{$\delta$}
       \UnaryInfC{$\Gamma\specialvdashdef{{\bf S}}{{\bf T}}\Psi$}
        \DisplayProof}
    \end{center}
    \smallskip
    with ${\bf S}={\bf S_{1}}\cup{\bf S_{2}}$ and ${\bf T}={\bf T_{1}}\cup{\bf T_{2}}$.
\end{itemize}
\end{definition}

Condition $(i)$  ensures that the $\Gastc$ calculus incorporates the extra-logical axioms from $\mathcal{W}$, while condition $(v)$ guarantees the presence of extra-logical rules corresponding to the default rules in $\mathcal{D}$. In turn, closure conditions $(vi)$  and $(vii)$  are required to ensure that the conclusions of extra-logical rules are closed under Contraction and Cut (cf. Lemmas \ref{rightcontraction} and \ref{leftcontraction}, and Theorems \ref{cutempty}, \ref{atomicnonanalytic}, \ref{cut}, \ref{wccut}).

It is easy to find cases in which there exists a $\Gastc$-derivation of a controlled sequent $\Gamma\specialvdashdef{{\bf S}}{{\bf T}}\Delta$ even though $\bigvee\Delta$ does not belong to any modified extension of $\langle\mathcal{W}\cup\Gamma,\mathcal{D}\rangle$.

\begin{example}
Let $\langle\mathcal{W},\mathcal{D}\rangle$ be defined as follows: $\mathcal{W}=\varnothing$ and $\mathcal{D}=\Big\{\dfrac{\top:p}{p}\Big\}$. Consider the  $\Gastc$-derivation:
\smallskip
\begin{center}
    {\AxiomC{$ $}
    \RightLabel{\scriptsize{$\delta$}}
    \UnaryInfC{$\neg p\specialvdashdef{{\bf S}}{{\bf T}}p$}
    \DisplayProof}
\end{center}
\smallskip
with ${\bf T}=\{\{\top^{\sf f}\}\}$ and ${\bf S}=\{\{\neg p^{\sf f}\}\}$. It is easy to check that $p$ does not belong to any modified extension of $\langle\mathcal{W}\cup\{\neg p\},\mathcal{D}\rangle$. 
\end{example}

As a result, we need to offer a criterion to single out $\Gastc$-derivations which deliver $m$-credulous consequences of $\mathcal{W}\cup\Gamma$ from those which do not. To this aim, we make some additions to our terminological and conceptual apparatus.

For any extra-logical rule $\delta$ in $\Gastc$, we define the label $\delta_{\mathcal{D}'}$ as follows:
\begin{itemize}
\item[$(a)$] if $\delta$ is generated in accordance with point $(iii)$ in Definition \ref{definitioncalculi}, then $\mathcal{D}'=\Big\{\dfrac{B:C_{1},\ldots,C_{k}}{D}\Big\}$;
\smallskip
\item[$(b)$] if $\delta$ is generated from extra-logical rules $\delta'$ and $\delta''$ with labels $\delta_{\mathcal{D}_{1}}$ and $\delta_{\mathcal{D}_{2}}$, respectively, in accordance with points $(iv)-(v)$ in Definition \ref{definitioncalculi}, then $\mathcal{D}'=\mathcal{D}_{1}\cup\mathcal{D}_{2}$.
\end{itemize}
\smallskip
For each $\G$-derivation $\pi$ we say that a default rule $\dfrac{B:C_{1},\ldots,C_{k}}{D}$ {\em belongs to $def(\pi)$} if and only if there is (at least) one extra-logical rule labelled $\delta_{\mathcal{D}'}$ which is applied in $\pi$ and such that $\dfrac{B:C_{1},\ldots,C_{k}}{D}$ belongs to $\mathcal{D}'$. Moreover, we employ $def'(\pi)$ to denote $def(\pi_{1})\cup\cdots\cup def(\pi_{m})$ whenever $\pi_{1},\ldots,\pi_{m}$ are the immediate subderivations yielding the premises of the lowermost extra-logical rules applications in $\pi$. 

We shall use $D_{\pi}$ to refer to the conjunction of the formulas in $concl(def(\pi))$, and $E_{\pi}$ to denote the conjunction of the formulas in $concl(def'(\pi))$.

\begin{definition}\label{soundness}
    Let $\pi$ be a $\Gastc$-derivation of $\Gamma\specialvdashdef{{\bf S}}{{\bf T}}\Delta$. The sequent $\Gamma\specialvdashdef{{\bf S}}{{\bf T}}\Delta$ is
    \smallskip
    \begin{itemize}
        \item[$(i)$] {\em sound under conditions} if and only if $([W,E_{\rho}]\cup\Gamma)\, ||\, \langle{\bf V},\varnothing\rangle$ for any subderivation $\rho$ of $\pi$ with $\Pi\specialvdashdef{{\bf U}}{{\bf V}}\Sigma$ as conclusion;
        \smallskip
        \item[$(ii)$] {\em sound under constraints} if and only if $([W,D_{\pi}]\cup\Gamma)\, ||\, \langle\varnothing,{\bf S}\rangle$;
        \smallskip
        \item[$(iii)$] {\em sound} if and only if it is sound under conditions and sound under constraints.
    \end{itemize}
\end{definition}

\begin{definition}\label{proof}
    Let $\pi$ be a $\Gastc$-derivation of $\Gamma\specialvdashdef{{\bf S}}{{\bf T}}\Delta$. Then, $\pi$ is a {\em proof} of $\Gamma\specialvdashdef{{\bf S}}{{\bf T}}\Delta$ if and only if any sequent $\Pi\specialvdashdef{{\bf S'}}{{\bf T'}}\Sigma$ occurring in $\pi$ is sound -- and a {\em paraproof}, otherwise.
\end{definition}

\begin{example}
    Let $\mathcal{W}=\varnothing$ and $\mathcal{D}=\Big\{\dfrac{p:q}{r}\,;\,\dfrac{q:s}{t}\Big\}$. The following derivation $\pi$ in the $\Gastc$ calculus for $\langle\mathcal{W},\mathcal{D}\rangle$ is a paraproof:
    \smallskip
    \begin{center}
        {\AxiomC{$ $}
        \LeftLabel{\scriptsize{$ax$}}
        \UnaryInfC{$p\specialvdashdef{\varnothing}{\varnothing}p$}
        \LeftLabel{\scriptsize{$\delta_{\mathcal{D}_{1}}$}}
        \UnaryInfC{$p\specialvdashdef{{\bf S_{1}}}{{\bf T_{1}}}r$}
        \LeftLabel{\scriptsize{$RW$}}
        \UnaryInfC{$p\specialvdashdef{{\bf S_{1}}}{{\bf T_{1}}}r,t$}
        \AxiomC{$ $}
        \RightLabel{\scriptsize{$ax$}}
        \UnaryInfC{$q\specialvdashdef{\varnothing}{\varnothing}q$}
        \RightLabel{\scriptsize{$\delta_{\mathcal{D}_{2}}$}}
        \UnaryInfC{$q\specialvdashdef{{\bf S_{2}}}{{\bf T_{2}}}t$}
        \RightLabel{\scriptsize{$RW$}}
        \UnaryInfC{$q\specialvdashdef{{\bf S_{2}}}{{\bf T_{2}}}r,t$}
        \RightLabel{\scriptsize{$L\vee$}}
        \BinaryInfC{$p\vee q \specialvdashdef{{\bf S}}{{\bf T}}r,t$}
        \DisplayProof}
    \end{center}
    \smallskip
    with ${\bf T_{1}}=\{\{p^{\sf f}\}\}$, ${\bf S_{1}}=\{\{\neg q^{\sf f}\}\}$, ${\bf T_{2}}=\{\{q^{\sf f}\}\}$ and ${\bf S_{2}}=\{\{\neg s^{\sf f}\}\}$. On the other hand, the immediate subderivations of $\pi$ are proofs.
\end{example}

\begin{example}
Let $\mathcal{W}=\varnothing$ and $\mathcal{D}$ be defined as $\Big\{\dfrac{p:q}{r}\,;\,\dfrac{q:s}{\neg r}\Big\}$. The following derivation $\pi$ in the $\Gastc$ calculus for $\langle\mathcal{W},\mathcal{D}\rangle$ is a paraproof:
\smallskip
    \begin{center}
        {\AxiomC{$ $}
        \LeftLabel{\scriptsize{$ax$}}
        \UnaryInfC{$p\specialvdashdef{\varnothing}{\varnothing}p$}
        \LeftLabel{\scriptsize{$LW$}}
        \UnaryInfC{$p,q\specialvdashdef{\varnothing}{\varnothing}p$}
        \LeftLabel{\scriptsize{$\delta_{\mathcal{D}_{1}}$}}
        \UnaryInfC{$p,q\specialvdashdef{{\bf S_{1}}}{{\bf T_{1}}}r$}
        \AxiomC{$ $}
        \RightLabel{\scriptsize{$ax$}}
        \UnaryInfC{$q\specialvdashdef{\varnothing}{\varnothing}q$}
        \RightLabel{\scriptsize{$LW$}}
        \UnaryInfC{$p,q\specialvdashdef{\varnothing}{\varnothing}q$}
        \RightLabel{\scriptsize{$\delta_{\mathcal{D}_{2}}$}}
        \UnaryInfC{$p,q\specialvdashdef{{\bf S_{2}}}{{\bf T_{2}}}\neg r$}
        \RightLabel{\scriptsize{$R\wedge$}}
        \BinaryInfC{$p,q\specialvdashdef{{\bf S}}{{\bf T}}r\wedge\neg r$}
        \DisplayProof}
    \end{center}
    \smallskip
    with ${\bf T_{1}}=\{\{p^{\sf f}\}\}$, ${\bf S_{1}}=\{\{\neg q^{\sf f}\}\}$, ${\bf T_{2}}=\{\{q^{\sf f}\}\}$ and ${\bf S_{2}}=\{\{\neg s^{\sf f}\}\}$.
On the other hand, the immediate subderivations of $\pi$ are proofs.
\end{example}

Before turning to the structural analysis of $\Gastc$ calculi, let us record a few additional remarks concerning the conditions in Definition \ref{definitioncalculi}.
\smallskip
\begin{itemize}
   \item[$(A)$] For any extra-logical rule, the formulas that are principal in its applications may occur only on the right-hand side of the sequent. Replacing $\topfstaro$ with $\topfstar$ would therefore lead to unsound formalizations of default rules. Consider, for example, the default rule
\[
\dfrac{\top : \neg p}{\neg p}.
\]
The corresponding extra-logical rule would be
\smallskip
\begin{center}
{\AxiomC{}\UnaryInfC{$\Gamma, p\specialvdashdef{\mathbf{S}}{\mathbf{T}}$}\DisplayProof}
\end{center}
\smallskip
with $\mathbf{T}=\{\{\top^{\sf f}\}\}$ and $\mathbf{S}=\{\{\neg\neg p^{\sf f}\}\}$. For any context $\Gamma$, this rule produces an unsound conclusion: adopting it would make every derivation of $\specialvdashdef{\mathbf{S}}{\mathbf{T}}\neg p$ have a paraproof as an immediate subderivation. The use of $\topfstaro$ is legitimate only because we work in the calculus $\fullGpn$ for classical logic, where negation is primitive and no rule moves a formula between the two sides of a sequent.
    \smallskip
    \item[$(B)$] We use the term {\em semi-analytic Cut} for Cut applications where (a) both Cut formulas occur on the right-hand side, and (b) the Cut formula is a subformula of the formulas on the left-hand side of the conclusion or of the extra-logical axioms. Condition $(iii)$ guarantees that $\Gastc$ calculi contain atomic, semi-analytic Cut as a primitive rule. This is required for strong adequacy (cf. Theorem \ref{adequacy}), since $\Gastc$ calculi without $cut_{asa}$ are not closed under such applications. Consider e.g. the $\Gastc$ calculus for $\langle\mathcal{W},\mathcal{D}\rangle$, with
    \smallskip
    \begin{center}
        $\mathcal{W}=\varnothing$
    \end{center}
    \smallskip
    \begin{center}
        $\mathcal{D}=\Big\{\dfrac{\top:p}{p}\,,\,\dfrac{p\vee q:\neg p}{\neg p}\Big\}$
    \end{center}
    \smallskip
    Now, take the  $\Gastc$-proof $\pi$:
    \smallskip
\begin{center}
    {\AxiomC{$ $}
    \LeftLabel{\scriptsize{$ax$}}
    \UnaryInfC{$p\specialvdash{\varnothing}{\varnothing}p$}
    \LeftLabel{\scriptsize{$RW$}}
    \UnaryInfC{$p\specialvdash{\varnothing}{\varnothing}p,q$}
    \AxiomC{$ $}
    \LeftLabel{\scriptsize{$ax$}}
    \UnaryInfC{$q\specialvdash{\varnothing}{\varnothing}q$}
    \LeftLabel{\scriptsize{$RW$}}
    \UnaryInfC{$q\specialvdash{\varnothing}{\varnothing}p,q$}
    \LeftLabel{\scriptsize{$L\vee$}}
    \BinaryInfC{$p\vee q\specialvdash{\varnothing}{\varnothing}p,q$} 
    \AxiomC{$ $}
    \RightLabel{\scriptsize{$ax$}}
    \UnaryInfC{$p\specialvdash{\varnothing}{\varnothing}p$}
    \RightLabel{\scriptsize{$RW$}}
    \UnaryInfC{$p\specialvdash{\varnothing}{\varnothing}p,q$}
    \AxiomC{$ $}
    \RightLabel{\scriptsize{$ax$}}
    \UnaryInfC{$q\specialvdash{\varnothing}{\varnothing}q$}
    \RightLabel{\scriptsize{$RW$}}
    \UnaryInfC{$q\specialvdash{\varnothing}{\varnothing}p,q$}
    \RightLabel{\scriptsize{$L\vee$}}
    \BinaryInfC{$p\vee q\specialvdash{\varnothing}{\varnothing}p,q$}
    \RightLabel{\scriptsize{$\delta_{\mathcal{D}''}$}}
    \UnaryInfC{$p\vee q\specialvdash{\{\neg\neg p\}}{\{p\vee q\}}\neg p$}
    \RightLabel{\scriptsize{$cut_{asa}$}}
    \BinaryInfC{$p\vee q\specialvdash{\{\neg\neg p\}}{\{p\vee q\}}q$}
    \DisplayProof}
\end{center}
    \smallskip
    It is easy to verify that closure conditions $(vi)$-$(vii)$ in Definition \ref{definitioncalculi} are insufficient to absorb this Cut application. 
    \smallskip
    \item[$(C)$] Let us call a semi-analytic Cut {\em analytic} when the Cut formula is a subformula of the formulas on the left-hand side of the conclusion. If $\topfstar$ in condition (i) of Definition \ref{definitioncalculi} is replaced by $\topfstaro$, one can show that every non-analytic instance of $cut_{asa}$ can be eliminated. By Theorems \ref{atomicanalytic}, \ref{atomicnonanalytic}, and \ref{cut}, it follows that $\Gastc$ calculi are non-analytic cut-free complete. To eliminate analytic Cuts, one might consider imposing the following closure condition on the set of extra-logical rules. For any rule
\[
\AxiomC{$\{\Gamma \specialvdashdef{\mathbf{S_i}}{\mathbf{T_i}} \Theta_i\}_{1 \leq i \leq m}$}
\UnaryInfC{$\Gamma \specialvdashdef{\mathbf{S}}{\mathbf{T}} \Phi$}
\DisplayProof
\]
if a formula $\quest\Phi'$ occurs in $\topfstaro(\quest(\bigvee\Phi)\wedge \bigwedge\Gamma)$ but not in $\topfstaro(\quest\bigvee\Phi)$ nor in $\topfstaro(\quest\bigwedge\Gamma)$, then one would also require the presence of the rule
\[
\AxiomC{$\{\Gamma \specialvdashdef{\mathbf{S_i}}{\mathbf{T_i}} \Theta_i\}_{1 \leq i \leq m}$}
\UnaryInfC{$\Gamma \specialvdashdef{\mathbf{S}}{\mathbf{T}} \Phi'$}
\DisplayProof
\]
Yet this closure condition would generate an \emph{infinite} collection of extra-logical rules, a situation we regard as highly undesirable.
\end{itemize}

\subsection{Structural properties} 

In this subsection, we carry out a structural analysis of the $\Gastc$ calculi for default theories in order to establish strong adequacy with respect to $m$-credulous consequence. To this end, we introduce some terminology.

A rule application is said to be \emph{safe} when the existence of proofs for its premises guarantees the existence of a proof for its conclusion. A rule is \emph{admissible} whenever
\smallskip
 \begin{itemize}
\item[$(i)$] all its instances are safe;
\smallskip
\item[$(ii)$] the set of conditions associated with the conclusion is included in the union of the sets of conditions associated with the premises; and
\smallskip
\item[$(iii)$]the set of constraints of the premises coincides with that of the conclusion.
\end{itemize}
\smallskip
Moreover, a rule is said to be \emph{invertible} when each inverse version of the rule is admissible.

\begin{lemma}\label{rightinvertibility}
    The rules $R\wedge,R\vee,R\neg\wedge,R\neg\vee$ and $R\neg\neg$ are invertible in $\Gastc$.
\end{lemma}

\begin{proof}
    We reason by induction on the height of a $\Gastc$-proof $\pi$ of the premise(s). Let us focus on the rule $R\wedge$, since the others are analogous. Moreover, let's take $B\wedge C$ to be the active formula in the premise of the inverse of $R\wedge$.
    
    If $h(\pi)=1$, the claim holds vacuously. Otherwise, we reason by cases over the last rule of $\pi$. If the last rule is  $R\wedge$ with principal formula $B\wedge C$, we take the $\Gastc$-proofs of the premises of the latter, possibly followed by $\sigma$ steps. Notice that the height may increase, but the $\sigma$-steps preserve correctness, so the result is still a proof. If the last rule is $RW$ with principal $B\wedge C$, we replace that weakening with one yielding either $B$ or $C$ to obtain the desired conclusion. If $B\wedge C$ is not principal in the last rule, apply the inductive hypothesis to the premises to obtain $\Gastc$-derivations of the conclusions of the inverse rules. Applying the hypothesis above the last sequent in $\pi$ preserves soundness, so each resulting $\Gastc$-derivation is again a proof, as required.
\end{proof}

\begin{lemma}\label{rightcontraction}
    The rule of Right Contraction
    \smallskip
    \begin{center}
        {\AxiomC{$\Gamma\specialvdashdef{{\bf S}}{{\bf T}}\Delta,A,A$}
        \RightLabel{$RC$}
        \UnaryInfC{$\Gamma\specialvdashdef{{\bf S}}{{\bf T'}}\Delta,A$}
        \DisplayProof}
    \end{center}
    \smallskip
    is admissible in $\Gastc$.
\end{lemma}

\begin{proof}
    We proceed by primary induction on the height of the $\Gastc$-proof $\pi$ of the premise and secondary induction on the logical complexity of $A$. If $h(\pi)=1$, conditions $(i)$ and $(v)$ of Definition \ref{definitioncalculi} yield the conclusion, since $\topfstar(\quest B)$ is closed under Contraction for every formula $B$. Otherwise, we analyze the last rule of $\pi$. If $A$ is atomic and both occurrences are principal, we apply conditions $(v)$ - $(vii)$ of Definition \ref{definitioncalculi}. If exactly one atomic occurrence of $A$ is principal, the latter must be introduced by $RW$, and taking the premise of that $RW$ instance suffices. If neither occurrence is principal, we apply the inductive hypothesis to the premises to permute $RC$ upwards; this remains possible even in the presence of $cut_{asa}$, due to its context-sharing formulation. Soundness is preserved, as the left-hand side of the sequents is unchanged. If $A$ is non-atomic, the only non-trivial cases occur when one occurrence of $A$ is principal in the last rule. To handle these, we use Lemma \ref{rightinvertibility}: since right-rule invertibility does not preserve height, the secondary inductive hypothesis is required to obtain the conclusion.
\end{proof}

\begin{lemma}\label{leftinvertibility}
    The rules $L\wedge,L\neg\vee$ and $L\neg\neg$ are invertible in $\Gastc$.
\end{lemma}

\begin{proof}
    We reason by induction on the height of a $\Gastc$-proof $\pi$ of the premise. We argue as in the proof of Lemma \ref{rightinvertibility}.
\end{proof}

\begin{lemma}\label{invertibilityempty}
    The logical rules of $\Gastc$ are height-preserving invertible under the empty set of constraints.
\end{lemma}

\begin{proof}
    We reason by induction on the height of a $\Gastc$-proof $\pi$ of the premise. We argue as in the proof of Lemma \ref{rightinvertibility}: since we deal with empty sets of constraints and conditions, the height of the $\Gastc$-proof of the conclusion of the inverse rule does not surpass the height of the $\Gastc$-proofs of the premise.
\end{proof}

\begin{lemma}\label{contractionempty}
    The rule of Left Contraction
    \smallskip
    \begin{center}
        {\AxiomC{$A,A,\Gamma\specialvdashdef{{\bf S}}{{\bf T}}\Delta$}
        \LeftLabel{$LC$}
        \UnaryInfC{$A,\Gamma\specialvdashdef{{\bf S}}{{\bf T}}\Delta$}
        \DisplayProof}
    \end{center}
    \smallskip
    is admissible in $\Gastc$ under the empty set of constraints. 
\end{lemma}

\begin{proof}
    We reason by induction on the height of the $\Gastc$-proof $\pi$ of the premise. As usual, we exploit Lemma \ref{invertibilityempty} whenever one of the two occurrences of the contracted formula is principal in the last rule  in $\pi$ and the latter is a logical rule. Notice that we do not need secondary induction over the logical complexity of $A$, due to the fact invertibility is height-preserving under the empty sets of constraints.
\end{proof}

    \begin{theorem}\label{cutempty}
    The rule of Cut
    \smallskip
\begin{center}
    {\AxiomC{$\Gamma\specialvdashdef{\varnothing}{\varnothing}A,\Delta$}
    \AxiomC{$A,\Pi\specialvdashdef{\varnothing}{\varnothing}\Sigma$}
    \RightLabel{$Cut$}
    \BinaryInfC{$\Pi,\Gamma\specialvdashdef{\varnothing}{\varnothing}\Delta,\Sigma$}
    \DisplayProof}
\end{center}
\smallskip
    is admissible in $\Gastc$.
\end{theorem}

\begin{proof}
    We argue as in \cite{PiazzaTesi24} and \cite{defaultACM}.
\end{proof}

\begin{proposition}\label{negation}
    The Left and Right Negation rules
    \smallskip
    \begin{center}
        {\AxiomC{$\Gamma\specialvdashdef{\varnothing}{\varnothing}\Delta,A$}
        \LeftLabel{$L\neg$}
        \UnaryInfC{$\neg A,\Gamma\specialvdashdef{\varnothing}{\varnothing}\Delta$}
        \DisplayProof}\qquad
        {\AxiomC{$A,\Gamma\specialvdashdef{\varnothing}{\varnothing}\Delta$}
        \RightLabel{$R\neg$}
        \UnaryInfC{$\Gamma\specialvdashdef{\varnothing}{\varnothing}\Delta,\neg A$}
        \DisplayProof}
    \end{center}
    \smallskip
    are admissible and invertible in $\Gastc$.
\end{proposition}

\begin{proof}
    Straightforward from Theorem \ref{cutempty}.\remove{We leverage Theorem \ref{cutempty} as follows:
    \smallskip
    \begin{center}
        {\AxiomC{$\Gamma\specialvdashdef{\varnothing}{\varnothing}\Delta,A$}
        \AxiomC{$\vdots$}
        \noLine
        \UnaryInfC{$A,\neg A\specialvdashdef{\varnothing}{\varnothing}$}
        \LeftLabel{\scriptsize{$Cut$}}
        \BinaryInfC{$\neg A,\Gamma\specialvdashdef{\varnothing}{\varnothing}\Delta$}
        \DisplayProof}\quad
        {\AxiomC{$\vdots$}
        \noLine
        \UnaryInfC{$\specialvdashdef{\varnothing}{\varnothing}A,\neg A$}
        \AxiomC{$A,\Gamma\specialvdashdef{\varnothing}{\varnothing}\Delta$}
        \RightLabel{\scriptsize{$Cut$}}
        \BinaryInfC{$\Gamma\specialvdashdef{\varnothing}{\varnothing}\Delta,\neg A$}
        \DisplayProof}
    \end{center}}
\end{proof}

\begin{theorem}\label{completenessbase}
    $\Gastc$ proves $\Gamma\specialvdashdef{\varnothing}{\varnothing}\Delta$ if and only if $\bigvee\Delta$ belongs to $Cn(\mathcal{W}\cup\Gamma)$.
\end{theorem}

\begin{proof}
We argue as in \cite{PiazzaTesi24} and \cite{defaultACM}. 
\end{proof}

\begin{example}
    Let $\langle\mathcal{W},\mathcal{D}\rangle$ be defined as follows:
    \smallskip
    \begin{center}
        $\mathcal{W}=\varnothing$
    \end{center}
    \smallskip
    \begin{center}
        $\mathcal{D}=\Big\{\dfrac{\top:p}{p}\,,\,\dfrac{p\vee q:\neg p}{\neg p}\Big\}$
    \end{center}
    \smallskip
    Consider $\Gastc$-derivation $\pi$:
    \smallskip
\begin{center}
    {\AxiomC{$ $}
    \LeftLabel{\scriptsize{$\delta_{\mathcal{D}'}$}}
    \UnaryInfC{$\specialvdash{\{\neg p\}}{\top}p,q $}
    \LeftLabel{\scriptsize{$LW$}}
    \UnaryInfC{$p\vee q\specialvdash{\{\neg p\}}{\top}p,q $}
    \AxiomC{$ $}
    \RightLabel{\scriptsize{$ax$}}
    \UnaryInfC{$p\specialvdash{\varnothing}{\varnothing} p$}
    \RightLabel{\scriptsize{$RW$}}
    \UnaryInfC{$p\specialvdash{\varnothing}{\varnothing} p,q$}
    \AxiomC{$ $}
    \RightLabel{\scriptsize{$ax$}}
    \UnaryInfC{$q\specialvdash{\varnothing}{\varnothing}q$}
    \RightLabel{\scriptsize{$RW$}}
    \UnaryInfC{$q\specialvdash{\varnothing}{\varnothing} p,q$}
    \RightLabel{\scriptsize{$L\vee$}}
    \BinaryInfC{$p\vee q\specialvdash{\varnothing}{\varnothing} p,q$}
    \RightLabel{\scriptsize{$\delta_{\mathcal{D}''}$}}
    \UnaryInfC{$p\vee q\specialvdash{\{\neg\neg p\}}{\{p\vee q\}}\neg p$}
    \RightLabel{\scriptsize{$cut_{asa}$}}
    \BinaryInfC{$p\vee q\specialvdash{\{\neg p,\neg\neg p\}}{\{\top,\ p\vee q\}}q $}
    \DisplayProof}
\end{center}
\smallskip
By Definition \ref{proof}, $\pi$ is a $\Gastc$-paraproof: the conclusion is sound with respect to conditions but unsound with respect to constraints. Hence, the $cut_{asa}$ application in $\pi$ is unsafe. On the other hand, consider the following $\Gastc$-derivation $\rho$:
\smallskip
\begin{center}
    {\AxiomC{$ $}
    \LeftLabel{\scriptsize{$ax$}}
    \UnaryInfC{$p\specialvdash{\varnothing}{\varnothing}p$}
    \LeftLabel{\scriptsize{$RW$}}
    \UnaryInfC{$p\specialvdash{\varnothing}{\varnothing}p,q$}
    \AxiomC{$ $}
    \LeftLabel{\scriptsize{$ax$}}
    \UnaryInfC{$q\specialvdash{\varnothing}{\varnothing}q$}
    \LeftLabel{\scriptsize{$RW$}}
    \UnaryInfC{$q\specialvdash{\varnothing}{\varnothing}p,q$}
    \LeftLabel{\scriptsize{$L\vee$}}
    \BinaryInfC{$p\vee q\specialvdash{\varnothing}{\varnothing}p,q$} 
    \AxiomC{$ $}
    \RightLabel{\scriptsize{$ax$}}
    \UnaryInfC{$p\specialvdash{\varnothing}{\varnothing}p$}
    \RightLabel{\scriptsize{$RW$}}
    \UnaryInfC{$p\specialvdash{\varnothing}{\varnothing}p,q$}
    \AxiomC{$ $}
    \RightLabel{\scriptsize{$ax$}}
    \UnaryInfC{$q\specialvdash{\varnothing}{\varnothing}q$}
    \RightLabel{\scriptsize{$RW$}}
    \UnaryInfC{$q\specialvdash{\varnothing}{\varnothing}p,q$}
    \RightLabel{\scriptsize{$L\vee$}}
    \BinaryInfC{$p\vee q\specialvdash{\varnothing}{\varnothing}p,q$}
    \RightLabel{\scriptsize{$\delta_{\mathcal{D}''}$}}
    \UnaryInfC{$p\vee q\specialvdash{\{\neg\neg p\}}{\{p\vee q\}}\neg p$}
    \RightLabel{\scriptsize{$cut_{asa}$}}
    \BinaryInfC{$p\vee q\specialvdash{\{\neg\neg p\}}{\{p\vee q\}}q$}
    \DisplayProof}
\end{center}
\smallskip
By Definition \ref{proof}, $\rho$ is a $\Gastc$-proof. Unlike the $cut_{asa}$ application in $\pi$, the $cut_{asa}$ application in $\rho$ is safe. 
\end{example}

\begin{theorem}\label{atomicanalytic}
    Let $A$ be an atom occurring in some formula in $\Gamma,\Pi,\mathcal{W}$. The rule of safe, atomic semi-analytic Cut
    \smallskip
    \begin{center}
        {\AxiomC{$\Gamma\specialvdashdef{{\bf S_{1}}}{{\bf T_{1}}}\Delta,A$}
        \AxiomC{$\Pi\specialvdashdef{{\bf S_{2}}}{{\bf T_{2}}}\Sigma,\neg A$}
        \RightLabel{$cut_{a}$}
        \BinaryInfC{$\Pi,\Gamma\specialvdashdef{{\bf S}}{{\bf T}}\Delta,\Sigma$}
        \DisplayProof}
    \end{center}
    \smallskip
    is admissible in $\Gastc$. 
\end{theorem}

\begin{proof}
    For any such Cut application, consider the following $\Gastc$-derivation:
    \smallskip
    \begin{center}
        {\AxiomC{$\Gamma\specialvdashdef{{\bf S_{1}}}{{\bf T_{1}}}\Delta,A$}
        \doubleLine
        \LeftLabel{\scriptsize{$LW$}}
        \UnaryInfC{$\Pi,\Gamma\specialvdashdef{{\bf S_{1}}}{{\bf T_{1}}}\Delta,A$}
        \doubleLine
        \LeftLabel{\scriptsize{$RW$}}
        \UnaryInfC{$\Pi,\Gamma\specialvdashdef{{\bf S_{1}}}{{\bf T_{1}}}\Delta,\Sigma,A$}
        \AxiomC{$\Pi\specialvdashdef{{\bf S_{2}}}{{\bf T_{2}}}\Sigma,\neg A$}
        \doubleLine
        \RightLabel{\scriptsize{$LW$}}
        \UnaryInfC{$\Pi,\Gamma\specialvdashdef{{\bf S_{2}}}{{\bf T_{2}}}\Sigma,\neg A$}
        \doubleLine
        \RightLabel{\scriptsize{$RW$}}
        \UnaryInfC{$\Pi,\Gamma\specialvdashdef{{\bf S_{2}}}{{\bf T_{2}}}\Delta,\Sigma,\neg A$}
        \RightLabel{\scriptsize{$cut_{asa}$}}
        \BinaryInfC{$\Pi,\Gamma\specialvdashdef{{\bf S}}{{\bf T}}\Delta,\Sigma$}
        \DisplayProof}
    \end{center}
    \smallskip
    By hypothesis, all $LW$ applications are safe: hence, such $\Gastc$-derivation is a proof.
\end{proof}

\begin{definition}
    Let $\pi$ be a $\Gastc$-derivation of $\Gamma\specialvdash{{\bf S}}{{\bf T}}\Delta$ and $A$ be a formula occurrence in $\Gamma$ ($\Delta$). The {\em predecessors tree} of $A$ in $\pi$, in symbols $\mathcal{T}_{\pi}(A)$ is a labelled tree generated according to the following procedure.
    \smallskip
    \begin{itemize}
        \item[$(i)$] The root of $\mathcal{T}_{\pi}(A)$ is labelled by the occurrence of $A$ in $\Gamma$ ($\Delta$).
        \smallskip
        \item[$(ii)$] Let $\nu$ be a node of $\mathcal{T}_{\pi}(A)$ is labelled by an occurrence of $A$ in the antecedent $\Pi$ (succedent $\Sigma$, respectively) of a sequent $\Pi\specialvdash{{\bf S'}}{{\bf T'}}\Sigma$ occurring in $\pi$, and $\star$ be the last rule  in the $\Gastc$-subderivation of $\pi$ concluding $\Pi\specialvdash{{\bf S'}}{{\bf T'}}\Sigma$. 
        \smallskip
        \begin{itemize}
            \item[$(ii.i)$] If the occurrence of $A$ is principal in $\star$, $\nu$ has no immediate children.
            \smallskip
            \item[$(ii.ii)$] If the occurrence of $A$ is not principal in $\star$, $\nu$ has (at most) two immediate children, and their labels are the instances of $A$ occurring in the antecedents (succedents, respectively) of the premises of $\star$ which correspond to the occurrence of $A$.
        \end{itemize}
    \end{itemize}
\end{definition}

\begin{theorem}\label{atomicnonanalytic}
    Let $A$ be an atom which does not occur in any formula in $\Gamma,\Pi,\mathcal{W}$. The rule of safe, atomic non-semi-analytic Cut
    \smallskip
    \begin{center}
        {\AxiomC{$\Gamma\specialvdashdef{{\bf S_{1}}}{{\bf T_{1}}}\Delta,A$}
        \AxiomC{$\Pi\specialvdashdef{{\bf S_{2}}}{{\bf T_{2}}}\Sigma,\neg A$}
        \RightLabel{$cut_{a}$}
        \BinaryInfC{$\Pi,\Gamma\specialvdashdef{{\bf S}}{{\bf T}}\Delta,\Sigma$}
        \DisplayProof}
    \end{center}
    \smallskip
    is admissible in $\Gastc$. 
\end{theorem}

\begin{proof}
    We focus on the topmost safe, atomic non-semi-analytic Cut application, reasoning by induction on the sum of the heights of the $\Gastc$-proofs $\pi_{1},\pi_{2}$ of the left and right premise, respectively. In the base case, $h(\pi_{1})=h(\pi_{2})=1$: since $A$ does not occur in any formula in $\mathcal{W}$, it cannot be the case that last rule  in $\pi_{1},\pi_{2}$ is an extra-logical instance of $ax$. Moreover, if the last rule in $\pi_{1}$ or $\pi_{2}$ is a logical instance of $ax$, we leverage suitable $LW$ and $cut_{asa}$ applications to conclude. Hence, the only non-trivial scenario arises whenever the last rules in $\pi_{1}$ and $\pi_{2}$ are extra-logical ones: we leverage condition $(vii)$ of Definition \ref{definitioncalculi} to conclude. To prove the induction step, we reason by cases over the last rule in $\pi_{1}$. 
    
    If $A$ is principal in the last rule of $\pi_{1}$, the only non-trivial cases occur when that last rule is extra-logical. To complete the proof, we then distinguish subcases according to the last rule in $\pi_{2}$.

    \smallskip
    \begin{itemize}
        \item[$(a)$] If $\neg A$ is principal in the last rule  in $\pi_{2}$, the latter is either $RW$ or an extra-logical rule $\delta$. In the first scenario, we apply $LW$, $RW$ and $\sigma$ to obtain a $\Gastc$-derivation of $\Pi,\Gamma\specialvdashdef{{\bf S}}{{\bf T_{2}}}\Delta,\Sigma$: it is easy to check that this derivation is a $\Gastc$-proof. In the second scenario, we leverage condition $(vii)$ of Definition \ref{definitioncalculi} to reach the conclusion: again, the application of $\delta$ is safe by construction. 
        \smallskip
        \item[$(b)$] If $\neg A$ is not principal in the last rule  in $\pi_{2}$ and latter is neither $L\vee$ nor $L\neg\wedge$, we simply apply the inductive hypothesis to its premises to permute upwards the Cut application: notice that soundness under constraints is always preserved. If the last rule  in $\pi_{2}$ is either $L\vee$ or $L\neg\wedge$, the same upwards permutation of Cut may fail, due to the fact that soundness under constraints may not be preserved. For instance, consider the following configuration -- where $B^{f}$ belongs to $\bigcup{\bf S_{1}}$, $\Pi=\Pi',B\vee C$, ${\bf T_{2}}={\bf V_{1}}\cup{\bf V_{2}}$ and ${\bf S_{2}}={\bf U_{1}}\cup{\bf U_{2}}$: 
        \smallskip
        \begin{center}
            {\AxiomC{$\vdots$}
            \noLine
            \UnaryInfC{$\Gamma\specialvdashdef{{\bf S_{1}}}{{\bf T_{1}}}\Delta,A$}
            \AxiomC{$\vdots$}
            \noLine
            \UnaryInfC{$B,\Pi'\specialvdashdef{{\bf U_{1}}}{\bf V_{1}}\Sigma,\neg A$}
            \AxiomC{$\vdots$}
            \noLine
            \UnaryInfC{$C,\Pi'\specialvdashdef{{\bf U_{2}}}{\bf V_{2}}\Sigma,\neg A$}
            \RightLabel{\scriptsize{$L\vee$}}
            \BinaryInfC{$B\vee C,\Pi'\specialvdashdef{{\bf S_{2}}}{\bf T_{2}}\Sigma,\neg A$}
            \RightLabel{\scriptsize{$cut_{a}$}}
            \BinaryInfC{$B\vee C,\Pi',\Gamma\specialvdashdef{{\bf S}}{{\bf T}}\Delta,\Sigma$}
            \DisplayProof}
        \end{center}
        \smallskip
        To reach the conclusion, we consider $\mathcal{T}_{\pi_{2}}(\neg A)$. If the label of any leaf in $\mathcal{T}_{\pi_{2}}(\neg A)$ is introduced by an $RW$ application, we remove such $RW$ applications altogether, thus obtaining a $\Gastc$-proof $\pi'_{2}$ of $\Pi\specialvdashdef{{\bf S_{2}}}{{\bf T_{2}}}\Sigma$: we apply $LW$, $RW$ and $\sigma$ to conclude. Otherwise, there must exist one extra-logical rule  in $\pi_{2}$. By contradiction, suppose this is not the case: from the fact that $A$ does not occur in $\mathcal{W}$ we infer the existence of (at least) one leaf in $\mathcal{T}_{\pi_{2}}(\neg A)$ labelled by an occurrence of $\neg A$ introduced by a logical instance of $ax$. This implies that $\neg A$ occurs in $\Pi$ -- contrary to the hypothesis. 

In the case where at least one extra-logical rule is applied in $\pi_{2}$, we first rewrite $\pi_{2}$ globally and then proceed by secondary induction on the number of logical rules occurring below the topmost extra-logical rule applications.
Starting from these topmost applications, the soundness of $\Pi,\Gamma \specialvdashdef{{\bf S}}{{\bf T}} \Delta,\Sigma$ guarantees that $\fullGpn$ proves $W,\Pi \vdash B$, where $B$ is any premise of the corresponding default rule. Hence $B \in Cn(\mathcal{W}\cup \Pi)$. By Theorem \ref{completenessbase}, $\Gastc$ proves $\Pi \specialvdashdef{\varnothing}{\varnothing} \Theta_{i}$ for every $\quest \Theta_{i}$ in $\topfo(\quest B)$.
We can now apply the extra-logical rule to the premises $\Pi \specialvdashdef{\varnothing}{\varnothing}\Theta_{i}$; soundness of $\Pi,\Gamma \specialvdashdef{{\bf S}}{{\bf T}} \Delta,\Sigma$ ensures that this application is safe. Next, we reapply all right rules that occur below this extra-logical rule in $\pi_{2}$. Instances of $cut_{asa}$ are preserved, and all resulting sequents remain sound.
We iterate this procedure for each subsequent extra-logical rule in $\pi_{2}$. The resulting $\Gastc$-proof $\pi_{2}'$ derives $\Pi\specialvdashdef{{\bf S_{2}}}{{\bf T_{2}'}}\Sigma,\neg A$ with strictly fewer logical rules below the highest extra-logical rule. At this point, the secondary inductive hypothesis applies, and the argument is complete.
    \end{itemize}
    \smallskip
    If $A$ is not principal in the last rule of $\pi_{1}$, and this rule is neither $L\vee$ nor $L\neg\wedge$, we apply the inductive hypothesis (at most twice) to obtain the desired conclusion. If the last rule in $\pi_{1}$ is, say, $L\vee$, we focus on $\mathcal{T}_{\pi_{1}}(A)$. If every leaf of $\mathcal{T}_{\pi_{1}}(A)$ is labelled by an occurrence of $A$ introduced by $RW$, we simply remove those $RW$ applications to obtain a $\Gastc$-proof of $\Gamma \specialvdashdef{{\bf S_{1}}}{{\bf T_{1}}} \Delta$, and then apply $LW$, $RW$, and $\sigma$ to conclude. Otherwise, we reason by contradiction -- as in case (b) for $\pi_{2}$ -- to obtain the existence of at least one extra-logical rule applied in $\pi_{1}$. We then rewrite $\pi_{1}$ globally (again as in (b)) to reduce the number of logical rules occurring below the topmost extra-logical applications, and proceed by secondary induction on that number to complete the proof. 
\end{proof}

\begin{lemma}\label{partialinvertibility}
    Let $i=1,2$. The rules
    \smallskip
    \begin{center}
        {\AxiomC{$A_{i},A_{1}\vee A_{2},\Gamma\specialvdashdef{{\bf S}}{{\bf T}}\Delta$}
        \LeftLabel{$E\vee$}
        \UnaryInfC{$A_{i},A_{i},\Gamma\specialvdashdef{{\bf S}}{{\bf T'}}\Delta$}
        \DisplayProof}\quad
        {\AxiomC{$\neg A_{i},\neg(A_{1}\wedge A_{2}),\Gamma\specialvdashdef{{\bf S}}{{\bf T}}\Delta$}
        \RightLabel{$E\neg\wedge$}
        \UnaryInfC{$\neg A_{i},\neg A_{i},\Gamma\specialvdashdef{{\bf S}}{{\bf T'}}\Delta$}
        \DisplayProof}
    \end{center}
    \smallskip
    are admissible in $\Gastc$.
\end{lemma}

\begin{proof}
    We focus on $E\vee$, since the argument for $E\neg\wedge$ is completely analogous. We reason by induction on the height of the $\Gastc$-proof $\pi$ of the premise. We assume that ${\bf S},{\bf T}\not=\varnothing$: if ${\bf S}={\bf T}=\varnothing$, we apply Lemma \ref{invertibilityempty} to get the conclusion. 
    
    If $h(\pi)=1$, the last rule  is a $0$-ary extra-logical rule: we simply replace $A_{1}\vee A_{2}$ with $A_{i}$ to obtain the conclusion. If $h(\pi)>1$, we reason by cases over the last rule  in $\pi$. If the latter is $L\vee$ and $A_{1}\vee A_{2}$ is principal in it, we perform the following transformation on $\pi$:
    \smallskip
    \begin{center}
        {\AxiomC{$\vdots$}
        \noLine
        \UnaryInfC{$A_{i},A_{i},\Gamma\specialvdashdef{{\bf S_{1}}}{{\bf T_{1}}}\Delta$}
        \AxiomC{$\vdots$}
        \noLine
        \UnaryInfC{$A_{i},A_{3-i},\Gamma\specialvdashdef{{\bf S_{2}}}{{\bf T_{2}}}\Delta$}
        \LeftLabel{\scriptsize{$L\vee$}}
        \BinaryInfC{$A_{i},A_{1}\vee A_{2},\Gamma\specialvdashdef{{\bf S}}{{\bf T}}\Delta$}
        \DisplayProof}\quad$\rightsquigarrow$\quad
        {\AxiomC{$\vdots$}
        \noLine
        \UnaryInfC{$A_{i},A_{i},\Gamma\specialvdashdef{{\bf S_{1}}}{{\bf T_{1}}}\Delta$}
        \doubleLine
        \RightLabel{\scriptsize{$\sigma$}}
        \UnaryInfC{$A_{i},A_{i},\Gamma\specialvdashdef{{\bf S}}{{\bf T_{1}}}\Delta$}
        \DisplayProof}
    \end{center}
    \smallskip
    If the last rule  in $\pi$ is $LW$ and $A_{1}\vee A_{2}$ is principal in it, we perform the following transformation:
    \smallskip
    \begin{center}
        {\AxiomC{$\vdots$}
        \noLine
        \UnaryInfC{$A_{i},\Gamma\specialvdashdef{{\bf S}}{{\bf T}}\Delta$}
        \LeftLabel{\scriptsize{$LW$}}
        \UnaryInfC{$A_{i},A_{1}\vee A_{2},\Gamma\specialvdashdef{{\bf S}}{{\bf T}}\Delta$}
        \DisplayProof}\quad$\rightsquigarrow$\quad
        {\AxiomC{$\vdots$}
        \noLine
        \UnaryInfC{$A_{i},\Gamma\specialvdashdef{{\bf S}}{{\bf T}}\Delta$}
        \RightLabel{\scriptsize{$LW$}}
        \UnaryInfC{$A_{i},A_{i},\Gamma\specialvdashdef{{\bf S}}{{\bf T}}\Delta$}
        \DisplayProof}
    \end{center}
    \smallskip
    If a formula in $\Gamma$ or $\Delta$ is principal in the last rule  in $\pi$, we apply the inductive hypothesis to the premises of the latter to conclude. 

    If the last rule  in $\pi$ is $cut_{asa}$, the application of the inductive hypothesis to the premises may not be possible: if the atomic cut formula is a subformula of $A_{3-i}$ without occurring in any formula in $\Gamma\cup[A_{i}]\cup\mathcal{W}$, the side condition for the applicability of $cut_{asa}$ is no more satisfied. In this scenario, we perform the upwards permutation of $E\vee$: next, we replace the $cut_{asa}$ application with an instance of $cut_{a}$. It suffices to exploit Theorem \ref{atomicnonanalytic} to reach the conclusion.
    
    If $A_{i}$ is principal in the last rule  in $\pi$, we cannot apply the inductive hypothesis to reach the conclusion: instead, we follow another argument, based on the fact that the premise $A_{i},A_{1}\vee A_{2},\Gamma\specialvdashdef{{\bf S}}{{\bf T}}\Delta$ is sound under conditions. Definition \ref{soundness} ensures that $([W,E_{\rho}]\cup[A_{i}]\cup[A_{1}\vee A_{2}]\cup\Gamma)\, ||\, \langle{\bf V},\varnothing\rangle$ for any subderivation $\rho$ of $\pi$ with $\Pi\specialvdashdef{{\bf U}}{{\bf V}}\Sigma$ as conclusion. Starting from the topmost applications of extra-logical rules in $\pi$, if $\Theta$ is the succedent of any premise of such applications, we leverage Theorem \ref{completenessbase} to infer that $\Gastc$ proves $A_{i},A_{i},\Gamma\specialvdashdef{\varnothing}{\varnothing}\Theta$. Hence, we apply the corresponding extra-logical rules to these sequents: since $([W,D_{\pi}]\cup[A_{i}]\cup[A_{1}\vee A_{2}]\cup\Gamma)\, ||\, \langle\varnothing,{\bf S}\rangle$, Lemma \ref{compatibility} ensures that such applications are safe. Next, we apply all the rules applied in $\pi$ below such extra-logical rules instances, with the exception of left (logical and structural) rules. Theorem \ref{atomicnonanalytic} ensures that any application of $cut_{asa}$ which is not preserved by the deletion of $A_{3-i}$ can be replaced by an instance of $cut_{a}$. Moreover, the fact that $([W,D_{\pi}]\cup[A_{i}]\cup[A_{1}\vee A_{2}]\cup\Gamma)\, ||\, \langle\varnothing,{\bf S}\rangle$, together with Lemma \ref{compatibility}, guarantees that all $\sigma$ applications remain safe. Such global rewriting of $\pi$ yields a $\Gastc$-proof of $A_{i},A_{i},\Gamma\specialvdashdef{{\bf S}}{{\bf T}}\Delta$.
\end{proof}

\begin{lemma}\label{leftcontraction}
    The rule of Left Contraction
    \smallskip
    \begin{center}
        {\AxiomC{$A,A,\Gamma\specialvdash{{\bf S}}{{\bf T}}\Delta$}
        \LeftLabel{$LC$}
        \UnaryInfC{$A,\Gamma\specialvdash{{\bf S}}{{\bf T}}\Delta$}
        \DisplayProof}
    \end{center}
    \smallskip
    is admissible in $\Gastc$.
\end{lemma}

\begin{proof}
We focus on Left Contraction, arguing by primary induction on the height of the $\Gastc$-proof $\pi$ of the premise and secondary induction on the logical complexity of $A$. If $h(\pi)=1$, we leverage conditions $(i)$ and $(v)$ in Definition \ref{definitioncalculi} to reach the conclusion. If $h(\pi)>1$, we reason by cases over the last rule  in $\pi$. If no occurrence of $A$ is principal in the last rule , we simply apply the inductive hypothesis to the premises. If one occurrence of $A$ is principal in the last rule , we proceed by cases over $A$'s principal connective. If $A$ is a literal, then $A$ is introduced by an $LW$ application: it suffices to remove the latter to conclude. If $A$ has either the form $B\wedge C$, or the form $\neg(B\vee C)$, or the form $\neg\neg B$, we leverage Lemma \ref{leftinvertibility} to apply the secondary inductive hypothesis and reach the conclusion. On the other hand, if $A$ has either the form $B\vee C$ or the form $\neg(B\wedge C)$, we exploit Lemma \ref{partialinvertibility} to apply twice the secondary inductive hypothesis and complete the proof. 
\end{proof}

\begin{theorem}\label{cut}
    The rule of safe Cut
    \smallskip
    \begin{center}
        {\AxiomC{$\Gamma\specialvdash{{\bf S_{1}}}{{\bf T_{1}}}\Delta,A$}
        \AxiomC{$\Pi\specialvdash{{\bf S_{2}}}{{\bf T_{2}}}\Sigma,\neg A$}
        \RightLabel{$cut$}
        \BinaryInfC{$\Gamma,\Pi\specialvdash{{\bf S}}{{\bf T}}\Delta,\Sigma$}
        \DisplayProof}
    \end{center}
    \smallskip
    is admissible in $\Gastc$.
\end{theorem}

\begin{proof}
We reason by induction over the complexity of $A$. If $A$ is a literal, Theorems \ref{atomicanalytic} and \ref{atomicnonanalytic} ensure the conclusion. If $A$ is not a literal, we reason by cases over $A$'s principal connective. We leverage Lemma \ref{rightinvertibility} to lower the complexity of the Cut formula: to reach the conclusion, we apply Lemmas \ref{rightcontraction} and \ref{leftcontraction}.
\end{proof}

\remove{\begin{proposition}
    The rules of Left Negated Contraction
    \smallskip
    \begin{center}
        {\AxiomC{$A,\Gamma\specialvdash{{\bf S}}{{\bf T}}\Delta,\neg A$}
        \LeftLabel{$LNC$}
        \UnaryInfC{$A,\Gamma\specialvdash{{\bf S}}{{\bf T}}\Delta$}
        \DisplayProof}\quad
        {\AxiomC{$\neg A,\Gamma\specialvdash{{\bf S}}{{\bf T}}\Delta,A$}
        \RightLabel{$LNC$}
        \UnaryInfC{$\neg A,\Gamma\specialvdash{{\bf S}}{{\bf T}}\Delta$}
        \DisplayProof}
    \end{center}
    \smallskip
    are admissible in $\Gastc$.
\end{proposition}

\begin{proof}
	Immediate from Theorem \ref{cut} and Lemma \ref{leftcontraction}.
\end{proof}}

$\Gastc$ calculi are expressive enough to formalize weak cumulativity properties satisfied by modified credulous consequence. First, let us consider the following Cumulative Cut rule -- a special case of the standard Cut rule:
\smallskip
\begin{center}
    {\AxiomC{$\Gamma\specialvdashdef{{\bf S_{1}}}{{\bf T_{1}}}A$}
    \AxiomC{$A,\Pi\specialvdashdef{{\bf S_{2}}}{{\bf T_{2}}}\Sigma$}
    \RightLabel{$cCut$}
    \BinaryInfC{$\Pi,\Gamma\specialvdashdef{{\bf S}}{{\bf T}}\Sigma$}
    \DisplayProof}
\end{center}

\begin{theorem}\label{wccut}
    The rule of safe Cumulative Cut is admissible in $\Gastc$.
\end{theorem}

\begin{proof}
    We focus on a topmost Cumulative Cut application, reasoning by primary induction on the logical complexity of the Cut formula $A$ and secondary induction on the sum of the heights of the $\Gastc$-proofs $\pi_{1},\pi_{2}$ of the left and right premise, respectively. In view of Theorem \ref{cutempty}, we assume that either ${\bf S}\not=\varnothing$ or ${\bf T}\not=\varnothing$. If $A$ is a literal, we reason by cases over the last rule in $\pi_{1}$. Here, we focus on the most meaningful cases.
    \smallskip
    \begin{itemize}
        \item[$(a)$] If the last rule in $\pi_{1}$ is $ax$, then by hypothesis, the last rule in $\pi_{2}$ is not $ax$. If the last rule  in $\pi_{2}$ is extra-logical, apply the secondary inductive hypothesis $m$ times to the premises to obtain the conclusion: if $m=0$, replace the $\delta$-application concluding $A,\Pi\specialvdashdef{{\bf S_{2}}}{{\bf T_{2}}}\Sigma$ with one concluding $\Pi,\Gamma\specialvdashdef{{\bf S_{2}}}{{\bf T_{2}}}\Sigma$, (possibly) followed by $\sigma$-applications yielding $\Pi,\Gamma\specialvdashdef{{\bf S}}{{\bf T_{2}}}\Sigma$. It is easy to see that all rule applications in the $\Gastc$-derivation thus obtained are safe. If the last rule  in $\pi_{2}$ is $LW$ and $A$ principal in it, we replace the Weakening application with $LW$ applications adding the formulas in $\Gamma$ (if any), plus $\sigma$ applications to add sets in ${\bf S}$ not occurring in ${\bf S_{2}}$ (if any). If the last rule  in $\pi_{2}$ is $cut_{asa}$, we apply the secondary inductive hypothesis to its premises: if the side condition for the applicability of $cut_{asa}$ is not fulfilled, we leverage Theorem \ref{atomicnonanalytic} to reach the conclusion. If the last rule  in $\pi_{2}$ is either $L\vee$ or $L\neg\wedge$, we may not be able to apply the secondary inductive hypothesis without losing soundness under constraints. If no extra-logical rule is applied in $\pi_{2}$, we remove any $\sigma$ application from $\pi_{2}$ and apply Proposition \ref{negation} to infer the existence of a $\Gastc$-proof of $\Pi\specialvdashdef{\varnothing}{\varnothing}\Sigma,\neg A$. We apply $\sigma$ to obtain a $\Gastc$-derivation of $\Pi\specialvdashdef{{\bf S_{2}}}{\varnothing}\Sigma,\neg A$: it is immediate to verify that these $\sigma$-application are safe, and thus that such derivation is a proof. If there exists (at least) one extra-logical rule applied in $\pi_{2}$, we globally rewrite $\pi_{2}$ to obtain a $\Gastc$-proof of $\Pi,\Gamma\specialvdashdef{{\bf S_{2}}}{{\bf T_{2}}}\Sigma$ (following the same procedure as in the proof of Theorem \ref{atomicnonanalytic}, point $(b)$). 
        \smallskip
        \item[$(b)$] If the last rule in $\pi_{1}$ is extra-logical and the last rule in $\pi_{2}$ is $ax$, we apply Lemma \ref{negation} and Theorem \ref{cut} to conclude. All other possible scenarios are treated as in $(a)$ above.
        \smallskip
        \item[$(c)$] If the last rule  in $\pi_{1}$ is either $L\vee$ or $L\neg\wedge$, we may not be able to apply the secondary inductive hypothesis to both premises to reach the conclusion. Instead, we globally rewrite $\pi_{2}$ to obtain a $\Gastc$-proof of $\Pi,\Gamma\specialvdashdef{{\bf S_{2}}}{{\bf T_{2}}}\Sigma$ (again, following the same procedure as in the proof of Theorem \ref{atomicnonanalytic}, point $(b)$).
    \end{itemize}
    \smallskip
    If $A$ is not a literal, we reason by cases over $A$'s principal connective. If $A$ is either $B\wedge C$ or $\neg(B\vee C)$ or $\neg\neg B$, we leverage Lemmas \ref{rightinvertibility} and \ref{leftinvertibility} to apply (at most, twice) the primary inductive hypothesis, and then Lemmas \ref{rightcontraction} and \ref{leftcontraction} to reach the conclusion. If $A$ is either $B\vee C$ or $\neg(B\wedge C)$, and $\pi_{2}$ does not contain applications of extra-logical rules, we remove all $\sigma$ applications (if any) and then apply Lemma \ref{negation} and Theorem \ref{cut} to obtain a $\Gastc$-proof of $\Pi,\Gamma\specialvdashdef{{\bf S_{1}}}{{\bf T_{1}}}\Sigma$: we employ safe $\sigma$ applications to prove $\Pi,\Gamma\specialvdashdef{{\bf S}}{{\bf T_{1}}}\Sigma$. On the other hand, if $\pi_{2}$ contains applications of extra-logical rules, we globally rewrite $\pi_{2}$ exploiting soundness under conditions to obtain a $\Gastc$-proof of $\Pi,\Gamma\specialvdashdef{{\bf S}}{{\bf T_{2}}}\Sigma$ (again, following the same procedure as in the proof of Theorem \ref{atomicnonanalytic}, point $(b)$).
\end{proof}

\remove{\begin{proposition}
    The rules of Right Negated Contraction
    \smallskip
    \begin{center}
        {\AxiomC{$A,\Gamma\specialvdash{\varnothing}{\varnothing}\Delta,\neg A$}
        \LeftLabel{$RNC$}
        \UnaryInfC{$\Gamma\specialvdash{\varnothing}{\varnothing}\Delta,\neg A$}
        \DisplayProof}\quad
        {\AxiomC{$\neg A,\Gamma\specialvdash{\varnothing}{\varnothing}\Delta,A$}
        \RightLabel{$RNC$}
        \UnaryInfC{$\Gamma\specialvdash{\varnothing}{\varnothing}\Delta,A$}
        \DisplayProof}
    \end{center}
    \smallskip
    are admissible in $\Gastc$.
\end{proposition}

\begin{proof}
We exploit Theorem \ref{wccut} as follows (the other case is completely analogous):
 \smallskip
 \begin{center}
     {\AxiomC{$\vdots$}
     \noLine
     \UnaryInfC{$\specialvdashdef{\varnothing}{\varnothing}A\vee\neg A$}
     \AxiomC{$A,\Gamma\specialvdashdef{\varnothing}{\varnothing}\Delta,\neg A$}
     \AxiomC{$\vdots$}
     \noLine
     \UnaryInfC{$\neg A,\Gamma\specialvdashdef{\varnothing}{\varnothing}\Delta,\neg A$}
     \RightLabel{\scriptsize{$L\vee$}}
     \BinaryInfC{$A\vee\neg A,\Gamma\specialvdashdef{\varnothing}{\varnothing}\Delta,\neg A$}
     \RightLabel{\scriptsize{$wcCut$}}
     \BinaryInfC{$\Gamma\specialvdashdef{\varnothing}{\varnothing}\Delta,\neg A$}
     \DisplayProof}
 \end{center}
\end{proof}}

Next, let us consider the following Cautious Monotony rule:

\begin{theorem}
The rule of safe Cautious Monotony
\smallskip
    \begin{center}
        {\AxiomC{$\Gamma\specialvdashdef{{\bf S_{1}}}{{\bf T_{1}}}\Delta_{1}$}
        \AxiomC{$\Pi\specialvdashdef{{\bf S_{2}}}{{\bf T_{2}}}\Delta_{2}$}
        \RightLabel{$cMon$}
        \BinaryInfC{$\Gamma,\bigvee\Delta_{1},\Pi\specialvdashdef{{\bf S}}{{\bf T}}\Delta_{2}$}
        \DisplayProof}
    \end{center}
    \smallskip
   is admissible in $\Gastc$.
\end{theorem}

\begin{proof}
  We reason by induction over the sum of the heights of the $\Gastc$-proofs $\pi_{1},\pi_{2}$ of the left and right premise, respectively. If $h(\pi_{1})=h(\pi_{2})=1$, the only non-trivial case is when (say) the last rule  in $\pi_{1}$ is an extra-logical rule $\delta$. We then inspect the last rule of $\pi_{2}$. If it is $ax$, applying $LW$ and $\sigma$ yields a $\Gastc$-derivation of $\bigvee\Delta_{1},\Gamma,\Pi \specialvdashdef{\mathbf S}{\varnothing}\Delta_{2}$, which is easily seen to be a proof. If the last rule in $\pi_{2}$ is an extra-logical $\delta'$, we replace its application with one whose conclusion has $\bigvee\Delta_{1},\Gamma,\Pi$ in the antecedent, and then apply $\sigma$; soundness under conditions is immediate, and soundness under constraints follows from Lemma~\ref{compatibility}. 
  
  If $h(\pi_{1})+h(\pi_{2})>2$, we proceed by cases on the last rule of $\pi_{2}$. If a formula in $\Pi$ is principal, the inductive hypothesis applies directly to the premises; the same reasoning covers the case where the last rule is $cut_{asa}$. If instead a formula in $\Delta_{2}$ is principal and the last rule is an extra-logical $\delta$, we apply $LW$ and $\sigma$ to the premises of $\delta$, and then reapply $\delta$. By Lemma~\ref{compatibility}, soundness under constraints is preserved, so the resulting $\Gastc$-derivation is a proof.
\end{proof}

Now, we proceed to the proof of strong adequacy for $\Gastc$ calculi -- the main result of this section.

\begin{theorem}\label{adequacy}
    Let $\langle\mathcal{W},\mathcal{D}\rangle$ be a default theory. Then there exists a $\Gastc$-proof $\pi$ of $\Gamma\specialvdashdef{{\bf S}}{{\bf T}}\Delta$, for some ${\bf S}$ and ${\bf T}$, iff $\bigvee\Delta$ belongs to some modified extension of $\langle\mathcal{W}\cup\Gamma,\mathcal{D}\rangle$.
\end{theorem}

\begin{proof}
    $(\Rightarrow)$ We reason by induction on the height of $\pi$ to establish the existence of a modified extension $\langle\mathcal{E},\mathcal{F}\rangle$ of $\langle\mathcal{W}\cup\Gamma,def(\pi)\rangle$ such that $\bigvee\Delta$ belongs to $\mathcal{E}$: by Lemma \ref{semimon}, this suffices to the conclusion. 
    
    If $h(\pi)=1$, the only non-trivial case is when the last rule  is an extra-logical rule $\delta_{\mathcal{D}'}$. By conditions $(v)-(vii)$ in Definition \ref{definitioncalculi}, $\bigvee\Delta$ belongs to $Cn(D)$, where $D$ is the conclusion of the single default rule in $\mathcal{D}'$. From the fact that $\Gamma\specialvdash{{\bf S}}{{\bf T}}\Delta$ is sound under constraints we infer that $\bigvee\Delta$ belongs to a modified extension $\langle\mathcal{E},\mathcal{F}\rangle$ of $\langle\mathcal{W}\cup\Gamma,\mathcal{D}'\rangle$. If $h(\pi)>1$, we reason by cases over the last rule  in $\pi$. If the latter is $RW$, $\sigma$ or a unary logical rule, we apply the inductive hypothesis to the premises to reach the conclusion: for left logical rules, we leverage the fact that any modified extension of $\langle\mathcal{W}',def(\pi)\rangle$ is a modified extension of $\langle\mathcal{W}'',def(\pi)\rangle$ and vice versa, whenever $\mathcal{W}'$ is classically equivalent to $\mathcal{W}''$. If the last rule  in $\pi$ is $LW$, the formula $A$ is principal in it and $\Gamma=\Gamma'\cup[A]$, we apply the inductive hypothesis to infer the existence of a modified extension $\langle\mathcal{E},\mathcal{F}\rangle$ of $\langle\mathcal{W}\cup\Gamma',def(\pi)\rangle$: from the fact that $A,\Gamma'\specialvdash{{\bf S}}{{\bf T}}\Delta$ is sound under constraints we obtain the conclusion. If the last rule  in $\pi$ is an extra-logical rule $\delta_{\mathcal{D}'}$ and $\delta_{\mathcal{D}'}$ is generated according to point $(v)$ in Definition \ref{definitioncalculi}, $\pi$ has the  form:
    \smallskip
    \begin{center}
    	{\AxiomC{$\vdots_{\pi_{1}}$}
	\noLine
	\UnaryInfC{$\Gamma\specialvdashdef{{\bf S_{1}}}{{\bf T_{1}}}\Theta_{1}$}
	\AxiomC{$\cdots$}
	\AxiomC{$\vdots_{\pi_{m}}$}
	\noLine
	\UnaryInfC{$\Gamma\specialvdashdef{{\bf S_{m}}}{{\bf T_{m}}}\Theta_{m}$}
	\RightLabel{\scriptsize{$\delta_{\mathcal{D}'}$}}
	\TrinaryInfC{$\Gamma\specialvdashdef{{\bf S}}{{\bf T}}\Phi$}
	\DisplayProof}
    \end{center}
    \smallskip
    with $m>1$. By Definitions \ref{controlpair} and \ref{soundness}, the fact that $\Gamma\specialvdashdef{{\bf S}}{{\bf T}}\Phi$ is sound under conditions entails that $\fullGpn$ proves $W,D'_{\pi},\Gamma\vdash\bigwedge_{i=1}^{m}(\bigvee\Theta_{i})$: by definition, this means that $\fullGpn$ proves $W,\bigwedge concl(def(\pi_{1})\cup\cdots\cup def(\pi_{m})),\Gamma\vdash\bigwedge_{i=1}^{m}(\bigvee\Theta_{i})$. On the other hand, $\Gamma\specialvdashdef{{\bf S}}{{\bf T}}\Phi$ is sound under constraints: Definitions \ref{controlpair}, \ref{soundness} and Lemma \ref{compatibility} ensure the existence of a modified extension $\langle\mathcal{E}',\mathcal{F}'\rangle$ of $\langle\mathcal{W}\cup\Gamma,def(\pi_{1})\cup\cdots\cup def(\pi_{m})\rangle$ such that $\mathcal{E}'=Cn(\mathcal{W}\cup concl(def(\pi_{1})\cup\cdots\cup def(\pi_{m})))$. As a result, we have that $\bigwedge_{i=1}^{m}(\bigvee\Theta_{i})$ belongs to $\mathcal{E}'$: hence, the default rule in $\mathcal{D}'$ is triggered in $\mathcal{E}'$. The soundness of $\Gamma\specialvdashdef{{\bf S}}{{\bf T}}\Phi$ under constraints guarantees the existence of a modified extension $\langle\mathcal{E}'',\mathcal{F}''\rangle$ of $\langle\mathcal{W}\cup\Gamma,def(\pi)\rangle$ such that the conclusion of the default rule belongs to $\mathcal{E}''$: by condition $(v)$ of Definition \ref{definitioncalculi} and closure under {\em modus ponens} of $\mathcal{E}''$ we infer that $\bigvee\Phi$ belongs to $\mathcal{E}''$. 
    The arguments are analogous when $\delta_{\mathcal{D}'}$ is generated according to points $(vi)$--$(vii)$. If the last rule in $\pi$ is a binary right logical rule or $cut_{asa}$, we apply the inductive hypothesis twice and use the soundness of $\Gamma \specialvdashdef{{\bf S}}{{\bf T}} \Delta$ to derive the conclusion. If the last rule is $L\vee$ or $L\neg\wedge$, then $\pi$ takes the form:
    
    \smallskip
    \begin{center}
        {\AxiomC{$\vdots$}
        \noLine
        \UnaryInfC{$B,\Gamma\specialvdashdef{{\bf S_{1}}}{{\bf T_{1}}}\Delta$}
        \AxiomC{$\vdots$}
        \noLine
        \UnaryInfC{$C,\Gamma\specialvdashdef{{\bf S_{2}}}{{\bf T_{2}}}\Delta$}
        \RightLabel{\scriptsize{$L\vee$}}
        \BinaryInfC{$B\vee C,\Gamma\specialvdashdef{{\bf S}}{{\bf T}}\Delta$}
        \DisplayProof}
    \end{center}
    \smallskip
    If no extra-logical rule is applied in $\pi$, Theorem \ref{completenessbase} ensures that $\bigvee\Delta$ belongs to $Cn(\mathcal{W}\cup\Gamma)$, and thus to any modified extension of $\langle\mathcal{W}\cup\Gamma,\mathcal{D}\rangle$. If at least one extra-logical rule occurs in $\pi$, we globally rewrite $\pi$ so as to permute $L\vee$ above any extra-logical rule application, thus obtaining a $\Gastc$-derivation $\rho$. Then, we proceed by induction on the number of instances of extra-logical rules in $\rho$ to conclude.

    $(\Leftarrow)$ If $\langle\mathcal{E},\mathcal{F}\rangle$ is a modified extension of $\langle\mathcal{W},\mathcal{D}\rangle$ and $A\in\mathcal{E}$, then there exists a $k$ such that $A\in\mathcal{E}^{k}$. We show by (transfinite) induction on $k$ that there exists (at least) one $\Gastc$-proof $\pi$ of $\Gamma\specialvdashdef{{\bf S}}{{\bf T}}A$, for some ${\bf S},{\bf T}$, such that $D_{\pi}\in\mathcal{E}$. If $k=j+1$ and $A$ follows from formulas in $\mathcal{E}^{j}$ {\em via} {\em modus ponens},  Theorem \ref{cut}s yields the conclusion. The crucial case is when $k=j+1$, $A$ is not a classical consequence of formulas in $\mathcal{E}^{j}$ and there exists a default rule of the form $\dfrac{B:C_{1},\ldots,C_{k}}{A}$ such that $B\in\mathcal{E}^{j}$. If $A\in\mathcal{E}$ and $\mathcal{D}'$ represents the maximal set of applicable defaults triggered by formulas in $\mathcal{E}^{j}$, then $C_{h}\in(\mathcal{F}\cup just(\mathcal{D}'))$, and so $\neg C_{h}\not\in Cn(\mathcal{E}\cup concl(\mathcal{D}'))$, for any $1\leq h\leq k$. By inductive hypothesis, there exists (at least) one $\Gastc$-proof $\rho$ of $\Gamma\specialvdashdef{{\bf U}}{{\bf V}}B$, for some ${\bf U},{\bf V}$, such that $D_{\rho}\in\mathcal{E}$. Theorem \ref{cut} and Lemma \ref{rightcontraction} guarantee the existence of a $\Gastc$-proof $\rho_{i}$ of $\Gamma\specialvdashdef{{\bf U}}{{\bf V'_{i}}}\Theta_{i}$, for any $\quest\Theta_{i}$ occurring in $\topfo(\quest B)=\{\quest\Theta_{i}\}_{1\leq i\leq m}$, such that $D_{\rho_{i}}\in\mathcal{E}$. Hence, for any $\quest\Phi$ in $\topfo(\quest A)$ consider the following $\Gastc$-derivation $\pi'$:
    \smallskip
    \begin{center}
        {\AxiomC{$\vdots_{\rho_{i}}$}
        \noLine
        \UnaryInfC{$\Gamma\specialvdashdef{{\bf U}}{{\bf V'_{i}}}\Theta_{i}$}
        \RightLabel{\scriptsize{$\delta_{\mathcal{D}'}$}}
        \UnaryInfC{$\Gamma\specialvdashdef{{\bf U'}}{{\bf V'}}\Phi$}
        \DisplayProof}
    \end{center}
    \smallskip
    where ${\bf U'}={\bf U}\cup\{\{\neg C_{1}^{f},\ldots,\neg C_{k}^{f}\}\}$. If $\neg C_{h}\not\in Cn(\mathcal{E}\cup concl(\mathcal{D}'))$, then $\neg C_{h}\not\in Cn(\mathcal{W}\cup\{D_{\rho_{i}}\}_{1\leq i\leq m}\cup\{A\})$: this implies that $\pi'$ is a $\Gastc$-proof, as desired.
\end{proof}

    
\subsection{Control pairs for normative systems}

In this subsection, we extend the $\Gastc$ calculi for $m$-credulous consequence by introducing rules for ${\sf O}$-labelled and ${\sf P}$-labelled sequents. This yields strongly complete $\Gastc$ calculi for both $m$- and $d$-credulous consequence. These calculi will serve in the next section as the basis for analysing paradoxical or dilemmatic scenarios involving conflicting information (such as typicality-based obligations, contrary-to-duty obligations, specificity-driven reasoning, and obligations with exceptions), as well as for treating deontic notions in dynamic settings (including dynamic positive permissions and dynamic obligations).

In $\Gastc$ calculi for $d$-credulous consequence, factual detachment is formalized {\em via} specific extra-logical rules, which have ${\sf F}$-labelled premises and $\times$-labelled conclusions, with $\times\in\{{\sf O},{\sf P}\}$. To keep trace of the {\em factual assumptions} on the left-hand side of the ${\sf F}$-labelled premises, we enrich the structure of each controlled $\times$-labelled sequent with an additional context, a {\em repository}. Formally, a {\em deontic controlled sequent} is a controlled sequent $$\Gamma\specialvdashany{{\bf S}}{{\bf T}}\Delta$$ with attached a context $\Pi$, as follows: $$\Pi\mid\Gamma\specialvdashany{{\bf S}}{{\bf T}}\Delta$$ The rules of $\Gastc$ transmit repositories along derivations in the same way as they transmit control pairs (see the rules in Figure \ref{fig:Gdeltadeontic}). 

Repositories were originally introduced in \cite{JLC17} to keep trace of the formulas shifted by the rules from the left-hand side of the sequent to the right-hand side, in order to ensure Cut elimination. In \cite{piatesi24}, the authors show that repositories can be dropped {\em via} a polarity-based definition of soundness. In the $\Gastc$ calculi presented here, no rule licenses shifts of formulas from the left-hand side of the sequent to the right-hand side: as a result, there is no need of repositories to store negative formulas, as in \cite{JLC17}, and no need of a polarity-based definition of soundness, like in \cite{piatesi24}. Rather, we employ repositories with the purpose of representing factual assumptions in parallel with deontic assumptions -- i.e., the formulas on the left-hand side of $\times$-labelled sequents.

The provability of a sequent of the form (say) $\Pi\mid\Gamma\specialvdashnorm{{\bf S}}{{\bf T}}\Delta$ corresponds to the truth of the statement `in (at least) one modified extension of $\langle\mathcal{W}\cup\Pi,\mathcal{D}\rangle$, $\bigvee\Delta$ is obligatory according to the normative system 
$\langle\langle\mathcal{W}^{\sf o}\cup\Gamma,\mathcal{O}\rangle,\langle\mathcal{W}^{\sf p},\mathcal{P}\rangle\rangle$'\footnote{We can see a modified extension of $\langle\mathcal{W}\cup\Pi,\mathcal{D}\rangle$ as an incomplete possible world, and $\langle\langle\mathcal{W}^{\sf o}\cup\Gamma,\mathcal{O}\rangle,\langle\mathcal{W}^{\sf p},\mathcal{P}\rangle\rangle$ as an incomplete normative system -- i.e., a normative system such that there exists (at least) one formula  which is neither explicitly required, nor explicitly permitted, nor explicitly prohibited by the system. Under this view, provability of $\Pi\mid\Gamma\specialvdashnorm{{\bf S}}{{\bf T}}\Delta$ would correspond to truth in {\em incomplete} versions of Gibbardian factual-normative worlds \cite{Gibbard90}.}. Let us note that labeling the turnstiles rather than the formulas they separate reflects the idea that a formula's content remains the same regardless of whether it appears in a factual or a normative statement. The same formula can play different roles in different games -- the `game' of facts, the `game' of obligations, the `game' of permissions -- just as a single card in a deck can assume different functions depending on the game being played\footnote{Formulas occurring in control pairs are labeled just for notational convenience.}.

\begin{figure}
\begin{flushleft}\sc axioms and structural rules \\
\medskip
\bigskip
{\AxiomC{}
\RightLabel{$ax^{\times}$}
\UnaryInfC{$\Pi\mid\Theta\specialvdashany{\varnothing}{\varnothing}\Lambda$}
\DisplayProof}
\\[\baselineskip]
{\AxiomC{$\Pi\mid\Gamma\specialvdashany{{\bf S}}{{\bf T}}\Delta$}
\RightLabel{$LW^{\times}$}
\UnaryInfC{$\Pi\mid A,\Gamma\specialvdashany{{\bf S}}{{\bf T}}\Delta$}
\DisplayProof}\quad
{\AxiomC{$\Pi\mid\Gamma\specialvdashany{{\bf S}}{{\bf T}}\Delta$}
\RightLabel{$RW^{\times}$}
\UnaryInfC{$\Pi\mid\Gamma\specialvdashany{{\bf S}}{{\bf T}}\Delta,A$}
\DisplayProof}
\\[\baselineskip]
{\AxiomC{$\Pi_{1}\mid\Gamma\specialvdashany{{\bf S_{1}}}{{\bf T_{1}}}\Delta,A$}
\AxiomC{$\Pi_{2}\mid\Gamma\specialvdashany{{\bf S_{2}}}{{\bf T_{2}}}\Delta,\neg A$}
\RightLabel{$cut_{asa}^{\times}$}
\BinaryInfC{$\Pi_{1},\Pi_{2}\mid\Gamma\specialvdashany{{\bf S}}{{\bf T}}\Delta$}
\DisplayProof}\quad
{\AxiomC{$\Gamma\specialvdashany{{\bf S}}{{\bf T}}\Delta$}
\RightLabel{$\sigma^{\times}$}
\UnaryInfC{$\Gamma\specialvdashany{{\bf S'}}{{\bf T}}\Delta$}
\DisplayProof}
\\[\baselineskip]
\medskip
\sc logical rules \\
\medskip
\bigskip
{\AxiomC{$\Pi\mid \Gamma,A,B\specialvdashany{{\bf S}}{{\bf T}} \Delta$}
\RightLabel{$L\wedge^{\times}$}
\UnaryInfC{$\Pi\mid\Gamma,A\wedge B\specialvdashany{{\bf S}}{{\bf T}} \Delta$}
\DisplayProof}\qquad
{\AxiomC{$\Pi_{1}\mid\Gamma\specialvdashany{{\bf S_{1}}}{{\bf T_{1}}} \Delta,A$}
\AxiomC{$\Pi_{2}\mid\Gamma\specialvdashany{{\bf S_{2}}}{{\bf T_{2}}} \Delta,B$}
\RightLabel{$R\wedge^{\times}$}
\BinaryInfC{$\Pi_{1},\Pi_{2}\mid\Gamma\specialvdashany{{\bf S}}{{\bf T}} \Delta,A\wedge B$}
\DisplayProof}
\\[\baselineskip]
{\AxiomC{$\Pi_{1}\mid\Gamma,A\specialvdashany{{\bf S_{1}}}{{\bf T_{1}}} \Delta$}
\AxiomC{$\Pi_{2}\mid\Gamma,B\specialvdashany{{\bf S_{2}}}{{\bf T_{2}}} \Delta$}
\RightLabel{$L\vee^{\times}$}
\BinaryInfC{$\Pi_{1},\Pi_{2}\mid\Gamma,A\vee B\specialvdashany{{\bf S}}{{\bf T}}\Delta$}
\DisplayProof}\qquad
{\AxiomC{$\Pi\mid\Gamma\specialvdashany{{\bf S}}{{\bf T}} \Delta,A,B$}
\RightLabel{$R\vee^{\times}$}
\UnaryInfC{$\Pi\mid\Gamma\specialvdashany{{\bf S}}{{\bf T}} \Delta, A\vee B$}
\DisplayProof}
\\[\baselineskip]
{\AxiomC{$\Pi_{1}\mid\Gamma,\neg A\specialvdashany{{\bf S_{1}}}{{\bf T_{1}}} \Delta$}
\AxiomC{$\Pi_{2}\mid\Gamma,\neg B\specialvdashany{{\bf S_{2}}}{{\bf T_{2}}} \Delta$}
\RightLabel{$L\neg\wedge^{\times}$}
\BinaryInfC{$\Pi_{1},\Pi_{2}\mid\Gamma,\neg(A\wedge B)\specialvdashany{{\bf S}}{{\bf T}}\Delta$}
\DisplayProof}\qquad
{\AxiomC{$\Pi\mid\Gamma\specialvdashany{{\bf S}}{{\bf T}} \Delta,\neg A,\neg B$}
\RightLabel{$R\neg\wedge^{\times}$}
\UnaryInfC{$\Pi\mid\Gamma\specialvdashany{{\bf S}}{{\bf T}} \Delta, \neg(A\wedge B)$}
\DisplayProof}
\\[\baselineskip]
{\AxiomC{$\Pi\mid\Gamma,\neg A,\neg B\specialvdashany{{\bf S}}{{\bf T}} \Delta$}
\RightLabel{$L\neg\vee^{\times}$}
\UnaryInfC{$\Pi\mid\Gamma,\neg(A\vee B)\specialvdashany{{\bf S}}{{\bf T}} \Delta$}
\DisplayProof}\qquad
{\AxiomC{$\Pi_{1}\mid\Gamma\specialvdashany{{\bf S_{1}}}{{\bf T_{1}}} \Delta,\neg A$}
\AxiomC{$\Pi_{2}\mid\Gamma\specialvdashany{{\bf S_{2}}}{{\bf T_{2}}} \Delta,\neg B$}
\RightLabel{$R\neg\vee^{\times}$}
\BinaryInfC{$\Pi_{1},\Pi_{2}\mid\Gamma\specialvdashany{{\bf S}}{{\bf T}} \Delta,\neg(A\vee B)$}
\DisplayProof}
\\[\baselineskip]
{\AxiomC{$\Pi\mid A,\Gamma\specialvdashany{{\bf S}}{{\bf T}}\Delta$}
\RightLabel{$L\neg\neg^{\times}$}
\UnaryInfC{$\Pi\mid\neg\neg A,\Gamma\specialvdashany{{\bf S}}{{\bf T}}\Delta$}
\DisplayProof}\qquad
{\AxiomC{$\Pi\mid\Gamma\specialvdashany{{\bf S}}{{\bf T}}\Delta,A$}
\RightLabel{$R\neg\neg^{\times}$}
\UnaryInfC{$\Pi\mid\Gamma\specialvdashany{{\bf S}}{{\bf T}}\Delta,\neg \neg A$}
\DisplayProof}
\\[\baselineskip]
\medskip
\sc extra-logical rules \\
\medskip
\bigskip
{\AxiomC{$\{\Gamma\specialvdashdef{{\bf S_{i}}}{{\bf T_{i}}}\Theta_{i}\}_{0\leq i\leq m}$}
\RightLabel{$\delta$}
\UnaryInfC{$\Gamma\mid\ \specialvdashany{{\bf S}}{{\bf T}}\Phi$}
\DisplayProof}\quad
{\AxiomC{$\{\Pi_{i}\mid\Gamma\specialvdashany{{\bf S_{i}}}{{\bf T_{i}}}\Theta_{i}\}_{0\leq i\leq m}$}
\RightLabel{$\delta$}
\UnaryInfC{$\Pi\mid\Gamma\specialvdashany{{\bf S}}{{\bf T}}\Phi$}
\DisplayProof}
\\[\baselineskip]
{\AxiomC{$\{\Gamma\specialvdashdef{{\bf S_{i}}}{{\bf T_{i}}}\Theta_{i}\}_{0\leq i\leq m}$}
\AxiomC{$\{\Pi_{i}\mid\Delta\specialvdashany{{\bf S_{j}}}{{\bf T_{j}}}\Theta_{i}\}_{0\leq j\leq n}$}
\RightLabel{$\delta$}
\BinaryInfC{$\Gamma,\Pi\mid\Delta\specialvdashany{{\bf S}}{{\bf T}}\Phi$}
\DisplayProof}
\end{flushleft}
\caption{Controlled sequent rules for a normative system $\langle\langle\mathcal{W}^{{\sf o}},\mathcal{O}\rangle,\langle\mathcal{W}^{{\sf p}},\mathcal{P}\rangle\rangle$}
\label{fig:Gdeltadeontic}
\end{figure}

\begin{definition}\label{definitioncalculideontic}
    Let $\langle\mathcal{W},\mathcal{D}\rangle$, $\langle\mathcal{W}^{\sf o},\mathcal{O}\rangle$ and $\langle\mathcal{W}^{\sf p},\mathcal{P}\rangle$ be a default theory, a system of obligations and a system of permissions, respectively. The controlled sequent calculus $\Gastc$ for $\langle\mathcal{W},\mathcal{D}\rangle$, $\langle\mathcal{W}^{\sf o},\mathcal{O}\rangle$ and $\langle\mathcal{W}^{\sf p},\mathcal{P}\rangle$ is an extension of the $\Gastc$ calculus for $\langle\mathcal{W},\mathcal{D}\rangle$ with rules in Figure \ref{fig:Gdeltadeontic}, provided that the conditions below are fulfilled. 
    \smallskip
    \begin{itemize}
        \item[$(i_{op})$] For each instance of $ax^{\times}$, one of the following conditions holds: $\Theta=\Lambda=\{A\}$ for some literal $A$; $\Theta=\{A,\neg A\}$ or $\Lambda=\{A,\neg A\}$ for some atom $A$; $\Theta\quest\Lambda$ belongs to $\topfstar(\quest W^{\times})$, with $W^{\times}$ being the conjunction of the formulas in $\mathcal{W}^{\times}$. 
        \smallskip
        \item[$(ii_{op})$] For each structural and logical rule, the left subscript of the turnstile $\specialvdash{}{}$ and the right superscript of the rule name, $\times$, is either ${\sf o}$ or ${\sf p}$.
        \smallskip
        \item[$(iii_{op})$] For any constrained conditional obligation  $\dfrac{B:C_{1},\ldots,C_{k}}{D}$ in $\mathcal{O}^{\sf f}$, if $\topfo(\quest B)=\{\quest\Theta_{1},\ldots,\quest\Theta_{m}\}$ with $m>0$ and $\quest\Phi$ is in $\topfstaro(\quest D)$, then there exists an extra-logical rule:
        \smallskip
        \begin{center}
            {\AxiomC{$\{\Gamma\specialvdashdef{{\bf S_{i}}}{{\bf T_{i}}}\Theta_{i}\}_{1\leq i\leq m}$}
            \RightLabel{$\delta$}
            \UnaryInfC{$\Gamma\mid\ \specialvdashnorm{{\bf S}}{{\bf T}}\Phi$}
            \DisplayProof}\remove{\quad
            {\AxiomC{$\{\Pi_{i}\mid\Gamma\specialvdashnorm{{\bf S_{i}}}{{\bf T_{i}}}\Theta_{i}\}_{1\leq i\leq m}$}
            \RightLabel{$\delta$}
            \UnaryInfC{$\Pi\mid\Gamma\specialvdashnorm{{\bf S'}}{{\bf T'}}\Phi$}
            \DisplayProof}}
        \end{center}
        \smallskip
        where ${\bf T}=\bigcup_{i=1}^{m}{\bf T_{i}}\cup\{\{B^{{\sf f}}\}\}$, ${\bf S}=\bigcup_{i=1}^{m}{\bf S_{i}}\cup\{\{\neg C_{1}^{{\sf o}},\ldots,\neg C_{k}^{{\sf o}}\}\}$\remove{, ${\bf T'}=\bigcup_{i=1}^{m}{\bf T_{i}}\cup\{\{B^{{\sf o}},B^{{\sf p}}\}\}$, ${\bf S'}={\bf S}$ and $\Pi=\Pi_{1},\ldots,\Pi_{m}$}. If instead $\topfo(\quest B)=\varnothing$, then there exists an extra-logical rule:
        \smallskip
        \begin{center}
            {\AxiomC{$ $}
            \RightLabel{$\delta$}\UnaryInfC{$\Pi\mid\ \specialvdashnorm{{\bf S}}{{\bf T}}\Phi$}
            \DisplayProof}
        \end{center}
        \smallskip 
        with ${\bf T}=\{\{\top^{{\sf f}}\}\}$ and ${\bf S}=\bigcup_{i=1}^{m}{\bf S_{i}}\cup\{\{\neg C_{1}^{{\sf o}},\ldots,\neg C_{k}^{{\sf o}}\}\}$.
        
        \noindent The same condition applies to any constrained conditional permission in $\mathcal{P}^{\sf f}$, replacing $\specialvdashnorm{}{}$ with $\specialvdashperm{}{}$, and ${\bf S}=\bigcup_{i=1}^{m}\mathbf{S_{i}}\cup\{\{\neg C_{1}^{\sf p},\ldots,\neg C_{k}^{\sf p},\neg C_{1}^{\sf o},\ldots,\neg C_{k}^{\sf o}\}\}$.
        
        \item[$(iv_{op})$] For any constrained conditional obligation $\dfrac{B:C_{1},\ldots,C_{k}}{D}$ in $\mathcal{O}^{\sf d}$, if $\topfo(\quest B)=\{\quest\Theta_{1},\ldots,\quest\Theta_{m}\}$ with $m>0$ and $\quest\Phi$ is in $\topfstaro(\quest D)$, then there exists an extra-logical rule:
        \smallskip
        \begin{center}
            {\AxiomC{$\{\Pi_{i}\mid\Gamma\specialvdashnorm{{\bf S_{i}}}{{\bf T_{i}}}\Theta_{i}\}_{1\leq i\leq m}$}
            \RightLabel{$\delta$}
            \UnaryInfC{$\Pi\mid\Gamma\mid\ \specialvdashnorm{{\bf S}}{{\bf T}}\Phi$}
            \DisplayProof}\remove{\quad
            {\AxiomC{$\{\Pi_{i}\mid\Gamma\specialvdashnorm{{\bf S_{i}}}{{\bf T_{i}}}\Theta_{i}\}_{1\leq i\leq m}$}
            \RightLabel{$\delta$}
            \UnaryInfC{$\Pi\mid\Gamma\specialvdashnorm{{\bf S'}}{{\bf T'}}\Phi$}
            \DisplayProof}}
        \end{center}
        \smallskip
      where ${\bf T}=\bigcup_{i=1}^{m}{\bf T_{i}}\cup\{\{B^{{\sf o}},B^{\sf p}\}\}$, ${\bf S}=\bigcup_{i=1}^{m}{\bf S_{i}}\cup\{\{\neg C_{1}^{{\sf o}},\ldots,\neg C_{k}^{{\sf o}}\}\}$ and $\Pi=\Pi_{1},\ldots,\Pi_{m}$. If instead $\topfo(\quest B)=\varnothing$, then there exists an extra-logical rule:
        \smallskip
        \begin{center}
            {\AxiomC{$ $}
            \RightLabel{$\delta$}\UnaryInfC{$\Pi\mid\Gamma\specialvdashnorm{{\bf S}}{{\bf T'}}\Phi$}
            \DisplayProof}
        \end{center}
        \smallskip 
        with ${\bf T'}=\{\{\top^{{\sf o}},\top^{{\sf p}}\}\}$ and ${\bf S}=\bigcup_{i=1}^{m}{\bf S_{i}}\cup\{\{\neg C_{1}^{{\sf o}},\ldots,\neg C_{k}^{{\sf o}}\}\}$.
        
        \noindent The same condition applies to any constrained conditional permission in $\mathcal{P}^{\sf d}$, replacing $\specialvdashnorm{}{}$ with $\specialvdashperm{}{}$, ${\bf T}=\bigcup_{i=1}^{m}{\bf T_{i}}\cup\{\{B^{{\sf p}}\}\}$, ${\bf T'}=\{\{\top^{{\sf p}}\}\}$ and ${\bf S}=\bigcup_{i=1}^{m}\mathbf{S_{i}}\cup\{\{\neg C_{1}^{\sf p},\ldots,\neg C_{k}^{\sf p},\neg C_{1}^{\sf o},\ldots,\neg C_{k}^{\sf o}\}\}$.
        \smallskip
        
        \item[$(v_{op})$] For any pair of extra-logical rules
    \smallskip
    \begin{center}
        {\AxiomC{$\{\Gamma\specialvdashdef{{\bf S_{i}'}}{{\bf T_{i}'}}\Theta_{i}\}_{1\leq i\leq m}$}
        \RightLabel{$\delta$}
        \UnaryInfC{$\Gamma\mid\ \specialvdashnorm{{\bf S_{1}}}{{\bf T_{1}}}\Phi$}
        \DisplayProof}\quad
        {\AxiomC{$\{\Gamma\specialvdashdef{{\bf S_{i}'}}{{\bf T_{i}'}}\Theta_{i}\}_{m+1\leq i\leq n}$}
        \RightLabel{$\delta$}
        \UnaryInfC{$\Gamma\mid\ \specialvdashnorm{{\bf S_{2}}}{{\bf T_{2}}}\Phi'$}
        \DisplayProof}
    \end{center}
    \smallskip
    with $m+n\geq0$, if $\quest\Psi$ occurs in $\topfstaro(\quest(\bigvee\Phi)\wedge(\bigvee\Phi'))$ without belonging to $\topfstaro(\quest(\bigvee\Phi))$, $\topfstaro(\quest(\bigvee\Phi'))$ or $\topfstaro(\quest W^{\sf o})$, then there exists an extra-logical rule of the  form:
    \smallskip
    \begin{center}
        {\AxiomC{$\{\Gamma\specialvdashdef{{\bf S_{i}'}}{{\bf T_{i}'}}\Theta_{i}\}_{1\leq i\leq n}$}
        \RightLabel{$\delta$}
       \UnaryInfC{$\Gamma\mid\ \specialvdashnorm{{\bf S}}{{\bf T}}\Psi$}
        \DisplayProof}
    \end{center}
    \smallskip
    with ${\bf S}={\bf S_{1}}\cup{\bf S_{2}}$ and ${\bf T}={\bf T_{1}}\cup{\bf T_{2}}$. 
    
    \noindent The same condition applies to any pair of extra-logical rules of the same form where $\specialvdashperm{}{}$ replaces $\specialvdashnorm{}{}$, with $W^{\sf p}$ replacing $W^{\sf o}$. 
        \smallskip
        \item[$(vi_{op})$] For any pair of extra-logical rules
    \smallskip
    \begin{center}
        {\AxiomC{$\{\Pi_{i}\mid\Gamma\specialvdashnorm{{\bf S_{i}'}}{{\bf T_{i}'}}\Theta_{i}\}_{1\leq i\leq m}$}
        \RightLabel{$\delta$}
        \UnaryInfC{$\Pi'\mid\Gamma\specialvdashnorm{{\bf S_{1}}}{{\bf T_{1}}}\Phi$}
        \DisplayProof}\quad
        {\AxiomC{$\{\Pi_{i}\mid\Gamma\specialvdashnorm{{\bf S_{i}'}}{{\bf T_{i}'}}\Theta_{i}\}_{m+1\leq i\leq n}$}
        \RightLabel{$\delta$}
        \UnaryInfC{$\Pi''\mid\Gamma\specialvdashnorm{{\bf S_{2}}}{{\bf T_{2}}}\Phi'$}
        \DisplayProof}
    \end{center}
    \smallskip
    with $m+n\geq0$, $\Pi'=\Pi_{1},\ldots,\Pi_{m}$ and $\Pi''=\Pi_{m+1},\ldots,\Pi_{n}$, if $\quest\Psi$ occurs in $\topfstaro(\quest(\to\bigvee\Phi)\wedge(\bigvee\Phi'))$ without belonging to $\topfstaro(\quest(\bigvee\Phi))$, $\topfstaro(\quest(\bigvee\Phi'))$ or $\topfstaro(\quest W^{\sf o})$, then there exists an extra-logical rule of the form:
    \smallskip
    \begin{center}
        {\AxiomC{$\{\Pi_{i}\mid\Gamma\specialvdashnorm{{\bf S_{i}'}}{{\bf T_{i}'}}\Theta_{i}\}_{1\leq i\leq n}$}
        \RightLabel{$\delta$}
       \UnaryInfC{$\Pi\mid\Gamma\specialvdashnorm{{\bf S}}{{\bf T}}\Psi$}
        \DisplayProof}
    \end{center}
    \smallskip
    with ${\bf S}={\bf S_{1}}\cup{\bf S_{2}}$, ${\bf T}={\bf T_{1}}\cup{\bf T_{2}}$ and $\Pi=\Pi',\Pi''$. 
    
    \noindent The same condition applies to any pair of extra-logical rules of the same form where $\specialvdashperm{}{}$ replaces $\specialvdashnorm{}{}$, with $W^{\sf p}$ replacing $W^{\sf o}$. 

    \item[$(vii_{op})$] For any pair of extra-logical rules
        \smallskip
        \begin{center}
            {\AxiomC{$\{\Gamma\specialvdashdef{{\bf S_{i}}}{{\bf T_{i}}}\Theta_{i}\}_{1\leq i\leq m}$}
            \RightLabel{$\delta$}
            \UnaryInfC{$\Gamma\mid\ \specialvdashnorm{{\bf S'}}{{\bf T'}}\Phi$}
            \DisplayProof}\quad
            {\AxiomC{$\{\Delta_{i}\mid\Pi\specialvdashnorm{{\bf S_{i}}}{{\bf T_{i}}}\Theta_{i}\}_{m+1\leq i\leq n}$}
            \RightLabel{$\delta$}
            \UnaryInfC{$\Delta\mid\Pi\specialvdashnorm{{\bf S''}}{{\bf T''}}\Phi'$}
            \DisplayProof}
        \end{center}
        \smallskip
        with $m\geq 0$ and $\Delta=\Delta_{m+1},\ldots,\Delta_{n}$, if $\quest\Phi''$ occurs in $\topfstaro(\quest(\bigvee\Phi)\wedge(\bigvee\Phi'))$ without occurring in $\topfstaro(\quest\bigvee\Phi)$, $\topfstaro(\quest\bigvee\Phi')$ or $\topfstaro(\quest W^{\sf o})$, there exists an extra-logical rule of the form 
        \smallskip
        \begin{center}
            {\AxiomC{$\{\Gamma\specialvdashdef{{\bf S_{i}}}{{\bf T_{i}}}\Theta_{i}\}_{1\leq i\leq m}$}
            \AxiomC{$\{\Delta_{i}\mid\Pi\specialvdashnorm{{\bf S_{i}'}}{{\bf T_{i}'}}\Theta_{i}\}_{m+1\leq i\leq n}$}
            \BinaryInfC{$\Gamma,\Delta\mid\Pi\specialvdashnorm{{\bf S}}{{\bf T}}\Phi''$}
            \DisplayProof}
        \end{center}
        \smallskip
        with ${\bf T}={\bf T'}\cup{\bf T''}$ and ${\bf S}={\bf S'}\cup{\bf S''}$.

        \noindent The same condition applies to any pair of extra-logical rules of the same form where $\specialvdashperm{}{}$ replaces $\specialvdashnorm{}{}$, with $W^{\sf p}$ correspondingly replaced by $W^{\sf o}$.

    \item[$(viii_{op})$] For any extra-logical rule
        \smallskip
        \begin{center}
            {\AxiomC{$\{\Gamma\specialvdashdef{{\bf S_{i}}}{{\bf T_{i}}}\Theta_{i}\}_{1\leq i\leq m}$}
            \RightLabel{$\delta$}
            \UnaryInfC{$\Gamma\mid\ \specialvdashnorm{{\bf S}}{{\bf T}}\Phi$}
            \DisplayProof}\quad
            {\AxiomC{$\{\Pi_{i}\mid\Gamma\specialvdashnorm{{\bf S_{i}}}{{\bf T_{i}}}\Theta_{i}\}_{1\leq i\leq m}$}
            \RightLabel{$\delta$}
            \UnaryInfC{$\Pi\mid\Gamma\specialvdashnorm{{\bf S}}{{\bf T}}\Phi$}
            \DisplayProof}
        \end{center}
        \smallskip
        with $m\geq 0$ and $\Pi=\Pi_{1},\ldots,\Pi_{m}$, if $\quest\Phi'$ occurs in $\topfstaro(\quest(\bigvee\Phi)\wedge W)$ without occurring in $\topfstaro(\quest W^{\sf o})$, there exists an extra-logical rule of the form 
        \smallskip
        \begin{center}
            {\AxiomC{$\{\Gamma\specialvdashdef{{\bf S_{i}}}{{\bf T_{i}}}\Theta_{i}\}_{1\leq i\leq m}$}
            \RightLabel{$\delta$}
            \UnaryInfC{$\Gamma\mid\ \specialvdashnorm{{\bf S}}{{\bf T}}\Phi'$}
            \DisplayProof}\quad
            {\AxiomC{$\{\Pi_{i}\mid\Gamma\specialvdashnorm{{\bf S_{i}}}{{\bf T_{i}}}\Theta_{i}\}_{1\leq i\leq m}$}
            \RightLabel{$\delta$}
            \UnaryInfC{$\Pi\mid\Gamma\specialvdashnorm{{\bf S}}{{\bf T}}\Phi'$}
            \DisplayProof}
        \end{center}
        \smallskip
        The same condition applies to any pair of extra-logical rules of the same form in which $\specialvdashperm{}{}$ replaces $\specialvdashnorm{}{}$, with $W^{\sf p}$ correspondingly replaced by $W^{\sf o}$.
    \end{itemize}
\end{definition}

For any extra-logical rule $\delta$ in $\Gastc$ whose conclusion is ${\sf O}$-labelled (${\sf P}$-labelled), we define the label $\delta_{\mathcal{O}'}$ ($\delta_{\mathcal{P}'}$) as follows:
\smallskip
\begin{itemize}
    \item[$(a_{op})$] if $\delta$ is generated in accordance with point $(iii_{op})$ in Definition \ref{definitioncalculideontic}, then $\mathcal{O}'$ ($\mathcal{P}'$, respectively) is $\Big\{\dfrac{B:C_{1},\ldots,C_{k}}{D}\Big\}$;
    \smallskip
    \item[$(b_{op})$] if $\delta$ is generated from extra-logical rules $\delta'$ and $\delta''$ with labels $\delta_{\mathcal{O}_{1}}$ and $\delta_{\mathcal{O}_{2}}$ ($\delta_{\mathcal{P}_{1}}$ and $\delta_{\mathcal{P}_{2}}$) in accordance with points $(iv_{op})-(viii_{op})$ in Definition \ref{definitioncalculideontic}, then $\mathcal{O}'$ is $\mathcal{O}_{1}\cup\mathcal{O}_{2}$ ($\mathcal{P}'$ is $\mathcal{P}_{1}\cup\mathcal{P}_{2}$, respectively).
\end{itemize}
\smallskip
For each $\G$-derivation $\pi$ we say that a constrained conditional obligation $\dfrac{B:C_{1},\ldots,C_{k}}{D}$ {\em belongs to $obl(\pi)$} if and only if there is (at least) one extra-logical rule labelled $\delta_{\mathcal{O}'}$ which is applied in $\pi$ and such that $\dfrac{B:C_{1},\ldots,C_{k}}{D}$ belongs to $\mathcal{O}'$. 

Moreover, we say that a rule $\dfrac{B:C_{1},\ldots,C_{k}}{D}$ {\em belongs to $perm(\pi)$} if and only if there is (at least) one extra-logical rule labelled $\delta_{\mathcal{P}'}$  in $\pi$ and such that $\dfrac{B:C_{1},\ldots,C_{k}}{D}$ belongs to $\mathcal{P}'$, or it belongs to $obl(\pi)$.

We employ $obl'(\pi)$ ($perm'(\pi)$) to denote $obl(\pi_{1})\cup\cdots\cup obl(\pi_{m})$ ($perm(\pi_{1})\cup\cdots\cup perm(\pi_{m})$, respectively) whenever $\pi_{1},\ldots,\pi_{m}$ are the immediate subderivations yielding the premises of the lowermost extra-logical rules applications in $\pi$. We shall use 
\smallskip
\begin{itemize}
    \item[$-$] $D_{\pi}^{\sf o}$ to refer to the conjunction of the formulas in $concl(obl(\pi))$, and $E_{\pi}$ to denote the conjunction of the formulas in $concl(obl'(\pi))$;
    \smallskip
    \item[$-$] $D_{\pi}^{\sf p}$ to refer to the conjunction of the formulas in $concl(perm(\pi))$, and $E_{\pi}$ to denote the conjunction of the formulas in $concl(perm'(\pi))$.
\end{itemize}

\begin{definition}\label{soundnessobligation}
Let $\times\in\{\sf O,{\sf P}\}$, $\pi$ be a $\Gastc$-derivation of $\Pi\mid\Gamma\specialvdashnorm{{\bf S}}{{\bf T}}\Delta$ and ${\bf T^{\sf o}},{\bf S^{\sf o}}$ (${\bf T^{\sf p}},{\bf S^{\sf p}}$; ${\bf T^{\sf f}},{\bf S^{\sf f}}$) be the sets obtained from ${\bf T},{\bf S}$ by deleting ${\sf f}$- and ${\sf p}$-labelled formulas (${\sf f}$- and ${\sf o}$-labelled formulas; ${\sf o}$- and ${\sf p}$-labelled formulas, respectively). \\ \noindent The sequent $\Pi\mid\Gamma\specialvdashnorm{{\bf S}}{{\bf T}}\Delta$ is
    \smallskip
    \begin{itemize}
        \item[$(i)$] {\em sound under conditions} if and only if 
        \smallskip
        \begin{itemize}
            \item[$(i.i)$] $([W^{\sf o},E_{\rho}^{\sf o}]\cup\Gamma)\, ||\, \langle{\bf V^{\sf o}},\varnothing\rangle$ for any subderivation $\rho$ of $\pi$ with $\Pi'\mid\Gamma'\specialvdashany{{\bf U}}{{\bf V}}\Delta'$ as conclusion;
            \smallskip
            \item[$(i.ii)$] $([W^{\sf p},E_{\rho}^{\sf p}]\cup\Gamma)\, ||\, \langle{\bf V^{\sf p}},\varnothing\rangle$ for any subderivation $\rho$ of $\pi$ with $\Pi'\mid\Gamma'\specialvdashany{{\bf U}}{{\bf V}}\Delta'$ as conclusion;
            \smallskip
            \item[$(i.iii)$] $([W^{\sf f},E_{\rho}^{\sf f}]\cup\Pi)\, ||\, \langle{\bf V^{\sf f}},\varnothing\rangle$ for any subderivation $\rho$ of $\pi$ with $\Pi'\mid\Gamma'\specialvdashany{{\bf U}}{{\bf V}}\Delta'$ as conclusion;
        \end{itemize}
        \smallskip
        \item[$(ii)$] {\em sound under constraints} if and only if 
        \smallskip
        \begin{itemize}
            \item[$(ii.i)$] $([W^{\sf o},D_{\pi}^{\sf o}]\cup\Gamma)\, ||\, \langle\varnothing,{\bf S^{\sf o}}\rangle$;
            \smallskip
            \item[$(ii.ii)$] $([W^{\sf p},D_{\pi}^{\sf p}]\cup\Gamma)\, ||\, \langle\varnothing,{\bf S^{\sf p}}\rangle$;
             \smallskip
            \item[$(ii.iii)$] $([W^{\sf f},D_{\pi}^{\sf f}]\cup\Pi)\, ||\, \langle\varnothing,{\bf S^{\sf f}}\rangle$;
        \end{itemize}
        \smallskip
        \item[$(iii)$] {\em sound} if and only if it is sound under conditions and sound under constraints.
    \end{itemize}
\end{definition}

\begin{definition}\label{soundnesspermission}
Let $\times\in\{\sf O,{\sf P}\}$, $\pi$ be a $\Gastc$-derivation of $\Pi\mid\Gamma\specialvdashany{{\bf S}}{{\bf T}}\Delta$ and ${\bf T^{\sf o}},{\bf S^{\sf o}}$ (${\bf T^{\sf p}},{\bf S^{\sf p}}$; ${\bf T^{\sf f}},{\bf S^{\sf f}}$) be the sets obtained from ${\bf T},{\bf S}$ by deleting ${\sf f}$- and ${\sf p}$-labelled formulas (${\sf f}$- and ${\sf o}$-labelled formulas; ${\sf o}$- and ${\sf p}$-labelled formulas, respectively). \noindent \\ The sequent $\Pi\mid\Gamma\specialvdashperm{{\bf S}}{{\bf T}}\Delta$ is
    \smallskip
    \begin{itemize}
        \item[$(i)$] {\em sound under conditions} if and only if 
        \smallskip
        \begin{itemize}
            \item[$(i.i)$] $([W^{\sf p},E_{\rho}^{\sf p}]\cup\Gamma)\, ||\, \langle{\bf V^{\sf p}},\varnothing\rangle$ for any subderivation $\rho$ of $\pi$ with $\Pi'\mid\Gamma'\specialvdashany{{\bf U}}{{\bf V}}\Delta'$ as conclusion;
            \smallskip
            \item[$(i.ii)$] $([W^{\sf f},E_{\rho}^{\sf f}]\cup\Pi)\, ||\, \langle{\bf V^{\sf f}},\varnothing\rangle$ for any subderivation $\rho$ of $\pi$ with $\Pi'\mid\Gamma'\specialvdashany{{\bf U}}{{\bf V}}\Delta'$ as conclusion;
        \end{itemize}
        \smallskip
        \item[$(ii)$] {\em sound under constraints} if and only if $([W^{\sf p},D_{\pi}^{\sf p}]\cup\Gamma)\, ||\, \langle\varnothing,{\bf S^{\sf p}}\rangle$;
        \begin{itemize}
            \item [$(ii.i)$] $([W^{\sf p},D_{\pi}^{\sf p}]\cup\Gamma)\, ||\, \langle\varnothing,{\bf S^{\sf p}}\rangle$;
            \smallskip
            \item[$(ii.ii)$] $([W^{\sf f},D_{\pi}^{\sf f}]\cup\Pi)\, ||\, \langle\varnothing,{\bf S^{\sf f}}\rangle$;
        \end{itemize}
        \smallskip
        \item[$(iii)$] {\em sound} if and only if it is sound under conditions and sound under constraints.
    \end{itemize}
\end{definition}

\begin{definition}\label{proofdeontic}
Let $\times\in\{{\sf O},{\sf P}\}$ and $\pi$ be a $\Gastc$-derivation of $\Pi\mid\Gamma\specialvdashany{{\bf S}}{{\bf T}}\Delta$. Then, $\pi$ is a {\em proof} of $\Pi\mid\Gamma\specialvdashany{{\bf S}}{{\bf T}}\Delta$ if and only if any sequent $\Pi'\mid\Gamma'\specialvdashany{{\bf S'}}{{\bf T'}}\Delta'$ in $\pi$ is sound -- and a {\em paraproof}, otherwise.
\end{definition}

\begin{lemma}
    Let $\times\in\{{\sf O},{\sf P}\}$. The rule of safe External Weakening
    \smallskip
    \begin{center}
        {\AxiomC{$\Pi\mid\Gamma\specialvdashany{{\bf S}}{{\bf T}}\Delta$}
        \RightLabel{$EW$}
        \UnaryInfC{$A,\Pi\mid\Gamma\specialvdashany{{\bf S}}{{\bf T}}\Delta$}
        \DisplayProof}
    \end{center}
    \smallskip
    is admissible in $\Gastc$. 
\end{lemma}

\begin{proof}
    We reason by induction on the height of the proof $\pi$ of $\Pi\mid\Gamma\specialvdashany{{\bf S}}{{\bf T}}\Delta$. If $h(\pi)=1$, the conclusion is immediate. Otherwise, we reason by cases over the last rule applied. If the latter is a structural or a logical rule, we simply apply the inductive hypothesis to one of the premises to conclude. The same argument holds whenever the last rule  is extra-logical and the premises are $\times$-labelled sequents. The only non-trivial case arises if the last rule  is an extra-logical rule $\delta$ and the premises are ${\sf F}$-labelled sequents. To reach the conclusion, we apply safe $LW$ to each premise and then apply $\delta$.  
\end{proof}

Let $\star$ vary over ${\sf f},{\sf o},{\sf p}$. For any set $\Gamma$ of $\star$-labelled formulas, we use $[\Gamma]^{u}$ to denote the set of formulas obtained from $\Gamma$ by deleting labels.

\begin{lemma}   
    The rule of safe Deontic Implication
    \smallskip
    \begin{center}
    {\AxiomC{$\Pi\mid\Gamma\specialvdashnorm{{\bf S}}{{\bf T}}\Delta$}
    \RightLabel{$op$}
    \UnaryInfC{$\Pi\mid\Gamma\specialvdashperm{{\bf S'}}{{\bf T'}}\Delta$}
    \DisplayProof}
    \end{center}
    \smallskip
    is admissible in $\Gastc$, with $[\bigcup({\bf T}^{\sf o}\cup{\bf T^{\sf p}})]^{u}=[\bigcup(({\bf T'})^{\sf o}\cup{({\bf T'})^{\sf p}})]^{u}$ and $[\bigcup({\bf S}^{\sf o}\cup{\bf S^{\sf p}})]^{u}=[\bigcup(({\bf S'})^{\sf o}\cup{({\bf S'})^{\sf p}})]^{u}$.
\end{lemma}

\begin{proof}
 By routine induction on the height of the proof of the premise, exploiting the fact that $\mathcal{W}^{\sf o}\subseteq\mathcal{W}^{\sf p}$ and $\mathcal{O}\subseteq\mathcal{P}$.
\end{proof}

\begin{lemma}\label{structuraldeontic}
Let $\times\in\{{\sf O},{\sf P}\}$ and $\otimes\in\{\wedge,\vee,\neg\}$. $\Gastc$ calculi enjoy the following structural properties.
\smallskip
\begin{itemize}
    \item[$(i)$] The rules $R\otimes^{\times}$, $R\neg\otimes^{\times}$, $L\wedge$, $L\neg\vee$ and $L\neg\neg$ are invertible in $\Gastc$.
    \smallskip
    \item[$(ii)$] The rules $E\vee^{\times}$ and $E\neg\wedge^{\times}$ are admissible in $\Gastc$. 
    \smallskip
    \item[$(iii)$] The $\times$-labelled versions of Left and Right Contraction are admissible in $\Gastc$.
\end{itemize}
\end{lemma}

\begin{proof}
    $(i)$ We argue as in the proofs of Lemmas \ref{rightinvertibility} and \ref{leftinvertibility}. $(ii)$ We proceed as in the proof of Lemma \ref{partialinvertibility}. $(iii)$ We argue as in the proofs of Lemmas \ref{rightcontraction} and \ref{leftcontraction}.
\end{proof}

\begin{theorem}\label{safecutdeontic}
    The rule of safe Cut 
    \smallskip
    \begin{center}
        {\AxiomC{$\Pi_{1}\mid\Gamma_{1}\specialvdashdef{{\bf S_{1}}}{{\bf T_{1}}}\Delta_{1},A$}
        \AxiomC{$\Pi_{2}\mid\Gamma_{2}\specialvdashdef{{\bf S_{2}}}{{\bf T_{2}}}\Delta_{2},\neg A$}
        \BinaryInfC{$\Pi_{1},\Pi_{2}\mid\Gamma_{2},\Gamma_{1}\specialvdashdef{{\bf S}}{{\bf T}}\Delta_{1},\Delta_{2}$}
        \DisplayProof}
    \end{center}
    \smallskip
    is admissible in $\Gastc$.
\end{theorem}

\begin{proof}
    We proceed as in the proof of Theorem \ref{cut}, exploiting Lemma \ref{structuraldeontic}. 
\end{proof}

\begin{theorem}
    Let $\langle\mathcal{W},\mathcal{D}\rangle$ be a default theory and $\langle\langle\mathcal{W}^{\sf o},\mathcal{O}\rangle,\langle\mathcal{W}^{\sf p},\mathcal{P}\rangle\rangle$ be a normative system. Then the following statements hold.
    \smallskip
    \begin{itemize}
        \item[$(i)$] $\Gastc$ proves $\Pi\mid\Gamma\specialvdashnorm{{\bf S}}{{\bf T}}A$ if and only if $A$ belongs to (at least) one $d$-extension of $\langle\langle\mathcal{W}\cup\Pi,\mathcal{D}\rangle,\langle\langle\mathcal{W}^{\sf o}\cup\Gamma,\mathcal{O}\rangle,\langle\mathcal{W}^{\sf p},\mathcal{P}\rangle\rangle\rangle$.
        \smallskip
        \item[$(ii)$] $\Gastc$ proves $\Pi\mid\Gamma\specialvdashperm{{\bf S}}{{\bf T}}A$ if and only if $A$ belongs to (at least) one $d$-extension of $\langle\langle\mathcal{W}\cup\Pi,\mathcal{D}\rangle,\langle\langle\mathcal{W}^{\sf o},\mathcal{O}\rangle,\langle\mathcal{W}^{\sf p}\cup\Gamma,\mathcal{P}\rangle\rangle\rangle$.
    \end{itemize}
\end{theorem}

\begin{proof}
    For the proof of each statement, we proceed as in the proof of Theorem \ref{adequacy} -- leveraging Lemma \ref{deonticsemimon}.
\end{proof}

\section{Navigating doxastic and normative conflicts: examples}\label{examples}

In this section, we examine a broad range of deontic scenarios using the proof-theoretic tools introduced so far. We show that controlled sequent calculi are flexible enough to capture deontic notions beyond plain obligations and permissions, and to address paradoxical or dilemmatic situations involving conflicting normative requirements. In certain cases, the $\Gastc$ calculi are extended with appropriate extra-logical rules whose conclusions accord with our intuitive evaluations of the problematic scenarios in question.

To keep the presentation focused, the examples below omit explicit definitions of the underlying default theories and normative systems (we leave these details to the reader). When extending any $\Gastc$ calculus with an extra-logical inference rule, we stipulate that additional extra-logical rules are simultaneously introduced to ensure that the closure conditions specified in Definition \ref{definitioncalculideontic} continue to obtain. This stipulation guarantees that admissibility of safe Cut is preserved. For simplicity, we also assume that constrained obligations and permissions in the underlying normative systems are normal.

\begin{example}[Typicality-based obligations]
Let the extra-logical rules
\smallskip
\begin{center}
    {\AxiomC{$\Gamma\specialvdashdef{{\bf S_{1}}}{{\bf T_{1}}}p$}
    \LeftLabel{\scriptsize{$\delta_{\mathcal{D}_{1}}$}}
    \UnaryInfC{$\Gamma\specialvdashdef{{\bf S'_{1}}}{{\bf T'_{1}}}q$}
    \DisplayProof}\quad
    {\AxiomC{$\Gamma\specialvdashdef{{\bf S_{2}}}{{\bf T_{2}}}q$}
    \RightLabel{\scriptsize{$\delta_{\mathcal{O}_{1}}$}}
    \UnaryInfC{$\Gamma\mid\ \specialvdashnorm{{\bf S'_{2}}}{{\bf T'_{2}}}s$}
    \DisplayProof}\quad
    {\AxiomC{$\Gamma\specialvdashdef{{\bf S_{3}}}{{\bf T_{3}}}r$}
    \RightLabel{\scriptsize{$\delta_{\mathcal{O}_{2}}$}}
    \UnaryInfC{$\Gamma\mid\ \specialvdashnorm{{\bf S'_{3}}}{{\bf T'_{3}}}t$}
    \DisplayProof}
\end{center}
stand for the default rule `If Mary is an adult, then she is employed, unless she is a university student', the conditional obligation `If Mary is employed, she ought to fill in an annual income-tax form' and the conditional obligation `If Mary is a university student, she ought to pay the tuition fee', respectively -- where 
\smallskip
\begin{itemize}
    \item[$(i)$] ${\bf T'_{1}}={\bf T_{1}}\cup\{\{p^{\sf f}\}\}$ and ${\bf S'_{1}}={\bf S_{1}}\cup\{\{(\neg p\vee r)^{\sf f}\}\}$;
    \smallskip
    \item[$(ii)$] ${\bf T'_{2}}={\bf T_{2}}\cup\{\{q^{\sf f}\}\}$ and ${\bf S'_{2}}={\bf S_{2}}\cup\{\{\neg s^{\sf o}\}\}$;
    \smallskip
    \item[$(iii)$] ${\bf T'_{3}}={\bf T_{3}}\cup\{\{r^{\sf f}\}\}$ and ${\bf S'_{3}}={\bf S_{3}}\cup\{\{\neg t^{\sf o}\}\}$.
\end{itemize}
\smallskip
Consider the following derivations:
\smallskip
\begin{center}
    {\AxiomC{$ $}
    \LeftLabel{\scriptsize{$ax$}}
    \UnaryInfC{$p\specialvdashdef{\varnothing}{\varnothing}p$}
    \LeftLabel{\scriptsize{$\delta_{\mathcal{D}_{1}}$}}
    \UnaryInfC{$p\specialvdashdef{{\bf U_{1}}}{{\bf V_{1}}}q$}
    \LeftLabel{\scriptsize{$\delta_{\mathcal{O}_{1}}$}}
    \UnaryInfC{$p\mid\ \specialvdashnorm{{\bf U'_{1}}}{{\bf V'_{1}}}s$}
    \DisplayProof}\quad
    {\AxiomC{$ $}
    \RightLabel{\scriptsize{$ax$}}
    \UnaryInfC{$p\specialvdashdef{\varnothing}{\varnothing}p$}
    \RightLabel{\scriptsize{$LW$}}
    \UnaryInfC{$r,p\specialvdashdef{\varnothing}{\varnothing}p$}
    \RightLabel{\scriptsize{$\delta_{\mathcal{D}_{1}}$}}
    \UnaryInfC{$r,p\specialvdashdef{{\bf U_{1}}}{{\bf V_{1}}}q$}
    \RightLabel{\scriptsize{$\delta_{\mathcal{O}_{1}}$}}
    \UnaryInfC{$r,p\mid\ \specialvdashnorm{{\bf U'_{1}}}{{\bf V'_{1}}}s$}
    \DisplayProof}\quad
    {\AxiomC{$ $}
    \RightLabel{\scriptsize{$ax$}}
    \UnaryInfC{$r\specialvdashdef{\varnothing}{\varnothing}r$}
    \RightLabel{\scriptsize{$LW$}}
    \UnaryInfC{$p,r\specialvdashdef{\varnothing}{\varnothing}r$}
    \RightLabel{\scriptsize{$\delta_{\mathcal{O}_{2}}$}}
    \UnaryInfC{$p,r\mid\ \specialvdashnorm{{\bf U_{2}}}{{\bf V_{2}}}t$}
    \DisplayProof}
\end{center}
\smallskip
where ${\bf V_{1}}=\{\{p^{\sf p}\}\}$, ${\bf U_{1}}=\{\{(\neg p\vee r)^{\sf f}\}\}$, ${\bf V'_{1}}={\bf V_{1}}\cup\{\{q^{\sf f}\}\}$, ${\bf U'_{1}}={\bf U_{1}}\cup\{\{\neg s^{\sf o}\}\}$, ${\bf V_{2}}=\{\{r^{\sf f}\}\}$ and ${\bf U_{2}}=\{\{\neg t^{\sf o}\}\}$. The leftmost derivation is a proof: if Mary is an adult, she is employed by default -- and then she ought to fill in an annual income-tax form. The rightmost derivation is a proof too: if Mary is a university student, she ought to pay the tuition fee. On the other hand, the derivation in the middle is a paraproof: accordingly, if Mary is an adult and a university student, there is no obligation for her to fill in an annual income-tax form.
\end{example}

\begin{example}[Practical syllogism]
    Let the extra-logical rule and axioms
    \smallskip
    \begin{center}
        {\AxiomC{$ $}
        \LeftLabel{\scriptsize{$\delta_{\mathcal{O}_{1}}$}}
        \UnaryInfC{$\Pi\mid\Gamma\specialvdashnorm{{\bf S}}{{\bf T}}\neg p$}
        \DisplayProof}\quad
        {\AxiomC{$ $}
        \RightLabel{\scriptsize{$ax$}}
        \UnaryInfC{$q\specialvdashdef{\varnothing}{\varnothing}\neg p$}
        \DisplayProof}
    \end{center}
    \smallskip
    stand for the unconditional obligation `One ought not pollute' and the factual statement `If one rides a bicycle, one does not pollute', respectively -- where ${\bf T}=\{\{\top^{\sf f}\}\}$ and ${\bf S}=\{\{\neg\neg p^{\sf o}\}\}$. In this scenario, a practical syllogism warrants the conclusion that there is an obligation to cycle. Consider the following inference rule for practical syllogistic reasoning:
    \smallskip
    \begin{center}
        {\AxiomC{$\vdots$}
        \noLine
        \UnaryInfC{$\Pi\mid\Gamma\specialvdashnorm{{\bf S_{1}}}{{\bf T_{1}}}\Delta$}
        \AxiomC{$ $}
        \RightLabel{\scriptsize{$ax$}}
        \UnaryInfC{$A_{1},\ldots,A_{m}\specialvdashdef{\varnothing}{\varnothing}\Delta$}
        \RightLabel{\scriptsize{$prsyl$}}
        \BinaryInfC{$\Pi\mid\Gamma\specialvdashnorm{{\bf S}}{{\bf T_{1}}}A_{i}$}
        \DisplayProof}
    \end{center}
    \smallskip
    with ${\bf S}={\bf S_{1}}\cup\{\{\neg A_{i}^{\sf o}\}\}$, $1\leq i\leq m$ and the side condition that $\fullGpn$ refutes $A_{1},\ldots,A_{m}\vdash\Delta$. If we extend the $\Gastc$ calculus for our scenario with $prsyl$, the following derivation becomes available:
    \smallskip
    \begin{center}
        {\AxiomC{$ $}
        \LeftLabel{\scriptsize{$\delta_{\mathcal{O}_{1}}$}}
        \UnaryInfC{$\varnothing\mid\ \specialvdashnorm{{\bf S}}{{\bf T}}\neg p$}
        \AxiomC{$ $}
        \RightLabel{\scriptsize{$ax$}}
        \UnaryInfC{$q\specialvdashdef{\varnothing}{\varnothing}\neg p$}
        \RightLabel{\scriptsize{$prsyl$}}
        \BinaryInfC{$\varnothing\mid\ \specialvdashnorm{{\bf S'}}{{\bf T}}q$}
        \DisplayProof}
    \end{center}
    \smallskip
    with ${\bf S'}={\bf S}\cup\{\{\neg q^{\sf o}\}\}$. Such derivation is a proof. This corresponds to the fact that if there exists an obligation not to pollute and cycling contributes to avoiding pollution, there exists an obligation to cycle. 
\end{example}

\begin{example}[Chisholm's paradox]
    Let the extra-logical rules
    \smallskip
    \begin{center}
        {\AxiomC{$ $}
        \LeftLabel{\scriptsize{$\delta_{\mathcal{O}_{1}}$}}\UnaryInfC{$\Pi\mid\Gamma\specialvdashnorm{{\bf S_{1}}}{{\bf T_{1}}}p$}
        \DisplayProof}\quad
        {\AxiomC{$\Gamma\specialvdashdef{{\bf S_{2}}}{{\bf T_{2}}} p$}
        \RightLabel{\scriptsize{$\delta_{\mathcal{O}_{2}}$}}
        \UnaryInfC{$\Gamma\mid\ \specialvdashnorm{{\bf S'_{2}}}{{\bf T'_{2}}} q$}
        \DisplayProof}\quad
        {\AxiomC{$\Gamma\specialvdashdef{{\bf S_{3}}}{{\bf T_{3}}} \neg p$}
        \RightLabel{\scriptsize{$\delta_{\mathcal{O}_{3}}$}}
       \UnaryInfC{$\Gamma\mid\ \specialvdashnorm{{\bf S'_{3}}}{{\bf T'_{3}}} \neg q$}
        \DisplayProof}
    \end{center}
    \smallskip
    stand for the (un)conditional obligations: `The manager ought to review the budget', `If the manager reviews the budget, then she ought to submit it for approval', and `If the manager does not review the budget, then she ought not submit it for approval', respectively \cite{Chis63} -- where 
    \smallskip
    \begin{itemize}
        \item[$(i)$] ${\bf T_{1}}=\{\{\top^{{\sf f}}\}\}$ and ${\bf S_{1}}=\varnothing$;
        \smallskip
        \item[$(ii)$] ${\bf T'_{2}}={\bf T_{2}}\cup\{\{p^{\sf f}\}\}$ and ${\bf S'_{2}}={\bf S_{2}}\cup\{\{\neg q^{\sf o}\}\}$;
        \smallskip
        \item[$(iii)$] ${\bf T'_{3}}={\bf T_{3}}\cup\{\{\neg p^{\sf f}\}\}$ and ${\bf S'_{3}}={\bf S_{3}}\cup\{\{\neg\neg q^{\sf o}\}\}$.
    \end{itemize}
    \smallskip
    There exists an unconditional obligation to $p$, and the latter is equivalent to the negation of the condition of the norm captured by the rule -- namely, $\neg p$. This corresponds to the scenario where there exists a primary obligation and a violation of the latter triggers a secondary obligation. The intuitive assessment is that the secondary obligation should hold, in spite of the fact that the primary obligation is violated.  

    To formalize such contrary-to-duty obligation, we require ${\sf O}$-labelled conclusions of extra-logical rules to be consistent with facts -- that is, extra-logical axioms in $\mathcal{W}$, conclusions of defaults applied in the derivations of ${\sf F}$-labelled premises and factual assumptions in the repository. For any derivation $\pi$ of an ${\sf O}$-labelled conclusion, we shall use $F_{\pi}$ to denote the conjunction of formulas in $\mathcal{W}$, $concl(def(\pi))$ and in the repository. Hence, we rewrite $\delta_{\mathcal{O}_{1}},\delta_{\mathcal{O}_{2}}$ and $\delta_{\mathcal{O}_{3}}$ according to the following conditions:
    \smallskip
    \begin{itemize}
        \item[$(i')$] ${\bf T_{1}}=\{\{\top^{{\sf f}}\}\}$ and ${\bf S_{1}}=\{\{\neg F_{\pi}^{\sf o}\}\}$;
        \smallskip
        \item[$(ii')$] ${\bf T'_{2}}={\bf T_{2}}\cup\{\{p^{\sf f}\}\}$ and ${\bf S'_{2}}={\bf S_{2}}\cup\{\{\neg q^{\sf o}\},\{\neg F_{\pi}^{\sf o}\}\}$;
        \smallskip
        \item[$(iii')$] ${\bf T'_{3}}={\bf T_{3}}\cup\{\{\neg p^{\sf f}\}\}$ and ${\bf S'_{3}}={\bf S_{3}}\cup\{\{\neg\neg q^{\sf o}\},\{\neg F_{\pi}^{\sf o}\}\}$
    \end{itemize}
    \smallskip
    Now, consider the following derivation:
\smallskip
\begin{center}
    {\AxiomC{$ $}
    \LeftLabel{\scriptsize{$\delta_{\mathcal{O}_{1}}$}}
    \UnaryInfC{$\varnothing\mid\ \specialvdashnorm{{\bf S_{1}}}{{\bf T_{1}}} p$}
    \AxiomC{$ $}
    \RightLabel{\scriptsize{$ax$}}
    \UnaryInfC{$\neg p\specialvdashdef{{\bf S_{3}}}{{\bf T_{3}}} \neg p$}
    \RightLabel{\scriptsize{$\delta_{\mathcal{O}_{3}}$}}
    \UnaryInfC{$\neg p\mid\ \specialvdashnorm{{\bf S'_{3}}}{{\bf T'_{3}}} \neg q$}
    \RightLabel{\scriptsize{$R\wedge^{\sf O}$}}
    \BinaryInfC{$\neg p\mid\ \specialvdashnorm{{\bf S}}{{\bf T}}\, p\wedge\neg q$}
    \DisplayProof}
\end{center}
\smallskip
\noindent Such derivation is a paraproof, whereas its immediate subderivations are proofs. Accordingly, if the manager does not review the budget, she ought not submit it -- even though the primary obligation to review it still holds under no factual assumption.
\remove{Consider the following application of a suitable inference rule $ctd$:
    \smallskip
    \begin{center}
        {\AxiomC{$ $}
        \LeftLabel{\scriptsize{$\delta_{\mathcal{O}_{1}}$}}
        \UnaryInfC{$\varnothing\mid\ \specialvdashnorm{{\bf S_{1}}}{{\bf T_{1}}} p$}
        \AxiomC{$ $}
        \RightLabel{\scriptsize{$ax$}}
        \UnaryInfC{$\neg p\specialvdashdef{{\bf S_{3}}}{{\bf T_{3}}} \neg p$}
        \RightLabel{\scriptsize{$\delta_{\mathcal{O}_{3}}$}}
        \UnaryInfC{$\neg p\mid\ \specialvdashnorm{{\bf S'_{3}}}{{\bf T'_{3}}} \neg q$}
        \RightLabel{\scriptsize{$ctd$}}
        \BinaryInfC{$\neg p\mid\ \specialvdashnorm{{\bf S}}{{\bf T}}p\wedge\neg q$}
        \DisplayProof}
    \end{center}
    \smallskip
    with ${\bf T_{3}}={\bf S_{3}}=\varnothing$, ${\bf T}={\bf T_{1}}\cup{\bf T'_{3}}$, ${\bf S}={\bf S_{1}}\cup{\bf S''_{3}}\cup\{\{p^{\sf o}\}\}$. Such derivation is a paraproof, whereas its immediate subderivations are proofs. Accordingly, if the manager does not review the budget, then she ought not submit it for approval -- in spite of the existence of an unconditional obligation to reviewing the budget.}
\end{example}

\begin{example}[Forrester's paradox]
    Let the extra-logical rules
    \smallskip
    \begin{center}
        {\AxiomC{$ $}
        \LeftLabel{\scriptsize{$\delta_{\mathcal{O}_{1}}$}}
        \UnaryInfC{$\Pi\mid\Gamma\specialvdashnorm{{\bf S_{1}}}{{\bf T_{1}}}\neg p$}
        \DisplayProof}\quad
        {\AxiomC{$\Gamma\specialvdashdef{{\bf S_{2}}}{{\bf T_{2}}}p$}
        \RightLabel{\scriptsize{$\delta_{\mathcal{O}_{2}}$}}
        \UnaryInfC{$\Gamma\mid\ \specialvdashnorm{{\bf S'_{2}}}{{\bf T'_{2}}}q$}
        \DisplayProof}
    \end{center}
    \smallskip
    stand for the (un)conditional obligations `One ought not kill' and `If one kills someone, she ought to kill her gently', respectively \cite{For84} -- where 
    \smallskip
    \begin{itemize}
        \item[$(i)$] ${\bf T_{1}}=\{\{\top^{\sf f}\}\}$ and ${\bf S_{1}}=\{\{\neg\neg p^{o}\},\{\neg F_{\pi}^{\sf o}\}\}$;
        \smallskip
        \item[$(ii)$] ${\bf T'_{2}}={\bf T_{2}}\cup\{\{p^{\sf f}\}\}$ and ${\bf S'_{2}}={\bf S_{2}}\cup\{\{\neg q^{o}\},\{\neg F_{\pi}^{\sf o}\}\}$.
    \end{itemize}
    \smallskip
    Consider the following derivation:
    \smallskip
    \begin{center}
        {\AxiomC{$ $}
        \LeftLabel{\scriptsize{$\delta_{\mathcal{O}_{1}}$}}
        \UnaryInfC{$\varnothing\mid\ \specialvdashnorm{{\bf S_{1}}}{{\bf T_{1}}}\neg p$}
        \AxiomC{$ $}
        \RightLabel{\scriptsize{$ax$}}
        \UnaryInfC{$p\specialvdashdef{{\bf S_{2}}}{{\bf T_{2}}}p$}
        \RightLabel{\scriptsize{$\delta_{\mathcal{O}_{2}}$}}
        \UnaryInfC{$p\mid\ \specialvdashnorm{{\bf S'_{2}}}{{\bf T'_{2}}}q$}
        \RightLabel{\scriptsize{$R\wedge^{\sf O}$}}
        \BinaryInfC{$p\mid\ \specialvdashnorm{{\bf S}}{{\bf T}}\neg p\wedge q$}
        \DisplayProof}
    \end{center}
    \smallskip
    with ${\bf T_{2}}={\bf S_{2}}=\varnothing$, ${\bf T}={\bf T'_{1}}\cup{\bf T'_{2}}$, ${\bf S}={\bf S_{1}}\cup{\bf S'_{2}}\cup\{\{\neg p^{\sf o}\}\}$. Such derivation is a paraproof, whereas its immediate subderivations are proofs. This corresponds to the fact that if one kills someone, she ought to kill them gently -- in spite of the existence of an unconditional obligation not to kill anyone under the empty set of factual assumptions.
\end{example}

\remove{\begin{example}[Pragmatic Oddity paradox]
    Let the extra-logical axiom and rules
    \smallskip
    \begin{center}
        {\AxiomC{$ $}
        \LeftLabel{\scriptsize{$ax$}}
        \UnaryInfC{$\specialvdashdef{\varnothing}{\varnothing}p$}
        \DisplayProof}\quad
        {\AxiomC{$ $}
        \RightLabel{\scriptsize{$\delta_{\mathcal{O}_{1}}$}}
        \UnaryInfC{$\Pi\mid\Gamma\specialvdashnorm{{\bf S_{1}}}{{\bf T_{1}}}\neg p$}
        \DisplayProof}\quad
        {\AxiomC{$\Gamma\specialvdashdef{{\bf S_{2}}}{{\bf T_{2}}}p$}
        \RightLabel{\scriptsize{$\delta_{\mathcal{O}_{2}}$}}
        \UnaryInfC{$\Gamma\mid\Pi\specialvdashnorm{{\bf S_{2}'}}{{\bf T_{2}'}}q$}
        \DisplayProof}
    \end{center}
    stand for the factual statement `There is a dog', the (un)conditional obligations `There should be no dog' and `If there is a dog, there ought to be a warning sign', respectively \cite{PS96} -- where
    \smallskip
    \begin{itemize}
        \item[$(i)$] 
        \smallskip
        \item[$(ii)$]
    \end{itemize}
    \smallskip
\end{example}}

\begin{example}[Specificity principle]
    Let the extra-logical rules
    \smallskip
    \begin{center}
        {\AxiomC{$\Gamma\specialvdashdef{{\bf S_{1}}}{{\bf T_{1}}}p $}
        \LeftLabel{\scriptsize{$\delta_{\mathcal{O}_{1}}$}}
        \UnaryInfC{$\Gamma\mid\ \specialvdashnorm{{\bf S'_{1}}}{{\bf T'_{1}}}\neg q$}
        \DisplayProof}\quad
        {\AxiomC{$\Gamma\specialvdashdef{{\bf S_{2}}}{{\bf T_{2}}}p$}
        \AxiomC{$\Gamma\specialvdashdef{{\bf S'_{2}}}{{\bf T'_{2}}}r$}
        \RightLabel{\scriptsize{$\delta_{\mathcal{O}_{2}}$}}
        \BinaryInfC{$\Gamma\mid\ \specialvdashnorm{{\bf S''_{2}}}{{\bf T''_{2}}}q$}
        \DisplayProof}
    \end{center}
    \smallskip
    stand for the conditional obligations `If Fido is a dog, you ought not kill Fido' and `If Fido is a dog and Fido is attacking a child, you ought to kill Fido', respectively -- where
    \smallskip
    \begin{itemize}
        \item[$(i)$] ${\bf T'_{1}}={\bf T_{1}}\cup\{\{p^{\sf f}\}\}$ and ${\bf S'_{1}}={\bf S_{1}}\cup\{\{\neg\neg q^{\sf o}\}\}$;
        \smallskip
        \item[$(ii)$] ${\bf T''_{2}}={\bf T_{2}}\cup{\bf T'_{2}}\cup\{\{p^{\sf f},r^{\sf f}\}\}$ and ${\bf S''_{2}}={\bf S_{2}}\cup{\bf S'_{2}}\cup\{\{\neg q^{\sf o}\}\}$.
    \end{itemize}
    \smallskip
    If Fido is a dog and is gravely attacking a child, one ought to kill him -- an obligation that overrides the less specific obligation not to kill him merely because he is a dog \cite{vdTT95specific}.

  To formalize scenarios where the specificity principle applies, we compare obligations pairwise and require that the conditions of more specific obligations appear as factual constraints in less specific ones. Thus, we rewrite  $\delta_{\mathcal{O}_{1}}$ under the following condition:
    \smallskip
    \begin{itemize}
        \item[$(i')$] ${\bf T'_{1}}={\bf T_{1}}\cup\{\{p^{\sf f}\}\}$ and ${\bf S'_{1}}={\bf S_{1}}\cup\{\{\neg\neg q^{\sf o}\},\{r\wedge p^{\sf f}\}\}$.
    \end{itemize}
    \smallskip
    Now, consider the following derivation:
    \smallskip
    \begin{center}
        {\AxiomC{$ $}
        \RightLabel{\scriptsize{$ax$}}
        \UnaryInfC{$p\specialvdashdef{\varnothing}{\varnothing}p$}
        \RightLabel{\scriptsize{$LW$}}
        \UnaryInfC{$r,p\specialvdashdef{\varnothing}{\varnothing}p$}
        \RightLabel{\scriptsize{$\delta_{\mathcal{O}_{1}}$}}
        \UnaryInfC{$r,p\mid\ \specialvdashnorm{{\bf U}}{{\bf V}}\neg q$}
        \DisplayProof}
    \end{center}
    \smallskip
    with ${\bf V}=\{\{p^{\sf f}\}\}$ and ${\bf U}=\{\{\neg\neg q^{\sf o}\},\{r\wedge p^{\sf f}\}\}$. Such a derivation constitutes a paraproof. This matches our intuitive assessment: if Fido is a dog and is posing extreme danger to a child, one ought to kill him, even though there is a weaker obligation not to kill him under the mere condition that he is a dog.
    \remove{Consider the following application of a suitable inference rule $(\delta_{\mathcal{O}_{1}})^{sp}$:
    \smallskip
    \begin{center}
        {\AxiomC{$ $}
        \RightLabel{\scriptsize{$ax$}}
        \UnaryInfC{$p\specialvdashdef{\varnothing}{\varnothing}p$}
        \RightLabel{\scriptsize{$LW$}}
        \UnaryInfC{$r,p\specialvdashdef{\varnothing}{\varnothing}p$}
        \RightLabel{\scriptsize{$(\delta_{\mathcal{O}_{1}})^{sp}$}}
        \UnaryInfC{$r,p\mid\ \specialvdashnorm{{\bf U}}{{\bf V}}\neg q$}
        \DisplayProof}
    \end{center}
    \smallskip
    with ${\bf V}=\{\{p^{\sf f}\}\}$ and ${\bf U}=\{\{\neg\neg q^{\sf o}\},\{r^{\sf f}\}\}$. Such derivation is a paraproof. This corresponds to the fact that if Fido is a dog and is attacking a child, one ought to kill him -- and this obligation overrides the obligation not to kill him under the weaker condition that he is just a dog.}
\end{example}

\begin{example}[Extended Forrester's paradox]
    Let the extra-logical rules
    \smallskip
    \begin{center}
        {\AxiomC{$ $}
        \LeftLabel{\scriptsize{$\delta_{\mathcal{O}_{1}}$}}
        \UnaryInfC{$\Pi\mid\Gamma\specialvdashnorm{{\bf S_{1}}}{{\bf T_{1}}}\neg p$}
        \DisplayProof}\quad
        {\AxiomC{$\Gamma\specialvdashdef{{\bf S_{2}}}{{\bf T_{2}}}p$}
        \RightLabel{\scriptsize{$\delta_{\mathcal{O}_{2}}'$}}
        \UnaryInfC{$\Gamma\mid\ \specialvdashnorm{{\bf S'_{2}}}{{\bf T'_{2}}}p$}
        \DisplayProof}\quad
        {\AxiomC{$\Gamma\specialvdashdef{{\bf S_{2}}}{{\bf T_{2}}}p$}
        \RightLabel{\scriptsize{$\delta_{\mathcal{O}_{2}}''$}}
        \UnaryInfC{$\Gamma\mid\ \specialvdashnorm{{\bf S'_{2}}}{{\bf T'_{2}}}q$}
        \DisplayProof}
    \end{center}
    \smallskip
    \begin{center}
         {\AxiomC{$\Gamma\specialvdashdef{{\bf S_{3}}}{{\bf T_{3}}}r$}
        \RightLabel{\scriptsize{$\delta_{\mathcal{O}_{3}}'$}}
        \UnaryInfC{$\Gamma\mid\ \specialvdashnorm{{\bf S'_{3}}}{{\bf T'_{3}}}p$}
        \DisplayProof}\quad
        {\AxiomC{$\Gamma\specialvdashdef{{\bf S_{3}}}{{\bf T_{3}}}r$}
        \RightLabel{\scriptsize{$\delta_{\mathcal{O}_{3}}''$}}
        \UnaryInfC{$\Gamma\mid\ \specialvdashnorm{{\bf S'_{3}}}{{\bf T'_{3}}}q$}
        \DisplayProof}
    \end{center}
    \smallskip
    stand for the (un)conditional obligations `There ought not be a fence around your house', `If there is a fence around your house, there ought be a white fence around your house' and `If you have a dog, there ought be a white fence around your house', respectively \cite{vdTT97} -- where 
    \smallskip
    \begin{itemize}
        \item[$(i)$] ${\bf T_{1}}=\{\{\top^{\sf f}\}\}$ and ${\bf S_{1}}=\{\{\neg\neg p^{\sf o}\},\{r^{\sf f}\}\}$;
        \smallskip
        \item[$(ii)$] ${\bf T'_{2}}={\bf T_{2}}\cup\{\{p^{\sf f}\}\}$, ${\bf S'_{2}}={\bf S_{2}}\cup\{\{\neg (p\wedge q)^{\sf o}\},\{\neg F_{\pi}^{\sf o}\}\}$;
        \smallskip
        \item[$(iii)$] ${\bf T'_{3}}={\bf T_{3}}\cup\{\{r^{\sf f}\}\}$, ${\bf S'_{3}}={\bf S_{3}}\cup\{\{\neg (p\wedge q)^{\sf o}\},\{\neg F_{\pi}^{\sf o}\}\}$.
    \end{itemize}
    \smallskip
    Consider the following derivations:
    \smallskip
    \begin{center}
         {\AxiomC{$ $}
        \LeftLabel{\scriptsize{$\delta_{\mathcal{O}_{1}}$}}
        \UnaryInfC{$r\mid\ \specialvdashnorm{{\bf S'_{1}}}{{\bf T_{1}}}\neg p$}
        \DisplayProof}\quad
        {\AxiomC{$ $}
        \UnaryInfC{$\varnothing\mid\ \specialvdashnorm{{\bf S_{1}}}{{\bf T_{1}}}\neg p$}
        \AxiomC{$ $}
        \RightLabel{\scriptsize{$ax$}}
        \UnaryInfC{$p\specialvdashdef{\varnothing}{\varnothing}p$}
        \RightLabel{\scriptsize{$\delta_{\mathcal{O}_{2}}'$}}
        \UnaryInfC{$p\mid\ \specialvdashnorm{{\bf U_{1}}}{{\bf V_{1}}}p$}
        \AxiomC{$ $}
        \RightLabel{\scriptsize{$ax$}}
        \UnaryInfC{$p\specialvdashdef{\varnothing}{\varnothing}p$}
        \RightLabel{\scriptsize{$\delta_{\mathcal{O}_{2}}''$}}
        \UnaryInfC{$p\mid\ \specialvdashnorm{{\bf U_{1}}}{{\bf V_{1}}}q $}
        \BinaryInfC{$p\mid\ \specialvdashnorm{{\bf U_{1}}}{{\bf V_{1}}}p\wedge q$}
        \RightLabel{\scriptsize{$R\wedge^{\sf O}$}}
        \BinaryInfC{$p\mid\ \specialvdashnorm{{\bf U}}{{\bf V}}\neg p\wedge(p\wedge q) $}
        \DisplayProof}
    \end{center}
    \smallskip
    where ${\bf V_{1}}=\{\{p^{\sf f}\}\}$, ${\bf U_{1}}=\{\{\neg (p\wedge q)^{\sf o}\},\{\neg p^{\sf o}\}\}$, ${\bf V}={\bf T_{1}}\cup{\bf V_{1}}$ and ${\bf U}={\bf S_{1}}\cup{\bf U_{1}}$. It is easy to verify that both derivations are paraproofs, and that both subderivations in the rightmost derivation are proofs. This corresponds to the fact that (a) the unconditional prohibition of having a fence around your house is overridden by the obligation of having a fence in case you possess a dog, and that (b) when the unconditional prohibition of having a fence around your house is violated, a secondary obligation is triggered for you to have a white one.
\end{example}

\begin{example}[Euthyphro's dilemma]\label{euthyphro}
    Euthyphro is prosecuting his father for murder, despite the traditional belief that it is impious to prosecute or dishonor a parent. Let the extra-logical axiom and rules 
    \smallskip
    \begin{center}
        {\AxiomC{$ $}
        \LeftLabel{\scriptsize{$ax$}}
        \UnaryInfC{$\specialvdashdef{\varnothing}{\varnothing}\neg p,q$}
        \DisplayProof}\quad
        {\AxiomC{$ $}
        \LeftLabel{\scriptsize{$\delta_{\mathcal{O}_{1}}$}}\UnaryInfC{$\Pi\mid\Gamma\specialvdashnorm{{\bf S_{1}}}{{\bf T_{1}}}p$}
        \DisplayProof}\quad
        {\AxiomC{$ $}
        \RightLabel{\scriptsize{$\delta_{\mathcal{O}_{2}}$}}
        \UnaryInfC{$\Pi\mid\Gamma\specialvdashnorm{{\bf S_{2}}}{{\bf T_{2}}}\neg q$}
        \DisplayProof}
    \end{center}
    \smallskip
    stand for the factual statement `If one prosecutes her father for murder, one dishonors him' and the unconditional obligations: `One ought to prosecute her father for murder' and `One ought not to dishonor her father', respectively -- where ${\bf T_{1}}={\bf T_{2}}=\{\{\top^{\sf f}\}\}$, ${\bf S_{1}}=\{\{\neg p^{\sf o}\}\}$ and ${\bf S_{2}}=\{\{\neg\neg q^{\sf o}\}\}$. 

    To formalize such dilemmatic scenario, we require ${\sf O}$-labelled conclusions of extra-logical rules to be compatible under the background of factual extra-logical axioms. Hence, we rewrite $\delta_{\mathcal{O}_{1}}$ and $\delta_{\mathcal{O}_{2}}$ so that ${\bf S_{1}}=\{\{\neg p^{\sf o}\},\{\neg F_{\pi}^{\sf o}\}\}$ and ${\bf S_{2}}=\{\{\neg\neg q^{\sf o}\},\{\neg F_{\pi}^{\sf o}\}\}$. Now, consider the following derivation:
    \smallskip
    \begin{center}
        {\AxiomC{$ $}
        \LeftLabel{\scriptsize{$\delta_{\mathcal{O}_{1}}$}}
        \UnaryInfC{$\varnothing\mid\ \specialvdashnorm{{\bf S_{1}}}{{\bf T_{1}}}p$}
        \AxiomC{$ $}
        \RightLabel{\scriptsize{$\delta_{\mathcal{O}_{2}}$}}
        \UnaryInfC{$\varnothing\mid\ \specialvdashnorm{{\bf S_{2}}}{{\bf T_{2}}}\neg q$}
        \RightLabel{\scriptsize{$R\wedge^{\sf O}$}}
        \BinaryInfC{$\varnothing\mid\ \specialvdashnorm{{\bf S}}{{\bf T}}p\wedge\neg q$}
        \DisplayProof}
    \end{center}
    \smallskip
    with ${\bf T}={\bf T_{1}}\cup{\bf T_{2}}$ and ${\bf S}={\bf S_{1}}\cup{\bf S_{2}}$. Such derivation is a paraproof. This corresponds to the fact that Eutyphro cannot be obliged to prosecute his father and not to dishonor him at the same time.
\end{example}

There are noteworthy deontic notions that emerge from the non-trivial interaction among facts, (un)conditioned (constrained) obligations, and permissions. In the following examples, we approach these non-primitive deontic notions through the lens of controlled calculi. In certain cases, we allow control pairs to {\em decrease} in size along derivations, contrary to what happens under standard extra-logical rules in $\Gastc$ calculi.


\begin{example}[Obligations with exceptions]
Let the extra-logical rule
\smallskip
\begin{center}
    {\AxiomC{$\Gamma\specialvdashdef{{\bf S}}{{\bf T}}p$}
    \RightLabel{\scriptsize{$\delta_{\mathcal{O}'}$}}
    \UnaryInfC{$\Gamma\mid\ \specialvdashnorm{{\bf S'}}{{\bf T'}}\neg q$}
    \DisplayProof}
\end{center}
\smallskip
stand for the conditional obligation `If one is served a meal, one ought not eat with fingers -- provided the meal is not asparagus' \cite{Horty97}, where ${\bf T'}={\bf T}\cup\{\{p^{\sf f}\}\}$ and ${\bf S'}={\bf S}\cup\{\{\neg\neg q^{\sf o}\},\{\neg r^{\sf f}\}\}$. Consider the following derivation:
\smallskip
\begin{center}
    {\AxiomC{$ $}
    \RightLabel{\scriptsize{$ax$}}
    \UnaryInfC{$p\specialvdashdef{\varnothing}{\varnothing}p$}
    \RightLabel{\scriptsize{$LW$}}
    \UnaryInfC{$r,p\specialvdashdef{\varnothing}{\varnothing}p$}
    \RightLabel{\scriptsize{$\delta_{\mathcal{O}'}$}}
    \UnaryInfC{$r,p\mid\ \specialvdashnorm{{\bf U}}{{\bf V}}\neg q$}
    \DisplayProof}
\end{center}
\smallskip
with ${\bf V}=\{\{p^{\sf f}\}\}$ and ${\bf U}=\{\{\neg\neg q^{\sf o}\},\{\neg r^{\sf f}\}\}$. Such derivation is a paraproof. This corresponds to the fact that if one is served asparagus, the obligation not to eat with fingers ceases to apply.
\end{example}

\begin{example}[Violations and sanctions]
    Let the extra-logical rule
    \smallskip
    \begin{center}
        {\AxiomC{$ $}
        \LeftLabel{\scriptsize{$\delta_{\mathcal{O}_{1}}$}}
        \UnaryInfC{$\Pi\mid\ \specialvdashnorm{{\bf S_{1}}}{{\bf T_{1}}}\neg p$}
        \DisplayProof}\quad
        {\AxiomC{$\Pi\specialvdashdef{{\bf S'_{2}}}{{\bf T'_{2}}}p$}
        \AxiomC{$\Pi\mid\Gamma\specialvdashnorm{{\bf S''_{2}}}{{\bf T''_{2}}}\neg p $}
        \RightLabel{\scriptsize{$\delta_{\mathcal{O}_{2}}$}}
        \BinaryInfC{$\Pi\mid\Gamma\specialvdashnorm{{\bf S_{2}}}{{\bf T_{2}}}q$}
        \DisplayProof}
    \end{center}
    \smallskip
    stand for `One ought not double-park' and `If one double-parks and she ought not double-park, she ought to pay a fine', respectively -- where
    \smallskip
    \begin{itemize}
        \item[$(i)$] ${\bf T_{1}}=\{\{\top^{\sf f}\}\}$ and ${\bf S_{1}}=\varnothing$;
        \smallskip
        \item[$(ii)$]  ${\bf T_{2}}={\bf T'_{2}}\cup{\bf T''_{2}}\cup\{\{p^{\sf f}\},\{\neg p^{\sf o}\}\}$ and ${\bf S_{2}}={\bf S'_{2}}\cup{\bf S''_{2}}\cup\{\{\neg q^{\sf o}\}\}$.
    \end{itemize}
    \smallskip
    Consider the following derivation:
    \smallskip
    \begin{center}
        {\AxiomC{$\vdots$}
        \noLine
        \UnaryInfC{$\Pi\specialvdashdef{{\bf S'_{2}}}{{\bf T'_{2}}}p$}
        \AxiomC{$ $}
        \RightLabel{\scriptsize{$\delta_{\mathcal{O}_{1}}$}}
        \UnaryInfC{$\Pi\mid\ \specialvdashnorm{{\bf S_{1}}}{{\bf T_{1}}}\neg p$}
        \RightLabel{\scriptsize{$\delta_{\mathcal{O}_{2}}$}}
        \BinaryInfC{$\Pi\mid\Gamma\specialvdashnorm{{\bf S}}{{\bf T}}q$}
        \DisplayProof}
    \end{center}
    \smallskip
    with ${\bf T}={\bf T'_{2}}\cup{\bf T_{1}}\cup\{\{p^{\sf f}\},\{\neg p^{\sf o}\}\}$ and ${\bf S}={\bf S'_{2}}\cup{\bf S_{1}}\cup\{\{\neg q^{\sf o}\}\}$. Such derivation is a proof. This corresponds to the fact that if the obligation of not double-parking is actual, then the obligation of paying a fine is actual.
\end{example}

\begin{example}[Guarded free choice permission]
    Let the extra-logical axiom and rule
    \smallskip
    \begin{center}
    {\AxiomC{$ $}
        \LeftLabel{\scriptsize{$ax^{\sf P}$}}
        \UnaryInfC{$\Pi\mid\ \specialvdashperm{\varnothing}{\varnothing}p,q,\neg r$}
        \DisplayProof}\quad
     {\AxiomC{$ $}
        \RightLabel{\scriptsize{$\delta_{\mathcal{O}'}$}}
        \UnaryInfC{$\Pi\mid\Gamma\specialvdashnorm{{\bf S_{1}}}{{\bf T_{1}}}r$}
        \DisplayProof}
    \end{center}
    \smallskip
    stand for the unconditional permission `One is permitted to work or relax or not pay the bill' and the unconditional obligation `One ought to pay the bill', respectively -- where ${\bf T_{1}}=\{\{\top^{\sf f}\}\}$ and ${\bf S_{1}}=\{\{\neg r^{\sf o}\}\}$.
   
    The free choice principle allows to infer that a formula $A_{i}$ is permitted whenever $A_{1}\vee\cdots\vee A_{m}$ is permitted, for any $1\leq i\leq m$ \cite[p. 21]{wrightbook}. Even when limiting ourselves to consider disjunctions obtained without any application of $RW$, unrestricted application of the free choice principle leads to permission explosion \cite{AB05}. In our scenario, the free choice principle ensures the existence of a permission not to pay the bill, whereas the obligation to pay the bill entails the permission to pay the bill (since $\mathcal{O}\subseteq\mathcal{P}$ in the underlying normative system). 

    The unrestricted free choice principle over unconstrained permissions is captured by closing ${\sf P}$-labelled extra-logical axioms under the rule
    \smallskip
    \begin{center}
    {\AxiomC{$\Pi\mid\ \specialvdashperm{\varnothing}{\varnothing}A_{1},\ldots,A_{m}$}
        \RightLabel{\scriptsize{$fcp$}}
        \UnaryInfC{$\Pi\mid\ \specialvdashperm{\varnothing}{\varnothing}A_{i}$}
        \DisplayProof}
    \end{center}
    \smallskip
    To block undesired permissions, we close ${\sf P}$-labelled extra-logical axioms under the rule
    \smallskip
    \begin{center}
    {\AxiomC{$\Pi\mid\ \specialvdashperm{\varnothing}{\varnothing}A_{1},\ldots,A_{m}$}
        \RightLabel{\scriptsize{$fcp$}}
        \UnaryInfC{$\Pi\mid\ \specialvdashperm{{\bf S}}{\varnothing}A_{i}$}
        \DisplayProof}
    \end{center}
    \smallskip
    where ${\bf S}=\{\{\neg O_{u},\neg O_{c}\}\}$, with $O_{u}$ being the conjunction of all formulas in $\mathcal{W}^{\sf o}$ and $O_{c}$ being the conjunction of all formulas in $concl(\mathcal{O})$. 
\end{example}

\begin{example}[Permissions as exceptions]
Let the extra-logical rule
\smallskip
\begin{center}
    {\AxiomC{$ $}
    \RightLabel{\scriptsize{$\delta_{\mathcal{O}'}$}}
    \UnaryInfC{$\Pi\mid\Gamma\specialvdashnorm{{\bf S}}{{\bf T}}\neg p$}
    \DisplayProof}
\end{center}
\smallskip
stand for the unconditional obligation `One ought not kill herself' -- where ${\bf T}=\{\{\top^{\sf f}\}\}$ and ${\bf S}=\{\{\neg\neg p^{\sf o}\}\}$. The Torah contains such obligation as a commandment (cf. Genesis 9:5). However, some Rabbinic scholars state that suicide (or surrender to death) is permissible when one is threatened with conversion (as in the case of King Saul in 2 Samuel, 1:5-10). Let the extra-logical rule
\smallskip
\begin{center}
    {\AxiomC{$ $}
    \RightLabel{\scriptsize{$\delta_{\mathcal{O}'}'$}}
    \UnaryInfC{$\Pi\mid\Gamma\specialvdashnorm{{\bf S'}}{{\bf T}}\neg p$}
    \DisplayProof}
\end{center}
\smallskip
stand for the unconditional obligation with exceptions `One ought not kill herself, provided that she is not threatened with conversion' -- where ${\bf S'}={\bf S}\cup\{\{q^{\sf f}\}\}$. To infer the positive permission to kill oneself, we can apply the following rule $pe$
\smallskip
\begin{center}
    {\AxiomC{$ $}
    \RightLabel{\scriptsize{$\delta_{\mathcal{O}'}'$}}
    \UnaryInfC{$\varnothing\mid\ \specialvdashnorm{{\bf S'}}{{\bf T}}\neg p$}
    \RightLabel{\scriptsize{$pe_{\mathcal{W}^{\sf p\star}}$}}
    \UnaryInfC{$q\mid\ \specialvdashperm{{\bf S}}{{\bf T}}p$}
    \DisplayProof}
\end{center}
\smallskip
with $\mathcal{W}^{\sf p\star}=\Big\{\dfrac{\top:\top}{\neg p}\Big\}$.
\remove{{\color{red}da riscrivere} {\color{blue}
    The Torah contains the commandment `One ought not kill herself' (cf. Genesis 9:5). Even the commandment to observe the Sabbath is overridden by the commandment of saving human life. Rav Tam and Rav Ritba state that suicide (or surrender to death) is permissible when one is threatened with conversion (as in the case of King Saul in 2 Samuel, 1:5-10), so that suicide may be justified for the sake of {\em kiddush ha-Shem}, the sanctification of God's name. Rav Meln, by contrast, maintains that a woman may convert to Christianity to save her life, provided that she continues to regard herself as Jewish.  Finally, during the German occupation of the Warsaw Ghetto, Rav Nissenbaum declared that suicide (or surrender to death) is impermissible even when one is compelled to desecrate the holy books, invoking instead the principle of {\em kiddush ha-Hayyim}, the sanctification of life.}}
\end{example}

\begin{example}[Dynamic positive permissions] 
The notion of dynamic positive permission is introduced in \cite{MvT03}: $B$ is dynamically permitted under a condition $A$ if prohibiting $B$ in the context of $A$ would  block the exercise of some explicit (static) permission, thereby generating incoherence. \remove{Consider e.g. the following scenario \cite{Stolpe2010b}:
\smallskip
\begin{quote}
    Freedom of expression [...] is recognized as a human right under Article 19 of the Universal Declaration of Human Rights [...]. An example that comes to mind is the Jyllands-Posten incident of 2005, when Muslim organizations led a complaint with the Danish police, following the publication of twelve cartoons depicting the Islamic prophet Mohammad. The investigation was discontinued by the Regional Prosecutor in Viborg, who concluded that Jyllands-Posten must be reckoned protected by the freedom of expression. [...] One may say, therefore, that the printing of the cartoons was deemed [dynamically] permitted by the Danish authorities.
\end{quote}
\smallskip}
Let the extra-logical rules
\smallskip
\begin{center}
    {\AxiomC{$ $}
    \LeftLabel{\scriptsize{$\delta_{\mathcal{P}_{1}}$}}
    \UnaryInfC{$\Pi\mid\Gamma\specialvdashperm{{\bf S_{1}}}{{\bf T_{1}}}p$}
    \DisplayProof}\quad
    {\AxiomC{$\Pi\mid\Gamma\specialvdashnorm{{\bf S_{2}}}{{\bf T_{2}}}\neg q$}
    \RightLabel{\scriptsize{$\delta_{\mathcal{O}_{1}}$}}
    \UnaryInfC{$\Pi\mid\Gamma\specialvdashnorm{{\bf S'_{2}}}{{\bf T'_{2}}}\neg p$}
    \DisplayProof}
\end{center}
\smallskip
stand for the unconditional permission `One is permitted to express herself freely' and the conditional obligation `If one ought not print cartoons depicting Mohammed, then one ought not express herself freely', respectively -- where
\smallskip
\begin{itemize}
    \item[$(i)$] ${\bf T_{1}}=\{\{\top^{\sf f}\}\}$ and ${\bf S_{1}}=\{\{\neg p^{\sf p}\}\}$;
    \smallskip
    \item[$(ii)$] ${\bf T'_{2}}={\bf T_{2}}\cup\{\{\neg q^{\sf o}\}\}$ and ${\bf S'_{2}}={\bf S_{2}}\cup\{\{\neg\neg p^{\sf o}\}\}$.
\end{itemize}
\smallskip
Suppose one adds to $\Gastc$ the extra-logical rule
\smallskip
\begin{center}
    {\AxiomC{$ $}
    \RightLabel{\scriptsize{$\delta_{\mathcal{O}_{2}}$}}
    \UnaryInfC{$\Pi\mid\Gamma\specialvdashnorm{{\bf U}}{{\bf V}}\neg q$}
    \DisplayProof}
\end{center}
\smallskip
standing for the unconditional obligation `One ought not print the cartoons depicting Mohammed' -- where ${\bf V}=\{\{\top^{\sf f}\}\}$ and ${\bf U}=\{\{\neg\neg q^{\sf o}\}\}$. The derivation of the dynamic permission to print the cartoons depicting Mohammed corresponds to the following application of a suitable inference rule $dp$:
\smallskip
\begin{center}
    {\AxiomC{$ $}
    \LeftLabel{\scriptsize{$\delta_{\mathcal{P}_{1}}$}}
    \UnaryInfC{$\varnothing\mid\ \specialvdashperm{{\bf S_{1}}}{{\bf T_{1}}}p$}
    \AxiomC{$ $}
    \RightLabel{\scriptsize{$\delta_{\mathcal{O}_{2}}$}}
    \UnaryInfC{$\varnothing\mid\ \specialvdashnorm{{\bf U}}{{\bf V}}\neg q$}
    \RightLabel{\scriptsize{$\delta_{\mathcal{O}_{1}}$}}
    \UnaryInfC{$\varnothing\mid\ \specialvdashnorm{{\bf U'}}{{\bf V'}}\neg p$}
    \RightLabel{\scriptsize{$dp_{\mathcal{W}^{\sf p\star}}$}}
    \BinaryInfC{$\varnothing\mid\ \specialvdashperm{{\bf S}}{{\bf T}}q$}
    \DisplayProof}
\end{center}
\smallskip
with ${\bf V'}={\bf V}\cup\{\{\neg q^{\sf o}\}\}$, ${\bf U'}={\bf U}\cup\{\{\neg\neg p^{\sf o}\}\}$, ${\bf T}={\bf T_{1}}\cup{\bf V}$, ${\bf S}={\bf S_{1}}$ and $\mathcal{W}^{\sf p\star}=\Big\{\dfrac{\top:\top}{q}\Big\}$. Notice that neither ${\bf T}$ is ${\bf T_{1}}\cup{\bf V'}$, nor ${\bf S}$ is ${\bf S_{1}}\cup{\bf U'}$. In other words, the $dp$ application causes the removal of sets of formulas from ${\bf T_{1}}\cup{\bf V'}$ and ${\bf S_{1}}\cup{\bf U'}$.
\end{example}

\begin{example}[Talmudic {\em Qal wa-\d homer}]
Let the extra-logical rules
\smallskip
\begin{center}
    {\AxiomC{$ $}
    \LeftLabel{\scriptsize{$\delta_{\mathcal{O}_{1}}$}}
    \UnaryInfC{$\Pi\mid\ \specialvdashnorm{{\bf S_{1}}}{{\bf T_{1}}}\neg p$}
    \DisplayProof}\quad
    {\AxiomC{$ $}
    \RightLabel{\scriptsize{$\delta_{\mathcal{P}_{1}}$}}
    \UnaryInfC{$\Pi\mid\ \specialvdashperm{{\bf S_{2}}}{{\bf T_{2}}}q$}
    \DisplayProof}\quad
    {\AxiomC{$\Pi\mid\Gamma\specialvdashperm{{\bf S_{3}}}{{\bf T_{3}}}q$}
    \RightLabel{\scriptsize{$\delta_{\mathcal{P}_{2}}$}}
    \UnaryInfC{$\Pi\mid\Gamma\specialvdashperm{{\bf S'_{3}}}{{\bf T'_{3}}}p$}
    \DisplayProof}
\end{center}
\smallskip
stand for the unconditional obligation `One ought not marry the daughter of his daughter', the unconditional permission `One is allowed to marry one's daughter' and the conditional permission `If one is allowed to marry one's daughter, one is allowed to marry the daughter of his daughter', respectively -- where
    \smallskip
    \begin{itemize}
    \item[$(i)$] ${\bf T_{1}}=\{\{\top^{\sf o}\}\}$ and ${\bf S_{1}}=\{\{\neg\neg p^{\sf o}\}\}$;
    \smallskip
    \item[$(ii)$] ${\bf T_{2}}=\{\{\top^{\sf p}\}\}$ and ${\bf S_{2}}=\{\{\neg q^{\sf p}\}\}$;
    \smallskip
    \item[$(iii)$] ${\bf T'_{3}}={\bf T_{3}}\cup\{\{q^{\sf p}\}\}$ and ${\bf S'_{3}}={\bf S_{3}}\cup\{\{\neg p^{\sf p}\}\}$.
    \end{itemize}
    \smallskip
The Talmudic principle of {\em Qal wa-\d homer} licenses the inference to the unconditional obligation `One ought not marry his daughter' \cite[p. 223]{Steinsaltz}. This corresponds to the following application of a suitable inference rule $qw$:
\smallskip
\begin{center}
    {\AxiomC{$ $}
    \LeftLabel{\scriptsize{$\delta_{\mathcal{O}_{1}}$}}
    \UnaryInfC{$\varnothing\mid\ \specialvdashnorm{{\bf S_{1}}}{{\bf T_{1}}}\neg p$}
    \AxiomC{$ $}
    \RightLabel{\scriptsize{$\delta_{\mathcal{P}_{1}}$}}
    \UnaryInfC{$\varnothing\mid\ \specialvdashperm{{\bf S_{2}}}{{\bf T_{2}}}q$}
    \RightLabel{\scriptsize{$\delta_{\mathcal{P}_{2}}$}}
    \UnaryInfC{$\varnothing\mid\ \specialvdashperm{{\bf S'_{2}}}{{\bf T'_{2}}}p$}
    \RightLabel{\scriptsize{$qw_{\mathcal{W}^{\sf o\star}}$}}
    \BinaryInfC{$\varnothing\mid\ \specialvdashnorm{{\bf S}}{{\bf T}}\neg q$}
    \DisplayProof}
\end{center}
\smallskip
with ${\bf T'_{2}}={\bf T_{2}}\cup\{\{q^{\sf p}\}\}$,  ${\bf S'_{2}}={\bf S_{2}}\cup\{\{\neg p^{\sf p}\}\}$, ${\bf T}={\bf T_{1}}$, ${\bf S}=\{\{\neg\neg q^{\sf o}\}\}$ and $\mathcal{W}^{\sf o\star}=\Big\{\dfrac{\top:\top}{\neg q}\Big\}$.
\end{example}

Control pairs can be further refined  to address problematic scenarios involving the unrestricted weakening of obligations. A symmetric control pair is a structure of the form $$\langle{\bf T},\langle{\bf S},{\bf S}'\rangle\rangle$$ where ${\bf T}$ collects sets of conditions, ${\bf S}$ collects sets of constraints for formulas occurring in antecedent position, and ${\bf S}'$ collects sets of constraints for formulas occurring in succedent position. We can now adapt the notions of compatibility and soundness to this symmetric setting as follows.

\begin{definition}\label{dualcompatibility}
\noindent Let $\langle{\bf T},{\frak S}\rangle$ be a symmetric control pair, with ${\frak S}=\langle{\bf S},{\bf S'}\rangle$.
\smallskip
\begin{itemize}
    \item[$(i)$]  $\Delta$ is {\em compatible} with $\langle{\bf T},{\frak S}\rangle$ exactly when $\fullGpn$ refutes $A\vdash\bigvee\Delta$ for every $A\in\bigcup{\bf S'}$.
    \smallskip
    \item[$(ii)$] $\Pi\mid\Gamma\specialvdashnorm{{\frak S}}{{\bf T}}\Delta$ is {\em sound} if and only if $\Pi\mid\Gamma\specialvdashnorm{{\bf S}}{{\bf T}}\Delta$ is sound and $\Delta$ is compatible with ${\bf S'}$.
    \smallskip
    \item[$(iii)$] A $\Gastc$-derivation $\pi$ of $\Pi\mid\Gamma\specialvdashnorm{{\frak S}}{{\bf T}}\Delta$ is a {\em proof} if and only if the conclusion is sound and $\pi$ is a proof when removing ${\bf S'}$ and its ancestors along $\pi$.
\end{itemize}
\end{definition}

We end this section with the following applications of symmetric control pairs.

\begin{example}[Ross's paradox]
    Let the extra-logical axiom and rules
    \smallskip
    \begin{center}
        {\AxiomC{$ $}
        \LeftLabel{\scriptsize{$\delta_{\mathcal{O}_{1}}$}}
        \UnaryInfC{$\Pi\mid\Gamma\specialvdashnorm{{\bf S_{1}}}{{\bf T_{1}}}p$}
        \DisplayProof}\quad
        {\AxiomC{$ $}
        \LeftLabel{\scriptsize{$\delta_{\mathcal{O}_{1}}$}}
        \UnaryInfC{$\Pi\mid\Gamma\specialvdashnorm{{\bf S_{2}}}{{\bf T_{2}}}\neg q$}
        \DisplayProof}
        {\AxiomC{$ $}
        \RightLabel{\scriptsize{$ax$}}
        \UnaryInfC{$r\specialvdashdef{\varnothing}{\varnothing}q$}
        \DisplayProof}
    \end{center}
    \smallskip
    stand for the unconditional obligations `One ought to mail the letter', `One ought not destroy the letter' and the factual statement `If one burns the letter, she destroys it', respectively -- where ${\bf T_{1}}={\bf T_{2}}=\{\{\top^{\sf o}\}\}$, ${\bf S_{1}}=\{\{\neg p^{\sf o}\}\}$ and ${\bf S_{2}}=\{\{\neg\neg q^{\sf o}\}\}$. 
    
    `One mails the letter or one burns the letter' can be inferred from `One mails the letter' in classical logic, closure of obligations under classical consequence yields the unconditional obligation  `One ought to post the letter or burn it'. This means that the obligation to mail the letter would entail an obligation that could be satisfied by burning it -- an outcome that appears undesirable \cite{Ross41}.

    To avoid this scenario and filter out the undesired instance of Right Weakening, we rewrite $\delta_{\mathcal{O}_{1}}$ replacing ${\bf S_{1}}$ with $\mathfrak{S}_{1}=\langle{\bf S_{1}},\{\{p\vee q^{\sf o}\}\}\rangle$. As a result, the following derivation is a paraproof:
    \smallskip
    \begin{center}
        {\AxiomC{$ $}
        \RightLabel{\scriptsize{$\delta_{\mathcal{O}_{1}}$}}
        \UnaryInfC{$\varnothing\mid\ \specialvdashnorm{{\mathfrak{S}_{1}}}{{\bf T_{1}}}p$}
        \RightLabel{\scriptsize{$LW^{\sf O}$}}
        \UnaryInfC{$\varnothing\mid\ \specialvdashnorm{{\mathfrak{S}_{1}}}{{\bf T_{1}}}p,q$}
        \DisplayProof}
    \end{center}
\end{example}

\begin{example}[Good Samaritan paradox]
Let the following extra-logical rules
\smallskip
\begin{center}
    {\AxiomC{$ $}
    \LeftLabel{\scriptsize{$\delta_{\mathcal{O}_{1}}'$}}
    \UnaryInfC{$\Pi\mid\Gamma\specialvdashnorm{{\frak S_{1}}}{{\bf T_{1}}}p $}
    \DisplayProof}\quad
    {\AxiomC{$ $}
    \RightLabel{\scriptsize{$\delta_{\mathcal{O}_{1}}''$}}
    \UnaryInfC{$\Pi\mid\Gamma\specialvdashnorm{{\frak S_{1}}}{{\bf T_{1}}}q $}
    \DisplayProof}
\end{center}
\smallskip
stand for the unconditional obligation `It ought to be that the Samaritan helps Jones whom
Smith has robbed', and the extra-logical rule
\smallskip
\begin{center}
    {\AxiomC{$ $}
    \RightLabel{\scriptsize{$\delta_{\mathcal{O}_{2}}$}}
    \UnaryInfC{$\Pi\mid\Gamma\specialvdashnorm{{\frak S_{2}}}{{\bf T_{2}}}\neg q $}
    \DisplayProof}
\end{center}
\smallskip
stand for the unconditional obligation `Smith ought not rob Jones' \cite{Aqvist67} -- where ${\bf T_{1}}={\bf T_{2}}=\{\{\top^{\sf f}\}\}$, ${\frak S_{1}}=\langle\{\{\neg(p\wedge q)^{\sf o}\}\},\{\{q^{\sf o}\}\}\rangle$ and ${\frak S_{2}}=\langle\{\{\neg\neg q^{\sf o}\}\},\{\{q^{\sf o}\}\}\rangle$. Consider the following derivation:
\smallskip
\begin{center}
    {\AxiomC{$ $}
    \LeftLabel{\scriptsize{$\delta_{\mathcal{O}_{1}}'$}}
    \UnaryInfC{$\varnothing\mid\ \specialvdashnorm{{\frak S_{1}}}{{\bf T_{1}}}p$}
    \AxiomC{$ $}
    \LeftLabel{\scriptsize{$\delta_{\mathcal{O}_{1}}''$}}
    \UnaryInfC{$\varnothing\mid\ \specialvdashnorm{{\frak S_{1}}}{{\bf T_{1}}}q$}
    \BinaryInfC{$\varnothing\mid\ \specialvdashnorm{{\frak S_{1}}}{{\bf T_{1}}}p\wedge q$}
    \DisplayProof}
\end{center}
\smallskip
Such derivation is a proof, while the immediate leftmost subderivation is a paraproof. This reflects the fact that the obligation to help victims does not entail that the existence of victims is itself obligatory.
\remove{where ${\bf T}={\bf T_{1}}\cup{\bf T_{2}}$ and ${\frak S}=\langle{\bf S_{1}}\cup{\bf S_{2}},\{\{q^{\sf o}\}\}\rangle$. Such a derivation is a paraproof. This reflects the fact that the obligation to help victims does not entail that the existence of victims is itself obligatory. Notice that this use of $rrw$ matches the intuitive assessment of the paradox even when consistency constraints are omitted -- i.e., when ${\bf S_{1}} = {\bf S_{2}} = \varnothing$.}
\end{example}

We conclude this section by observing that symmetric control pairs can be used to enforce the metaethical principle of `no vacuous
obligations', according to which tautologies are never obligatory and contradictions are never prohibited \cite{vW51}. Indeed, we could design $\Gastc$ calculi without ${\sf O}$-labelled identity axioms, and in which every ${\sf O}$-labelled sequent includes, among its constraints, the singleton containing an arbitrary tautology for formulas in succedent position. Such $\Gastc$ calculi would thereby capture deontic reasoning without allowing unconstrained Right Weakening \cite{stolpephd,PvdT14}.

\section{Concluding remarks}\label{conclusion}

In this paper, we introduced a uniform proof-theoretic framework for defaults, obligations, and permissions based on controlled sequent calculi. We defined deontic extensions, paralleling \L ukasiewicz extensions for defaults and generalizing outfamilies in I/O logics. Controlled calculi were shown to admit contraction and non-analytic cut, to exhibit a weak form of analyticity, and to be strongly complete with respect to credulous consequence grounded on \L ukasiewicz and deontic extensions. Moreover, they support weak versions of cumulative transitivity and cautious monotony, thus providing a unified and robust basis for deontic and nonmonotonic reasoning. 

Several perspectives open up. First, control sets can be used to embed specific instances of paraconsistent reasoning into the classical base logic \cite{JLC17}. Furthermore, natural extensions of the framework seem capable of capturing default deontic reasoning on non-classical bases, particularly ${\bf FDE}$-based logics \cite{anticut}.

It would be interesting to design $\Gastc$ calculi in which extra-logical rules transmit arbitrary auxiliary formulas from the succedents of the premises to the succedent of the conclusion, thereby providing a uniform and modular treatment of {\em disjunctive} default and deontic reasoning \cite{IOhypersequent}. With appropriate adjustments to the labelling of turnstiles and to the notion of soundness, it should also be possible to obtain $\Gastc$ calculi for constrained I/O logics. Unlike the hybrid hypersequent calculi developed for I/O logics in \cite{IOhypersequent}, these $\Gastc$ calculi would enjoy strong adequacy.

Moreover, with respect to prohibitions understood as obligations to the contrary, certain contexts call for a more refined treatment, for instance, within Talmudic deontic logic \cite{Talmudicdeontic}. It would therefore be of particular interest to integrate explicit prohibitions into the very definition of normative system, introducing a dedicated labelled turnstile $\specialvdashprohib{}{}$ to develop strongly complete $\Gastc$ calculi. Additional extra-logical rules involving $\specialvdashprohib{}{}$ could then capture the intricate interactions between prohibitions and other deontic notions.

Additionally, it would be worthwhile to examine whether $\Gastc$ calculi can be employed to address prioritized default and normative reasoning \cite{Brewka95, Horty07, vdPS13}. Since prioritized $m$- and $d$-extensions lack the key property of semimonotonicity, defining a sound notion of $\Gastc$-proof capable of tracking prioritized credulous consequence may require a suitable notion of {\em hypothetical soundness} -- i.e., soundness under the assumption that some formulas are not provable.

\bibliographystyle{plain}
\bibliography{default}

\end{document}